\documentclass[prodmode]{acmsmall} % Aptara syntax
\usepackage{tikz}
\usepackage{grffile}
\usepackage[bookmarks]{hyperref}
\usepackage{dsfont}
\usepackage[ruled,linesnumbered]{algorithm2e}

\SetAlFnt{\small}
\SetAlCapFnt{\small}
\SetAlCapNameFnt{\small}
\SetAlCapHSkip{0pt}
\IncMargin{-\parindent}

\usepackage{amsmath,amssymb}
\DeclareGraphicsExtensions{.pdf,.png,.jpg,.jpeg,.mps}
\graphicspath{{figs/}{./}}
\usepackage{subfig}
\usepackage{placeins}

\newtheorem{thm}{Theorem}[section]  
\usepackage{appendix}
  
\usepackage{booktabs}

\newcommand{\cL}{\mathcal{L}}
\newcommand{\cS}{\mathcal{S}}
\newcommand{\MSE}{\mbox{MSE}}

\newcommand{\NRMSE}{\mbox{\textnormal{NRMSE}} }

\newcommand{\vol}{\mbox{\textnormal{vol}}}

\newcommand{\comment}[1]{}
\newcommand{\techreport}[2]{#2}

\newcommand{\fm}[1]{#1}

\newcommand{\thetitle}{Characterizing Directed and Undirected Networks via Multidimensional Walks with Jumps}

\newlength{\plotwidth}
\newlength{\plotheight}
\newcommand{\btheta}{{\boldsymbol \theta}}
\newcommand{\bbeta}{{\boldsymbol \beta}}

\newcommand{\bbS}{\mathbb{S}}

\begin{document}

% Page heads
\markboth{F. Murai et al.}{\thetitle}

% Title portion
\title{\thetitle}
\author{FABRICIO MURAI
\affil{Universidade Federal de Minas Gerais}
BRUNO RIBEIRO
\affil{Purdue University}
DON TOWSLEY
\affil{University of Massachusetts Amherst}
PINGHUI WANG
\affil{Xi'an Jiaotong University}}
% NOTE! Affiliations placed here should be for the institution where the
%       BULK of the research was done. If the author has gone to a new
%       institution, before publication, the (above) affiliation should NOT be changed.
%       The authors 'current' address may be given in the "Author's addresses:" block (below).
%       So for example, Mr. Abdelzaher, the bulk of the research was done at UIUC, and he is
%       currently affiliated with NASA.

\begin{abstract}
Estimating distributions of node characteristics (labels) such as number of connections or citizenship of users in a social network via edge and node sampling is a vital part of the study of complex networks. Due to its low cost, sampling via a random walk (RW) has been proposed as an attractive solution to this task. Most RW methods assume either that the network is undirected or that walkers can traverse edges regardless of their direction. Some RW methods have been designed for directed networks where edges coming into a node are not directly observable. In this work, we propose Directed Unbiased Frontier Sampling (DUFS), a sampling method based on a large number of coordinated walkers, each starting from a node chosen uniformly at random. It is applicable to directed networks  with invisible incoming edges because it constructs, in real-time, an undirected graph consistent with the walkers trajectories, and due to the use of random jumps which prevent walkers from being trapped.  DUFS generalizes previous RW methods and is suited for undirected networks and to directed networks regardless of in-edges visibility. We also propose an improved estimator of node label distributions that combines information from the initial walker locations with subsequent RW observations. We evaluate DUFS, compare it to other RW methods, investigate the impact of its parameters on estimation accuracy and provide practical guidelines for choosing them. In estimating out-degree distributions, DUFS yields significantly better estimates of the head of the distribution than other methods, while matching or exceeding estimation accuracy of the tail.  Last, we show that DUFS outperforms uniform node sampling when estimating distributions of node labels of the top 10\% largest degree nodes, even when sampling a node uniformly has the same cost as RW steps.

\end{abstract}




\acmformat{Fabricio Murai, Bruno Ribeiro, Don Towsley and Pinghui Wang, 2017. \thetitle.}
% At a minimum you need to supply the author names, year and a title.
% IMPORTANT:
% Full first names whenever they are known, surname last, followed by a period.
% In the case of two authors, 'and' is placed between them.
% In the case of three or more authors, the serial comma is used, that is, all author names
% except the last one but including the penultimate author's name are followed by a comma,
% and then 'and' is placed before the final author's name.
% If only first and middle initials are known, then each initial
% is followed by a period and they are separated by a space.
% The remaining information (journal title, volume, article number, date, etc.) is 'auto-generated'.

\begin{bottomstuff}
This work is supported by the Army Research Laboratory
under Cooperative Agreement Number W911NF-09-2-0053
and the CNPq,
National Council for Scientific and Technological Development -
Brazil.

Author's addresses: F. Murai, Universidade Federal de Minas Gerais; B. Ribeiro, Purdue University; D. Towsley, College of Information and Computer Sciences,
University of Massachusetts Amherst; Pinghui Wang,
Xi'an Jiaotong University.
\end{bottomstuff}

\maketitle

% Activate the following line by filling in the right side. If for example the name of the root file is Main.tex, write
% "...root = Main.tex" if the chapter file is in the same directory, and "...root = ../Main.tex" if the chapter is in a subdirectory.
 
%!TEX root = ../frontier-tkdd.tex  

\section{Introduction}

% why sample networks?
A number of studies
\cite{SurveyMeasuresGraphs,MobileFriends,Leskovec2006,LeskovecCommunity,Mislove,MySpace,WillingerRDS,RDSprob,Gjoka2010,KurantJSAC2011,ChierichettiWWW2016}
are dedicated to the characterization of complex networks. Examples of networks of interest include the Internet,
the Web, social, business, and biological networks. Characterizing a network
consists of computing or estimating a set of statistics that describe the
network. In this work we model a network as a directed or
undirected graph with labeled vertices. A label can be, for instance, the
degree of a node or, in a social network setting, someone's hometown. Label
statistics (e.g., average, distribution) are often used to characterize a
network.

Characterizing a network with respect to its labels requires querying vertices
and/or edges; associated with each query is a resource cost (time, bandwidth,
money).  For example, information about web pages must be obtained by querying
web servers subject to a maximum query rate.  Characterizing a large
network by querying the entire network is often too costly. Even if the network is
stored on disk it may constitute several terabytes of data.  As a result,
researchers have turned their attention to the
characterization of networks based on incomplete (sampled) data.

Simple strategies such as uniform node and uniform edge sampling possess
desirable statistical properties: the former yields unbiased samples of the
population and the bias introduced by the latter is easily removed.  However,
these strategies are often rendered unfeasible because they require either a directory
containing the list of all node (edge) ids, or an API that allows uniform
sampling from the node (edge) space.  Even when the space of possible node
(edge) ids is known, its occupancy is usually so low that querying randomly
generated ids is expensive.  An alternate, cheaper, way to sample a network is
via a random walk (RW). A RW samples a network by moving a particle (walker)
from a node to a neighboring node. It is applicable to any network where one
can query the edges connected to a given node. Furthermore, RWs share some of
the desirable properties of uniform edge sampling (i.e., easy bias removal,
accurate estimation of characteristics such as the tail of the degree
distribution).% after a transient phase known as mixing time.

On one hand, a great deal of research has focused on the design of sampling methods
for {\em undirected networks} using RWs \cite{Heckathorn97,WillingerRDS}.
Ribeiro and Towsley proposed Frontier Sampling (FS), a multidimensional random
walk that uses $n$ {\em coupled} random walkers. This method
yields more accurate estimates than the standard RW and also outperforms the use
of $n$ independent walkers. In the presence of disconnected or loosely
connected components, FS is even better suited than the standard RW and
independent RWs to sample the tail of the degree distribution of the graph. On
the other hand, few works have focused on the development of tools for characterizing
{\em directed networks} in the wild. A network is said to be directed when edges
are not necessarily reciprocated. Characterizing directed networks through
crawling becomes especially challenging when only outgoing edges from a node are visible
(incoming edges are hidden): unless all vertices have a directed path to all
other vertices, a walker will eventually be restricted to a (strongly
connected) component of the graph. Furthermore, a standard RW incurs a bias that
can only be removed by conditioning on the entire graph structure. 
\cite{RibeiroINFOCOM2012} addressed these issues by proposing Directed
Unbiased Random Walk (DURW), a sampling technique that builds a virtual
undirected graph on-the-fly and performs
degree-proportional jumps to obtain asymptotically unbiased estimates of the
distribution of node labels on a directed graph.

%Originally, random walk-based techniques have been proposed only for undirected
%networks, where edges indicate reciprocal relationships (e.g., friendship on
%Facebook, interaction on protein networks). Yet, the same techniques can also
%be applied to directed networks if both incoming and outgoing links can be
%observed from a node (e.g., Flickr, Twitter). Conversely, if only outgoing
%links are observed (e.g., Google+, the Web), additional subtleties such as the
%existence of sink nodes (or more generally, multiple strongly connected components)
%have to be addressed in a way that allows us to remove any incurred sampling bias.
%Our second goal is to design a method that can handle both undirected networks
%and directed networks in a unified way.
  \begin{figure}[]
  \centering
     \includegraphics[height=0.35\textwidth]{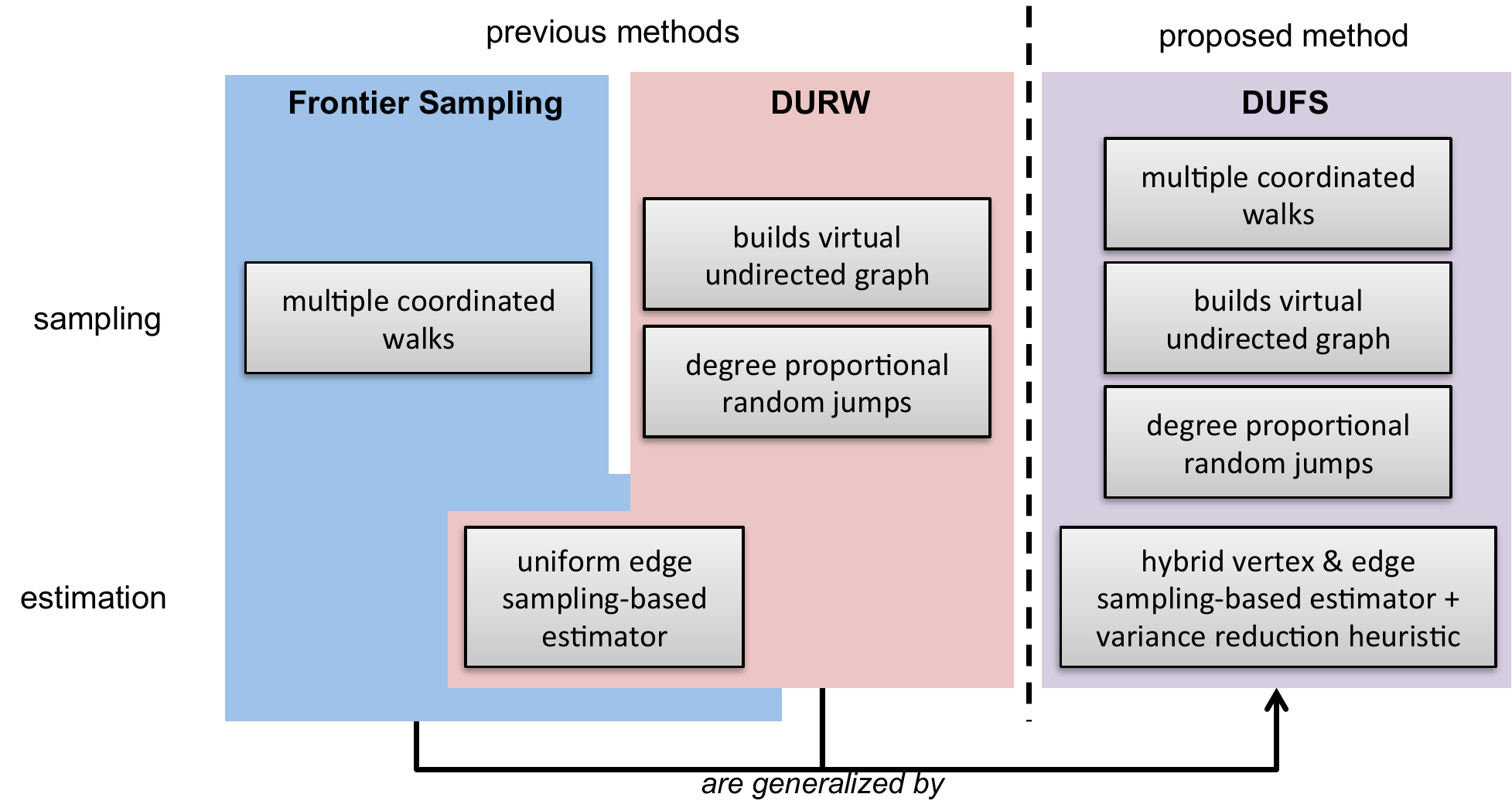}
%    \subfloat[email-EuAll\label{fig:email-EuAll-DUFS}]{%
%      \includegraphics[width=0.5\textwidth]{email-EuAll-DUFS.pdf}
%      } \subfloat[p2p-Gnutella31\label{fig:p2p-Gnutella31-DUFS}]{%
%      \includegraphics[width=0.5\textwidth]{p2p-Gnutella31-DUFS.pdf}
%    }
     \caption{Proposed method (DUFS) generalizes Frontier Sampling and DURW.}
    \label{fig:sketch}
 \end{figure}

%this work
In this work\footnote{Parts of this work are based on previous papers from the authors:
\cite{TechReport} and~\cite{RibeiroINFOCOM2012}.}, we propose
Directed Unbiased Frontier Sampling (DUFS), a method that generalizes
the FS and the DURW algorithms (see Figure~\ref{fig:sketch}).
%make another step towards developing tools for characterizing directed graphs.
Building on ideas in \cite{RibeiroINFOCOM2012}, we extend FS
to allow the characterization of networks regardless of whether they
are undirected, directed with observable incoming edges, or directed with
unobservable incoming edges. From another perspective, we adapt DURW to use multiple coordinated walkers.
DUFS matches or exceeds the accuracy of
FS and DURW\footnote{The software and all results presented in this work are available
at \url{http://bitbucket.com/after-acceptance}. %All the data used in this work is publicly available from different sources.
}, as illustrated in Figure~\ref{fig:showcase}. Method parameters ($w$ and $b$), simulation setup, datasets and the error metric -- NRMSE (normalized root mean square error) -- are described in Section~\ref{sec:parameters}.

\begin{figure}[]
\centering
    \subfloat[soc-Slashdot0902\label{fig:soc-Slashdot0902-mlefs-jnt}]{%
      \includegraphics[width=0.45\textwidth]{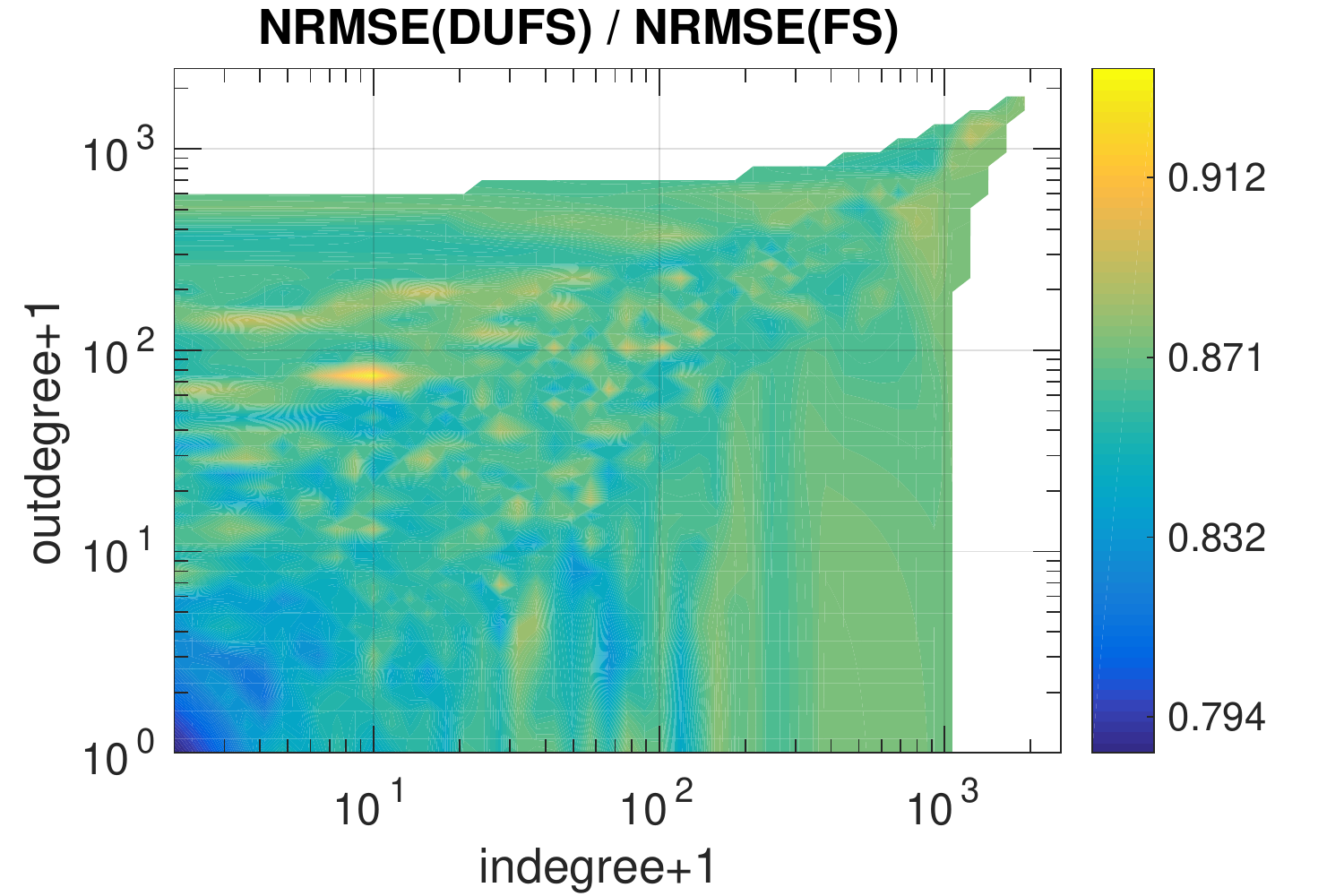}
    }
    \subfloat[livejournal-links\label{fig:livejournal-links-v0}]{%
      \includegraphics[trim={0 0.4cm 0 0},clip,width=0.42\textwidth]{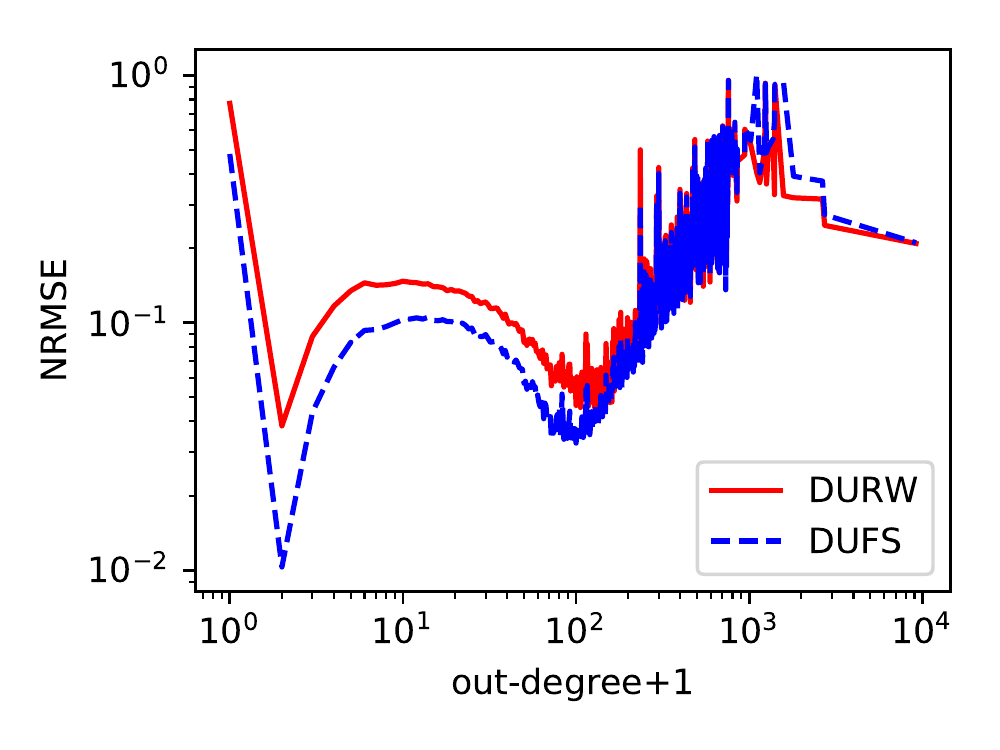}
    }%\\
%    \subfloat[web-Google\label{fig:web-Google-RVT}]{%
%      \includegraphics[width=0.49\textwidth]{web-Google-RVT.pdf}
%    }
%     \subfloat[wiki-Vote\label{fig:wiki-Vote}]{%
%      \includegraphics[width=0.49\textwidth]{wiki-Vote-RVT.pdf}
%    }
    \caption{Comparison between proposed method (DUFS) and previous state-of-the-art
    respectively for visible and for invisible incoming edges scenarios; (a) NRMSE ratios between
    DUFS $(w=0.1,b=10)$ and FS $(b=10)$ of the estimated joint in- and out-degree distribution
    on the soc-Slashdot0902 dataset; (b) NRMSEs associated with DUFS and DURW of the estimated
    out-degree distribution on the livejournal-links dataset. 
    }
    \label{fig:showcase}
 \end{figure}

\paragraph{Contributions} Our main contributions are as follows:
\begin{enumerate}

  \item {\em Directed Unbiased Frontier Sampling (DUFS)}: we propose a
  new algorithm based on multiple coordinated random walks that extends Frontier Sampling (FS) to directed networks.
  DUFS extends DURW to multiple random walks. %\cite{TechReport} and  DURW \cite{RibeiroINFOCOM2012}.

    \item {\em A more accurate estimator for node label distribution}: when the number
    of walkers is a large fraction of the number of random walk steps (e.g., 10\%),
    a considerable amount of information is thrown out by
    not accounting for the walkers initial locations as observations. We introduce a new estimator
    that combines these observations with those made during the walks to produce
    better estimates.

   \item {\em Practical recommendations}: we investigate the impact of
   the number of walkers and the probability of jumping to an uniformly chosen node (controlled via a parameter called
   random jump weight) on DUFS estimation error, given a fixed budget. By increasing the number of walkers the sequence
   of sampled edges approaches the uniform distribution faster,
   but this also increases the fraction of the budget spent to place the walkers in their initial locations.
   Moreover, increasing the random jump weight favors sampling node labels with large probability masses,
   which translates into more accurate estimates for these labels, but worse estimates for those in the tail.   
   We study these trade-offs through simulation and propose guidelines for choosing DUFS parameters.
   %We conduct a similar study to choose DUFS parameters in Section~\ref{sec:parameters}.

  \item {\em Comprehensive evaluation}: we compare DUFS
  to other random walk-based methods  applied to directed networks w.r.t.\ estimation
  errors, both
  when
  incoming edges are directly observable and when they are not.
In the first scenario, in addition to some graph properties evaluated in
    previous works, we evaluate DUFS performance on estimating joint
    in- and out-degree distributions, and on estimating distribution of
    group memberships among the 10\% largest degree nodes.
    
    \item {\em Theoretical analysis:} we derive expressions for the
    normalized mean squared error associated with uniform node
    and uniform edge sampling on power law networks and show
    that in both cases error behaves asymptotically as a power law function of
    the observed degree. This helps explain our evaluation results.

\end{enumerate}

\paragraph{Outline} Definitions are presented in Section~\ref{sec:definitions}.
In Section~\ref{sec:background}, we review FS and DURW methods.
 In Section~\ref{sec:dufs}, we propose Directed Unbiased Frontier Sampling (DUFS)
 (along with some estimators), which generalizes previous methods. We investigate the impact of DUFS parameters on
 estimation accuracy of degree distributions and node label distributions respectively in Sections~\ref{sec:results} and \ref{sec:attributes}, providing practical guidelines on how to set them. A comparison to other random walk techniques is also provided. Section~\ref{sec:discussion} discusses 
%the relationship between  NRMSE and out-degree distribution and also 
the performance of DUFS when the uniform node sampling mechanism is faulty. We present
 some related work and present our conclusions in Sections~\ref{sec:related}~and~\ref{sec:conclusions}, respectively.

% Activate the following line by filling in the right side. If for example the name of the root file is Main.tex, write
% "...root = Main.tex" if the chapter file is in the same directory, and "...root = ../Main.tex" if the chapter is in a subdirectory.
 
%!TEX root = ../frontier-tkdd.tex

\section{Terminology Setting}\label{sec:definitions}
                
In what follows we present terminology used throughout the paper. We also present two scenarios considered in our work.
 Let $G_d = (V,E_d)$ be a labeled directed graph representing the network graph, where $V$ is a set of vertices and $E_d$ is a set of ordered pairs of vertices $(u,v)$ representing a connection from $u$ to $v$ (a.k.a.\ edges). We refer to an edge $(u,v)$ as an {\em in-edge} with
 respect to $v$ and an {\em out-edge} with respect to $u$. 
 The {\em in-degree} and {\em out-degree} of a node $u$ in $G_d$ are the number of distinct edges
 respectively into and out of $u$. We assume that each node in $G_d$ has at least one edge (either an in-edge or an out-edge).
 Some networks can be modeled as undirected graphs.
 In this case, $G_d$ is a symmetric directed graph, i.e., $(u,v) \in E_d$ iff $(v,u) \in E_d$. 

 Let $\cL_v$ and $\cL_e$ be finite (possibly empty) sets of node labels and edge labels, respectively.
 Each edge $(u,v) \in E_d$ is associated with a set of labels $\cL_e(u,v) \subseteq \cL_e$.
 For instance, one label $\ell \in \cL_e(u,v)$ could be the nature of the relationship between
 two individuals (e.g., family, work, school) in a social network represented by nodes $u$ and $v$.
 Similarly, we can associate a set of labels to each node, $\cL_v(v) \subseteq \cL_v,\, \forall v \in V$.

\subsection*{Input scenarios}
When performing a random walk, we assume that a walker retrieves the out-edges of node where it resides by performing a query (e.g., followers list on Twitter)  
and that vertices are distinguishable. We define two scenarios depending on whether the walker can also retrieve in-edges.

In the {\em first scenario}, both out- and in-edges can be retrieved and it is possible to move
the walker over any edge regardless of the edge direction (if the edge is $(u,v) \in E_d$ a walker can move from $u$ to $v$ and vice versa).
In this case, the walker can be seen as
moving over $G=(V,E)$, an undirected version of  $G_d$, i.e., $E=\{(u,v) : (u,v) \in E_d \lor (v,u) \in E_d\}$.
Define $\deg(v) = | \{(u,v) : (u,v) \in E\} |$. Let $\vol(S) = \sum_{\forall v \in S} \deg(v), \, \forall S \subseteq V$, denote the volume of the
set of vertices in $S \subseteq V$.

%we are able to build (on-the-fly) a symmetric directed graph while walking over $G_d$.
% Let $G=(V,E)$ be the symmetric counterpart of $G_d$, i.e., $$E = \bigcup_{\forall (u,v) \in E_d} \{(u,v),(v,u)\}.$$
% For practical purposes, the random walk can be seen as moving over $G$.
%In particular, if $G_d$ is undirected, then $G=G_d$ once all edges have been sampled.

In the {\em second scenario}, only out-edges are directly observable and we can build on-the-fly an undirected graph $G_u$ based on
the out-edges that have been sampled. Note that $G_u$ is not an undirected version of $G_d$ as some of the
in-edges of a node may not have been observed. By moving the
walker over $G_u$ -- possibly traversing edges in $G_d$ in
the opposite direction -- we can compute its stationary behavior
and thus, remove any bias by accounting for the probability that each observation appears in the sample.

% useful to keep track
While this has been mostly overlooked by other works, we emphasize that,
in either scenario, it is useful to keep track of some variant of the observed graph during
the sampling process. Storing information about visited nodes in memory
saves resources that would be consumed to query those nodes in subsequent visits -- i.e., revisiting a node has no cost. The specific
variant of the observed graph to be stored will be described in the context of two random
walk-based methods in the following section.

% first scenario

%In the second class, when {\em incoming edges} are hidden, we work with $G_d$ directly.
%\fm{Think about how to describe the second class}.

%For the sake of simplicity, in the remainder of this paper we assume that all queries of edges and vertices have unitary cost and that we have a fixed sampling budget $B$ unless stated otherwise.

% Activate the following line by filling in the right side. If for example the name of the root file is Main.tex, write
% "...root = Main.tex" if the chapter file is in the same directory, and "...root = ../Main.tex" if the chapter is in a subdirectory.
 
%!TEX root = ../frontier-tkdd.tex 

\section{Background}\label{sec:background}

The method proposed in this paper generalizes two representative random-walk based methods
designed for each of the respective scenarios described in Section~\ref{sec:definitions}.
Therefore, we dedicate this section to briefly reviewing these methods.
First,
we describe the Frontier Sampling algorithm proposed in \cite{TechReport},
an $n$-dimensional random walk that benefits from starting its walkers at uniformly sampled vertices.
This technique can be applied to undirected graphs and to directed graphs provided that edges coming into a node are observable.
Then, we describe the Directed Unbiased Random Walk algorithm proposed in \cite{RibeiroINFOCOM2012},
that adapts a single random walk to a directed graph when incoming edges are not directly observable.
The goal of these methods is to obtain samples from a graph, which are then used to infer graph characteristics via an estimator. 
An {\em estimator} is a function that takes a sequence of observations (sampled data) as input and outputs an estimate of an unknown population parameter (graph characteristic).

%Then we propose the Directed Unbiased Frontier Sampling
%algorithm, which generalizes
%both methods and applies indistinctly to either class.

%Last, we present results on the accuracy of those estimators.

%In the second part, we propose DUFS ({\em Directed Unbiased Frontier Sampling}) -- a variation of the original technique that works even for directed graphs where in-edges are not observable.

\subsection{Frontier Sampling: a multidimensional random walk for undirected networks}

In essence, {\em Frontier Sampling} (FS) is a random walk-based algorithm for sampling and estimating characteristics of an
undirected graph.
%It can also be applied to directed networks with observable in-edges by ignoring edges directions when moving a walker.
FS performs $n$ {\em coordinated} random walks on the graph.
One of the advantages of using multiple walkers is that they
can cover multiple connected components (when they exist), while a single walker is restricted to one component in the absence
of a random jump or restart mechanism.
%\fm{REMOVE: However, when random walks are independent (not coordinated) the number of samples obtained from a component
%is proportional to the number of walkers in that component. Therefore, the probability of sampling an
%edge in steady state will differ for different components, unless the number of walkers in each
%component is set to be proportional to its volume. Unfortunately, initializing the walkers in such a way requires
%knowing the component volumes in advance, which cannot be done in practice.}
By coordinating multiple random walkers, FS is able to sample edges uniformly at random in steady state
regardless of how the walkers are initially placed.

\begin{algorithm}[t]
\SetKwInOut{Input}{Input}
\Input{sampling budget $B$, budget per walker $b$, cost of uniform node sampling $c$}
$n \gets  B/(c+b) $ \;
Initialize $L=(v_1,\dots,v_N)$ with $n$ randomly chosen vertices (uniformly)\;
$i \gets N \times c$ \{$i$ is the used portion of the budget\}\;
\While{$i < B$}{
 Select $u \in L$ with probability $\deg(u)/\sum_{\forall v \in L} \deg(v)$ \; \label{FSloop}
 Select an edge $(u,v)$, uniformly at random\; \label{FSselect}
 Replace $u$ by $v$ in $L$ and add $(u,v)$ to sequence of sampled edges\;
 $i \leftarrow i + 1$ \{can be skipped if node was previously sampled\} \;
}
\caption{Frontier Sampling (FS)}\label{alg:FS}
\end{algorithm}
Algorithm~\ref{alg:FS} describes FS. There are three parameters: the sampling budget $B$,  the initial cost of placing a walker $c \geq 1$ and the average number of nodes  $b$ sampled by a walker.
The initial walker locations are chosen uniformly at random over the node set (line 2).
Note that the number of walkers is taken to be $n = B/(c+b) $,
that the cost of a random walk step is one (except for previously sampled nodes) and that the cost of initially
placing a walker, $c$, can be greater than one because uniform node sampling
is often expensive.
FS keeps a list $L$ of $n$ vertices representing the locations of the $n$ walkers.
At each step, a walker is chosen from $L$ in proportion to the degree of the node where it is currently located (line 5).
The walker then moves from $u$ to an adjacent node $v$ (lines 6 and 7).

Frontier sampling is equivalent to the sampling process of a single random walker over the $n$-th Cartesian power of $G$.
%\fm{REMOVE:, $G^n = (V^n,E_n)$, where 
%\[
%V^n= \{(v_1,\dots,v_n) \given v_1 \in V \wedge \dots \wedge v_n \in V\}
%\]
%is the $n$-th Cartesian power of $V$. For all $ {\bf v},{\bf u}\in V^n$, $({\bf v},{\bf u}) \in E_n\,$ if there exists an index $i \in \{1,\ldots,n\}$ such that $(v_i,u_i) \in E$ and $u_j = v_j$ for $j \in \{1,\ldots,n\}/\{i\}$ \cite[Lemma 5.1]{TechReport}.}
 For this reason, Frontier Sampling can be thought of as an $n$-dimensional random walk (see \cite[Lemma 5.1]{TechReport}).

%\begin{figure}[htb]
%\begin{center}
%%\def\JPicScale{0.8}
%\includegraphics[trim={0 0 0 32},clip]{FrontierExample_trim.pdf}
%\caption{REMOVE: Illustration of the Markov chain associated to FS with dimension $n=2$.\label{fig:exFS}}
%\end{center}
%\end{figure}

%\fm{REMOVE: Let $L_t = (v_1, \ldots, v_n)$ denote the state of FS before the $t$-th step, $t=1,\ldots$. Theorem \ref{thm:FST} establishes key statistical properties of Frontier Sampling.
%A more complete version of this theorem is presented and proved in~\cite[Theorem 5.2]{TechReport}.
%\begin{thm} \label{thm:FST}
% Recall that $G$ is an undirected graph. 
% If $G$ is connected and non-bipartite, then the stationary behavior FS exhibits the following properties: 
%\renewcommand{\labelenumi}{(\Roman{enumi})}
%\begin{enumerate}
% \item sampled edges form a stationary sequence and their marginal distribution is uniform on $E$, 
% \item $L_\infty=(v_1,\dots,v_n)$ has the unique distribution
% $$ \pi_\mathbf{v} = \frac{\sum_{i = 1}^n \deg(v_i)}{ n \vert V \vert^{n-1} \vol(V) }, \quad \textrm{for }\mathbf{v} \in V^n.$$
% %\item the sequence of sampled edges satisfies the Strong Law of Large Numbers.
%\end{enumerate}
%\end{thm}}

Using FS samples to estimate node label distributions is simple
when the input corresponds to the first scenario described in Section~\ref{sec:definitions}.
The probability of sampling a given node is proportional to its undirected degree in $G$.
Hence, each sample must be weighted inversely proportional to the
respective node's undirected degree. Storing the undirected version of the observed graph
along with labels associated with sampled nodes allows the sampler to
avoid having to pay the cost of revisiting a node. 

Conversely, when incoming edges are not observed, Frontier Sampling
can still be adapted to remove bias. We
present this method in Section~\ref{sec:dufs}.

\subsection{Directed Unbiased Random Walk: a random walk adapted for directed networks with unobservable in-edges}\label{sec:durw}

 The presence of hidden incoming edges but observable outgoing edges makes characterizing large directed graphs through crawling challenging.
 Edge $(u,v)$ is a hidden incoming edge of node $v$ if $(u,v)$ can only be observed from node $u$.
 For instance, in Wikipedia we cannot observe the edge (``Columbia Records'', ``Thomas Edison'') from Thomas Edison's wiki entry (but this edge is observable if we access the Columbia Records's wiki entry).
 
 These hidden incoming edges make it impossible to remove any bias incurred by walking on the observed graph, unless we crawl the entire graph.
 Moreover, there may not even be a directed path from a given node to all other nodes.
 Graphs with hidden outgoing edges but observable incoming edges exhibit essentially the same problem.
 In \cite{RibeiroINFOCOM2012}, we proposed the Directed Unbiased Random Walk (DURW) algorithm,
 which obtains asymptotically unbiased estimates of node label densities on a directed graph  with
unobservable incoming edges.
 Our random walk algorithm follows two main principles to achieve unbiased samples and reduce variance:
\begin{itemize}

 \item {\em Backward edge traversals}: in real-time we construct an undirected graph $G_u$ using nodes that are sampled by the walker on the directed graph $G_d$. The role of the undirected graph is to guarantee that, at the end of the sampling process, we can approximate the probability of sampling a node, even though in-edges are not observed. The random walk proceeds in such a way that its trajectory on $G_d$ is consistent with that of a random walk on $G_u$. The walker is allowed to traverse some of the edges in $G_d$ in a reverse direction. However, we prevent some
 of the observed edges to be traversed in the reverse direction by not including them in $G_u$. More precisely, once a node $z$ is visited at
 the $i$-th step, no in-edges to $z$ observed at step $j > i$ (by visiting nodes $s$ such that $(s,z) \in E_d$) are added to $G_u$. This is an important feature
 to reduce the random walk transient and thus, reduce estimation errors.  
%  This allows us to generate an asymptotically unbiased estimate of the out-degree distribution.

 \item {\em Degree-proportional jumps}: the walker makes a limited number of random jumps to guarantee that different parts of the directed graph are explored.
  In DURW, the probability of randomly jumping out of a node $v$, $\forall v \in V$, is $w/(w + \deg(v))$, $w > 0$.
% \fm{REMOVE: This modification is based on the following observation: let $G_u$ be a weighted undirected graph formed by adding a {\em virtual} node $\sigma$ such that $\sigma$ is connected to all nodes in $V$ with edges having weight $w$. All remaining edges have unit weight.
%In a weighted graph a walker transverses a given edge with probability proportional to the weight of this edge.}
The steady state probability of visiting a node $v$ on $G_u$ is $(w + \deg(v))/(\vol(V)+ w \vert V \vert)$.
Similar to the cost of placing a FS walker through uniform node sampling, we assume that each random jump incurs cost $c \geq 1$.  

 \end{itemize}

 %For instance, in some countries individuals, politicians, and NGOs may be required to publicly disclose their contributors but contributors are not required to publicly list their contributions.

\subsubsection*{The DURW algorithm} \label{sec:DURWdescription}
 DURW is a random walk over a {\em weighted undirected connected graph} $G_u = (V,E_u)$, which is built on-the-fly.
 We build an undirected graph using the underlying directed graph $G_d$ and the ability to perform random jumps.
 Let $G^{(i)} = (V,E^{(i)})$ denote the undirected graph constructed by DURW at step $i$, where $V$ is the node set and $E^{(i)}$ is the edge set. 
 %We call $G^{(i)}$ a ``graph'' because we allow $E^{(i)}$ to have edges of nodes that are not in $V^{(i)}$.
 %Denote by $G_u \equiv \lim_{i \to  \infty} G^{(i)}$.
 In what follows we describe the construction of $G^{(i)}$ in Algorithm~\ref{alg:DURW}, since this is one of the building blocks of the proposed algorithm, DUFS.

\begin{algorithm}[t]
\SetKwInOut{Input}{Input}
\Input{sampling budget $B$, random jump weight $w$, cost of uniform node sampling $c$}
Select $s \in V$ uniformly at random \{$s = s_1$\} \;
Initialize $\mathcal{S} = \{s\}$ and $E = \mathcal{E}(s)$ \;
$i \leftarrow c$  \{$i$ is the used portion of the budget\}\;
\While{$i < B$}{
%$s \leftarrow \textrm{random walker location}$ \;
$p \sim \textrm{Uniform}(0,1)$ \;
\eIf{$p \leq w/(w+\textrm{deg}(s))$}{
 Select $s$ uniformly at random from $V$ \{random jump\} \;
  $i \leftarrow i+c$ \; \label{lin:inc1}
}{
 Select $s$ uniformly at random from $\{v: (s,v) \in E\}$ \{random walk step\} \;
 $i \leftarrow i+1$ \; \label{lin:inc2}
}
 \If{$s \notin \mathcal{S}$}{
 $\mathcal{S} \leftarrow \mathcal{S} \cup \{s\} $ \;
 $E \leftarrow E \cup \{(s,v) \in \mathcal{E}(s): v \not \in \mathcal{S}\}$
 }
}
\caption{Construction of undirected graph (common to DURW and DUFS)}\label{alg:DURW}
\end{algorithm}

Let $\mathcal{E}(v)$ denote the set of out-edges from a node $v$ in $G_d$.
 %To simplify our exposition, we include a virtual node $\sigma$ in the constructed graph, which represents a random jump.
  Let $\mathcal{S}^{(i)}=\{s_1,\ldots,s_i\}$ be the set of nodes from $V$ sampled by the random walk up to step $i$,
 where $s_j$ denotes the node on which the walker resides at step $j$. Since $V$ is not known, we track $G^{(i)}$
 using variables $\mathcal{S} = \mathcal{S}^{(i)}$ and $E = E^{(i)}$.
 The walker starts at node $s_1 \in V$ (line 1).
 We initialize $G^{(1)} = (V,E^{(1)})$, where $E^{(1)} = \mathcal{E}(s_1)$ (line 2).
 %, where $\{(u,\sigma):\forall u \in V\}$ is the set of all undirected virtual edges to node $\sigma$. 
%Let
%$$
%   W(u,v) = \begin{cases}
%             w                  & \text{if }u = \sigma \text{ or }v = \sigma \\
%             1                   & \text{otherwise}
%            \end{cases}
%$$
% denote the weight of edge $(u,v)$, $\forall (u,v) \in E^{(i)}$.
The next node, $s_{i+1}$, is selected uniformly at random from $V$ with probability $w/(w+\textrm{deg}(s_i))$ (lines 6 to 8), where $\textrm{deg}(s_i)$ is the degree of $s_i$ in $G^{(i)}$. With probability $1 - w/(w+\textrm{deg}(s_i))$, node $s_{i+1}$ is selected by performing a random walk step from $s_i$, i.e.\ by selecting a node adjacent to $s_i$ in $E^{(i)}$ uniformly at random (lines 9 to 12). When node $s_{i+1}$ is visited for the first time, it is necessary to set $\mathcal{S}^{(i+1)}$ to $\mathcal{S}^{(i)} \cup \{s_{i+1}\}$ and $E^{(i+1)}$ to $E^{(i)} \cup \{(s_i,v) \in \mathcal{E}(s_i): v \not \in \mathcal{S}^{(i)}\}$ (lines 13 to 16). By restricting the set of new edges to  $\{(s_i,v) \in \mathcal{E}(s_i): v \not \in \mathcal{S}^{(i)}\}$ instead of all edges visible from $s_i$ (i.e., $\mathcal{E}(s_i)$), we comply with the requirement that once a node $z$, $\forall z \in V$, is visited by the RW, no edge can be added to $G_u$ with $z$ as an endpoint.

%with probability $W(s_i,s_{i+1})/\sum_{\forall (s_i,v) \in E^{(i)}} W(s_i,v)$.
%Upon selecting $s_{i+1}$ we update $G^{(i+1)} = (V,E^{(i+1)})$, where
%\begin{equation}\label{eq:Ei+1}
%E^{(i+1)} = E^{(i)}  \cup  \mathcal{N}^\prime(s_{i+1}) \, ,
%\end{equation}
%and
%$$
%\mathcal{N}^\prime(s_{i+1}) = \{(s_{i+1},v) : \forall (s_{i+1},v) \in \mathcal{N}(s_{i+1}) \text{ s.t.\ }v \not \in \mathcal{S}^{(i)}\}
%$$
%is the set of all edges $(u,v)$ in $\mathcal{N}(s_{i+1})$ where node $v \notin \mathcal{S}^{(i)}$.
%Note that $\mathcal{N}^\prime(s_{i+1}) \subseteq \mathcal{N}(s_{i+1})$.
%By using $\mathcal{N}^\prime(s_{i+1})$ instead of $\mathcal{N}(s_{i+1})$ in equation~(\ref{eq:Ei+1}) we guarantee that no node in $\mathcal{S}^{(i)}$  changes its degree, i.e., $\forall v \in \mathcal{S}^{(i)}$ the degree of $v$ in $G^{(i)}$ is also the degree of $v$ in $G_u$.
%Thus, we comply with the requirement that once a node $v$, $\forall v \in V$, is visited by the RW no edge can be added to $G_u$ with $v$ as an endpoint.
 
%In the actual implementation, it is only necessary to keep track of nodes in $\mathcal{S}^{(i)} \cup \bigcup_{v \in \mathcal{S}^{(i)}\setminus \{\sigma\} } \mathcal{N}(v)$
%and the edges in $E_d$ leaving each node $v \in \mathcal{S}^{(i)} \setminus \{\sigma\} $.
%In fact, while the virtual node $\sigma$ is connected to all nodes in $V$, the sampler does not have access to the identities of nodes other than the ones that were already observed.
In order to estimate node label distributions from DURW observations, we weight samples in proportion to the inverse
 probability that the corresponding vertices are visited by a random walk in $G_u$, in steady state.
 Storing labels and edges associated with nodes in $\mathcal{S}^{(i)}$ saves the cost of querying repeated nodes. Such savings could be reflected in  Algorithm~\ref{alg:DURW} by conditioning the increase in $i$ (lines \ref{lin:inc1} and \ref{lin:inc2}) on $s \notin \mathcal{S}$.

\section{Generalizing FS and DURW: a new method applicable regardless of in-edge visibility} \label{sec:dufs}

This section is divided into two parts.
In Section~\ref{sec:proposed} we propose Directed Unbiased Frontier Sampling (DUFS), which generalizes FS to allow estimation on directed graphs with unobservable in-edges
(second scenario described in Section~\ref{sec:definitions}). DUFS also generalizes DURW: the latter
is a special case of DUFS where the number of walkers is one.
Next, in Section~\ref{sec:estimators}, we describe two ways to estimate node label
distributions using DUFS. The first uses only on the observations collected
during the walks. The second estimator we
leverages observations obtained from the initial walker locations
in addition to observations obtained during the walks.

\subsection{Directed Unbiased Frontier Sampling}\label{sec:proposed}

Like FS, Directed Unbiased Frontier Sampling (DUFS) samples a network through $n$ coordinated walks.
At each step, it selects a walker in proportion to the degree of the node where it
currently resides. Similar to the Directed Unbiased Random Walk, it
constructs an undirected graph in real-time that allows {\em backward edge traversals}.
Denote by $G^{(i)} = (V,E^{(i)})$ the undirected graph constructed by
DUFS at step $i$. DUFS does not include edges in $G^{(i)}$ that would cause walkers to
have a view of the graph inconsistent with the view at a previous point in time. In other words,
when node $u$ is visited for the first time at step $i$, $u$ is inserted in $G^{(i)}$ along
with all edges $(u,v) \in E_d$ such that $v$ has not been sampled.
Thus, the degree of $u$ is fixed in $G^{(j)}$, for all $ j \geq i$.
\fm{Alternatively, letting the degree of $u$ change at a given point would require us
to discard the the entire sample up to that point, otherwise the resulting estimator would not be consistent.
In fact, even that approach would not yield a consistent estimator for an infinite
power law graph: node degrees would never stop changing.}

%\footnote{Note
%that a node $j$ which was previously seen can have its degree increased in $G_u$
%as long as it hasn't been sampled.}

It may seem that there is no need to include {\em degree-proportional jumps} to visit different graph components
when a large number of walkers are initially spread throughout the graph (e.g., on nodes chosen uniformly).
However, including degree-proportional jumps in DUFS is extremely beneficial
because it prevents walkers from being trapped when initially located on vertices whose out-degree is zero or in components with no outgoing edges.
More generally, it allows walkers to move from small volume to large volume components and, hence,
obtain more samples among large degree nodes.

Algorithm~\ref{alg:dufs} describes DUFS. In addition to FS' three parameters, it takes a random jump weight $w$ as input. The number of
walkers and their initial locations are chosen as in FS (lines 1-3). In the extreme case where $b=0$, DUFS degenerates to uniform node sampling.
When the underlying graph is symmetric and the jump weight is $w=0$, it becomes FS. When in-edges are invisible and the number of walkers is 1, DUFS degenerates to DURW. We initialize $\mathcal{S} = L$ and $E^{(i)} = \cup_{s \in L} \mathcal{E}(s)$ (line 4).
Unlike in FS, a walker is chosen from $L$ in proportion to the sum
of the {\em random jump weight} $w$ and the degree of node where it is currently located {\em based on} $E^{(i)}$ (line 6).
Similar to DURW, the next node is selected based on either a random jump or on following an edge (lines 7-14).
Last, the undirected graph is updated (lines 15-18) and so is set $L$ (line 19).

%DURW is a special case of DUFS when the budget per walker is $b=B-c$ (i.e., when the number of walkers is $n=1$).

\begin{algorithm}[t]
\SetKwInOut{Input}{Input}
%%\begin{algorithmic}[1]
\Input{sampling budget $B$, budget per walker $b$, cost of uniform sampling $c$, jump weight $w$}
$n \gets  B/(c+b) $ \{$n$ is the number of walkers\}\;
Initialize $L=\{v_1,\dots,v_N\}$ with $n$ randomly chosen vertices (uniformly)\;
$i \gets N \times c$ \{$i$ is the used portion of the budget\}\;
Initialize $\mathcal{S} = L$ and $E = \cup_{s \in L} \mathcal{E}(s)$ \;
\While{$i < B$}{
Select $v \in L$ with probability $(w+\deg(v))/(nw + \sum_{\forall v_j \in L} \deg(v_j))$ \; \label{FSloop}
Sample $p \sim \textrm{Uniform}(0,1)$\;
 \eIf{ $p < w/(w+\deg(v))$}{
     Select a node $v \in V$ uniformly at random\;
      $i \leftarrow i + c$\;
% \Else
}{
Select an outgoing edge of $v$, $(v,v^\prime)$, uniformly at random\; \label{FSselect}
 $i \leftarrow i + 1$\;
}
 \If{$s \notin \mathcal{S}$}{
 $\mathcal{S} \leftarrow \mathcal{S} \cup \{s\} $ \;
 $E \leftarrow E \cup \{(s,v) \in \mathcal{E}(s): v \not \in \mathcal{S}\}$
 }
 Replace $v$ by $v^\prime$ in $L$ and add $(v,v^\prime)$ to sequence of sampled edges\;
}
%%\end{algorithmic}
\caption{Directed Unbiased Frontier Sampling (DUFS)}\label{alg:dufs}
\end{algorithm}

\subsection{Estimation}\label{sec:estimators}

In this section we describe two estimators of node label distributions from samples obtained by DUFS.
%These estimators generalize estimators proposed for FS and DURW.
\fm{The first estimator is based on the observations
obtained from edges traversed by the random walks. The second estimator combines these observations with those obtained
from the walkers initial locations. When used with a variance reduction heuristic, the latter
produces better estimates than the former.} For a description of estimators of edge label distribution
and other graph characteristics, please refer to \cite{TechReport}.

% Estimators that take  the sampled as input are commonly used to estimate node-oriented metrics (such as the degree distribution) and can be found in the literature~\cite{RDSprob}.
 
% We present estimators for the edge label distribution (the fraction of edges with a given label in the graph)
% %the assortative mixing coefficient~\cite{NewmanAssortativity}
% and the node label distribution. %, and the global clustering coefficient~\cite{GCC}. 
%% In a slight abuse of notation we use $E^\prime = E$ to denote that all edges of $G$ that are sampled exactly once.
% Designing estimators for these and other graph characteristics is straightforward:
%\begin{enumerate}
% \item[(1)] First we find a function $f$ that computes the characteristic of $G$ using $E$; 
% \item[(2)] then we replace $E$ with the sequence of edges sampled by a stationary RW.
%\end{enumerate}
% In what follows we illustrate how to build an estimator of the edge label and node label distributions.
% Estimators for other graph properties can be derived in a similar way. In \cite{TechReport},
% estimators for the assortative mixing coefficient --  a measure of the correlation of labels between two neighboring vertices --
% and for the global clustering coefficient are derived and evaluated.
 
 %\input{TEX/uedge}
 
% Activate the following line by filling in the right side. If for example the name of the root file is Main.tex, write
% "...root = Main.tex" if the chapter file is in the same directory, and "...root = ../Main.tex" if the chapter is in a subdirectory.
 
%!TEX root =  ../frontier-tkdd.tex 

\subsubsection{Node Label Distribution: random edge-based estimator}\label{sec:labeldistribution}

Let $s_i$ denote the $i$-th node visited by DUFS, $i=1,\dots,t$, $t \leq B - Nc$. 
 Let $\theta_\ell$ be the fraction of nodes in $V$ with label $\ell \in \cL_v$.
 Let $\pi(v)$ be the steady state probability of sampling node $v$ in $G_u$, $\forall v \in V$.
 The node label distribution is estimated at step $t$ as
\begin{equation} \label{eq:hattheta}
\hat{\theta}_\ell= \frac{1}{n} \sum_{i = 1}^t \frac{ \mathds{1}\{\ell \in \cL_v(v) \}}{\hat{\pi}(s_i)} \, , \qquad \ell \in \cL_v, \, t=1,\ldots,B-Nc,
\end{equation}
where $\mathds{1}\{P\}$ takes value one if predicate $P$ is true and zero otherwise,
%\begin{equation*}
%h_\ell(v)=\begin{cases}
%1 & \textrm{if $\ell \in \cL_v(v)$ ,}\\
%0 & \textrm{otherwise}
%\end{cases}
%\end{equation*}
and $\hat{\pi}(s_i)$ is an estimate of $\pi(s_i)$:
$\hat{\pi}(s_i) = (w + \deg(s_i)) S$.
Here $\deg(v)$ is the degree of $v$ in $G^{(\infty)}$
and 
\begin{equation}\label{eq:norm}
 S = \frac{1}{t} \sum_{i = 1}^t \frac{1}{w + \deg(s_i) } \, .
\end{equation}
The following theorem states that $\hat{\pi}(s_i)$ is asymptotically unbiased. 

\begin{theorem}
$\hat{\pi}(s_i)$ is an asymptotically unbiased estimator of $\pi(s_i)$.
\end{theorem}
\begin{proof}
To show that $\hat{\pi}(s_i)$ is asymptotically unbiased, we first note that the limit  $\lim_{t \to \infty} E^{(t)} = E^{(\infty)}$ exists, since after visiting all vertices we will never add any additional edges. We then invoke Theorem~4.1 of~\cite{TechReport}, yielding $\lim_{t \to \infty} S = \vert V \vert / (\vert E^{(\infty)}\vert + \vert V \vert w) $ almost surely.
Thus, $\lim_{t \to \infty} \hat{\pi}(s_i) = {\pi}(s_i)$ almost surely.
Taking the expectation of~(\ref{eq:hattheta}) in the limit as $t \to \infty$ yields
$E[\lim_{t \to \infty} \hat{\theta}_\ell] = {\theta}_\ell,$ which concludes our proof.
\end{proof}

\subsubsection{Node Label Distribution: leveraging information from walkers' initial locations}\label{sec:hybrid}

%Later in Section~\ref{sec:degdist} we investigate the impact of the number of walkers in estimation accuracy and conclude that, in most cases, 1/10th of the budget is a good choice \fm{for node label estimation only?}. However, this proportion yields a large number of random node samples that are thrown away by our original estimator.
%We will see in Section~\ref{sec:nwalkers} that is beneficial to use a large number of walkers (e.g., 1/10th of the budget).
The estimator presented in~\eqref{eq:hattheta} does not make use of information associated with the initial set of nodes on which the walkers
are placed. When the number of walkers
is large this results in the loss of a considerable amount of statistical information.
%In FS and DUFS, we use uniform node sampling to determine the walkers initial locations.
%When the number of walkers is a large fraction of the budget, their initial locations
%contain a valuable source of information about the node label distribution.
%Unfortunately, the estimator in eq.~\eqref{eq:hattheta} does not account for the walkers' initial locations.
However, including these observations is challenging because
subsequent observations from random walk steps are not independent of the initial observations. %,
%(ii) observations from the random walk steps are not independent of each other and
Moreover, the normalizing constant for the random walk observations is no longer given by~\eqref{eq:norm},
since degree distribution estimates also depend on the information contained in the node samples.
%Since this
%distribution is estimated using both types of samples, we need an iterative procedure
%to estimate node label distributions.

In this section, we
derive a new estimator that circumvents these problems by
approximating the likelihood of RW samples by
that associated with random edge sampling. %and then estimating the corresponding normalizing constant
%by maximizing the likelihood function.
We call it the {\em hybrid estimator} because it combines
observations from initial walker locations and random walks steps.
The hybrid estimator significantly improves
the estimation accuracy for labels associated with large probability masses.

Let us index the node labels $\cL_v$ from $1$ to $W$, where $W=|\cL_v|$. We refer to the sum $\deg(v)+w$ in DUFS as the {\em random walk bias} for node $v \in V$.
To simplify the notation, we assume that each node has exactly one label and that random walk biases take on integer values in $[1,\ldots,Z]$, for some maximum value $Z$.
Denote the node label distribution as $\btheta = (\theta_1,\ldots,\theta_W)$.
%We define the joint label-bias distribution to be $\theta = [\theta_{i,j}]_{W\times Z}$, for some $\theta  \in \Delta_W \times \{1,\ldots,Z\}$, where $\Delta_W$ is the
%$W$-simplex.
Let $n_{i}$ denote the number of walkers starting on label $i$ nodes and $m_{i,j}$ the number of subsequent observations of label $i$ and bias $j$ nodes.
The notation is summarized in Table~\ref{tab:notation}.

\begin{table}[t]
\center
%\footnotesize
\begin{tabular}{cl}
\hline
Variable & Description \\
\hline
$n_{i}$ & number of node samples with label $i$\tabularnewline
$\theta_{i,j}$ & fraction of nodes in $G^{(t)}$ with label $i$ and undirected degree
$j$\tabularnewline
$m_{i,j}$ & number of edge samples with label $i$ and bias $j$\tabularnewline
$m_{i}=\sum_{j}m_{i,j}$ & total number of edge samples with label $i$\tabularnewline
$N=\sum_{i}n_{i}$ & total number of node samples\tabularnewline
$M=\sum_{i}m_{i}$ & total number of edge samples\tabularnewline
$B=N+M$ & total budget\tabularnewline
\hline
\end{tabular}
\caption{Notation used in hybrid estimator.}
\label{tab:notation}
\end{table}

We approximate random walk samples in DUFS by uniform edge samples from $G_u$. Experience
from previous studies shows us that this approximation works very well in practice.
Hence, the likelihood function given samples $\mathbf{n}=\{n_i:i=1,\ldots,W\}$ and $\mathbf{m} = \{m_{i,j}:i=1,\ldots,W\textrm{ and } j=1,\ldots,Z\}$
is expressed as
\begin{equation}\label{eq:likelihood2}
L(\btheta|\mathbf{n},\mathbf{m})  =  \frac{\prod_{i} \theta_i^{n_{i}} \prod_k (k\theta_{i,k})^{m_{i,k}}}{\left(\sum_{s,t} t\theta_{s,t}\right)^{M}}.
\end{equation}

The maximum likelihood estimator $\theta^\star$ is the value of $\theta$ that maximizes~\eqref{eq:likelihood2} subject to $0 \leq \theta_i \leq 1$ and $\sum_i \theta_i = 1$.
This defines a constrained non-convex optimization problem. However, we can convert this optimization problem into an unconstrained problem using the reparameterization
$\theta_i = e^{\beta_i}/\sum_k e^{\beta_k}$ for $i=1,\ldots,W$. As shown in Appendix~\ref{app:properties}, the partial derivatives of the resulting objective function are
 \begin{equation}\label{eq:derivative2}
\frac{\partial \mathcal{L}(\bbeta|\mathbf{n},\mathbf{m})}{\partial\beta_{i}} =  n_i + m_i  -\frac{Ne^{\beta_i}}{\sum_j e^{\beta_j}} -\frac{M e^{\beta_i} m_{i}/ \mu_i}{\sum_{s} e^{\beta_s} m_{s}/ \mu_s },  \qquad i=1,\ldots,W,
 \end{equation}
where $m_i=\sum_k m_{i,k}$ and $\mu_i = \sum_k m_{i,k}/k$. Setting one of the variables to a constant (say, $\beta_W = 1$) for identifiability and then using the gradient descent method to change the remaining variables according to~\eqref{eq:derivative2} is guaranteed to converge provided that we make small enough steps. 
 An interesting interpretation of~\eqref{eq:derivative2} is obtained by setting the derivatives to zero and substituting back $\theta_i = e^{\beta_i}/\sum_k e^{\beta_k}$:
\begin{equation}
% \theta_i^\star  =  \frac{n_i + m_i}{N + M \frac{m_{i}/ \mu_i}{\sum_{s} \theta_s^\star m_{s}/ \mu_s }},\qquad i=1,\ldots,W.\label{eq:estimator_alpha}
  \theta_i^\star  =  (n_i + m_i)\left(N + M \frac{m_{i}/ \mu_i}{\sum_{s} \theta_s^\star m_{s}/ \mu_s }\right)^{-1},\qquad i=1,\ldots,W.\label{eq:estimator_alpha}
\end{equation}
According to~\eqref{eq:estimator_alpha}, the estimated fraction of nodes with label $i$ is the total number of times label $i$ was observed  (i.e., $n_i + m_i$) normalized by sum of (i) the number of random node samples and (ii) the
number of random edge samples weighted by the probability of sampling label $i$ from one random edge sample.
In the limit as $N$ and $M$ go to infinity, we can show that $\btheta^\star = \btheta$
is a solution, but we cannot prove that it is unique or that $\btheta^\star$ converges to $\btheta$.
Hence, we cannot prove that $\btheta^\star$ is asymptotically unbiased.

The system of non-linear equations determined by~\eqref{eq:estimator_alpha} cannot be solved directly, but
can be tackled by Expectation Maximization (EM). In this case, the term $\sum_{s} \theta_s^\star m_{s}/ \mu_s$ 
in the denominator is replaced by its expected value given $\theta_i$'s from the previous iteration.
Based on the same idea, if we replace $\sum_{s} \theta_s^\star m_{s}/ \mu_s$ with an edge sampled-based estimator $\hat d$ for the average degree in $G_u$,
we obtain the following non-recursive variant of the hybrid estimator,
\begin{equation}\label{eq:nonrecursive}
\hat{\theta}_{i} = (n_{i}+m_{i})\left(N+M\frac{m_{i}}{\mu_{i} \hat{d}}\right)^{-1},\qquad i=1,\ldots,W,
%\hat{\theta}_{i} = \frac{n_{i}+m_{i}}{N+M\frac{m_{i}}{\mu_{i} \hat{d}}},\qquad i=1,\ldots,W,
\end{equation}
where $\hat d = M/(\sum_i \mu_i)$.
Theorem~\ref{thm:asymp} below states the conditions under which $\hat \theta_i$ is asymptotically unbiased (see appendix for proof). 
In practice, we find no significant difference between $\theta_i^\star$ and $\hat \theta_i$, except when the number of walkers $N$
is very large and the jump weight $w$ is very small. For those cases, $\theta_i^\star$ tends to be slightly more accurate
than $\hat \theta_i$ for small values of $i$, which in some applications may justify the additional computational cost
of executing gradient descent or EM.
\begin{thm} \label{thm:asymp}
Let $N=\alpha B$
and $M=(1-\alpha)B$, for some $0<\alpha<1$. In the limit as $B\rightarrow \infty$, the estimator $\hat{\theta}_{i}$ is an unbiased estimator of $\theta_{i}$.
\end{thm}

%
%The Lagrangian becomes
%
%\begin{equation}
%H(\theta,\lambda|\mathbf{n},\mathbf{m})=\mathcal{L}(\theta|\mathbf{n},\mathbf{m}) - \lambda\left(\sum_{j} \theta_j-1\right)\label{eq:lagrange}
%\end{equation}
%whose corresponding partial derivatives are
%\begin{eqnarray}
%\frac{\partial H(\theta,\lambda|\mathbf{n},\mathbf{m})}{\partial\theta_{i}} & = & -\frac{M\sum_k m_{i,k}/\mu_i}{\sum_s \theta_s \sum_z m_{s,z}/\mu_s}+\frac{n_i}{\theta_i} + \frac{m_{i}}{\theta_i}-\lambda\quad\textrm{and}\label{eq:partialg_alpha}\\
%\frac{\partial H(\alpha,\lambda|\mathbf{n},\mathbf{m})}{\partial\lambda} & = & -\left(\sum_{i}\theta_{i}-1\right).\label{eq:partialg_alpha2}
%\end{eqnarray}
%
%Setting eqs. (\ref{eq:partialg_alpha}) and (\ref{eq:partialg_alpha2}) to zero, yields
%\begin{equation}
%\frac{n_i+m_i}{\theta_i}=\lambda+\frac{M\sum_k m_{i,k}/\mu_i}{\sum_s \theta_s \sum_z m_{s,z}/\mu_s}\qquad\textrm{and}\qquad\sum_{i}\theta_i=1.
%\end{equation}
%Multiplying both sides of the equation on the LHS by $\theta_i$ and summing over all $i$, it follows that $\lambda=N$ at the critical point.
%Substituting $\lambda=N$ in (\ref{eq:partialg_alpha}), we obtain the maximum likelihood estimator
%\begin{equation}\label{eq:estimator_alpha}
%\hat\theta_{i}=\frac{n_i+m_i}{N+M\frac{m_i/\mu_i}{\sum_s \hat \theta_s \sum_z m_{s,z}/\mu_s}},
%\end{equation}
%
In the special case where the label is the undirected degree itself, we have $\mu_i = m_i/i$. Hence, eq.~\eqref{eq:nonrecursive} reduces to
\begin{equation}\label{eq:und_theta}
\bar \theta_i = \frac{n_i+m_i}{N+Mi/\hat d},
%$\bar \theta_i = (n_i+m_i)\left(N+Mi/\hat d\right)^{-1}$,
\end{equation}
where $\hat d$ is the estimated average degree.
When the average degree is known, we can show that $\bar \theta_i$ is unbiased and, moreover, the minimum variance unbiased estimator (MVUE) of $\theta_i$ (see appendix for proof).

%\begin{thm}\label{thm:mvue}
% When $\hat d = d$ (i.e., average degree is known), $\bar \theta_i$ is unbiased and, moreover, the minimum variance unbiased estimator (MVUE) of $\theta_i$.
%\end{thm}

%In Appendix~\ref{app:properties} we prove that estimator $\hat \theta_i$ -- and hence $\bar \theta_i$ --
%is asymptotically unbiased. Moreover, when the average degree $\hat d$ is known, we prove that $\bar \theta_i$ is the minimum variance unbiased estimator.

%The hybrid estimator can generate more accurate estimates than the edge-based estimator for samples acquired using FS or DUFS,
%especially on regions of the distribution where probability mass is larger.

When $n_i > 0$ but $m_i = 0$, the estimator in eq.~\eqref{eq:nonrecursive} reduces to $\hat \theta_i = n_i/N$,
which is essentially the MLE for uniform node sampling. It is well known that this estimator is
not nearly as accurate as a random walk based estimator for large out-degree values with small probability mass.
In some sense, the estimator $\hat \theta_i = n_i/N$ does not account for the
fact that the number of random walk samples is zero. As a result, mass estimates
for large out-degrees tend to have very large variance when no random walk samples
are observed. Fortunately, we find that the following heuristic rule can drastically reduce the estimator variance
in these cases. %(as we empirically show in Section~XXX).

\paragraph{Variance reduction rule} If no random edge samples
are observed for out-degree $i$, we set the estimate $\hat{\theta}_i = 0$. This implies that we ignore any random
node samples seen of nodes that have out-degree $i$. While this clearly results in a biased estimate, as
the budget per walker $b$ goes to infinity, the probability of invoking this rule goes to zero.
Hence, it produces an asymptotically unbiased
estimate. This rule can be interpreted as a combination of node-based and edge-based estimates
in proportion to the reciprocals of their estimated variances. That is, when no random edge samples are observed
for a given out-degree, the corresponding estimated variance is zero and hence, random node samples
should be ignored. We note that the converse rule (i.e., set $\hat{\theta}_i = 0$ if no random node samples were observed)
would not perform well, as the probability of sampling large out-degrees with random node sampling is very small.

We simulate DUFS on several datasets and compare the results obtained
with the hybrid estimator when the rule is used and when it is not.
 Simulation details, datasets and the error metric (normalized root mean square error)
  will be described in Section~\ref{sec:parameters}.
 Figure~\ref{fig:rule} shows representative results of the impact of the rule
when estimating out-degree distributions using DUFS in conjunction with the hybrid estimator
 on two network datasets (averaged over 1000 runs).
The results show that the rule consistently reduces
estimation error in the distribution tail without affecting estimation quality for small values of $i$.

\begin{figure}[]
\centering
      \includegraphics[width=0.96\textwidth]{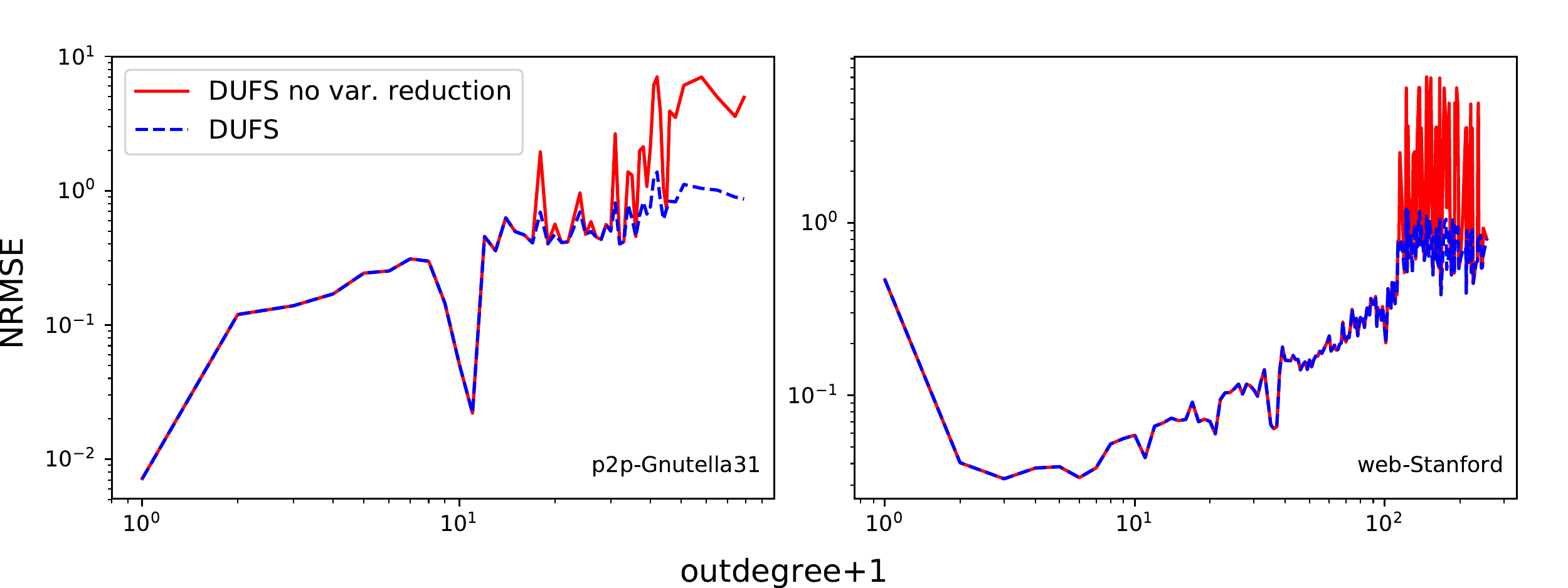}

%    \subfloat[p2p-Gnutella31\label{fig:p2p-Gnutella31-RVT}]{%
%      \includegraphics[width=0.49\textwidth,height=0.40\textwidth]{p2p-Gnutella31-RVT.pdf}
%    }
%    \subfloat[web-Stanford\label{fig:web-Stanford-RVT}]{%
%      \includegraphics[width=0.49\textwidth,height=0.40\textwidth]{web-Stanford-RVT.pdf}
%    }%
    \caption{(visible in-edges) Effect of variance reduction rule on NRMSE, when $B=0.1|V|$ and $c=1$.
    Using information contained in random node samples can increase variance for large out-degree estimates.
    However, the proposed rule effectively controls for that effect without decreasing head estimates accuracy.}
    \label{fig:rule}
 \end{figure}

\paragraph{In-degree distribution: impossibility result}

The fact that long random walks are often approximated by random edge sampling brings up the question
of whether they can be used to estimate in-degree distributions when the in-degree is not observed directly.
Under random edge sampling, the number of observed edges pointing to a node is binomially distributed
and a maximum likelihood estimator can be derived for estimating the in-degree distribution.
This problem is related to the set size distribution estimation problem, where elements
are randomly sampled from a collection of non-overlapping sets and the goal is to recover
the original set size distribution from samples.
In addition to in-degree distribution in large graphs, this problem is related to
the uncovering of TCP/IP flow size distributions on the Internet.

In~\cite{jsac13}, we derive error bounds for the set size distribution estimation problem from an
information-theoretic perspective. The recoverability of original set size distributions presents a sharp threshold with respect
to the fraction of elements sampled from the sets. If this fraction lies below the threshold,
typically half of the elements in power-law and heavier-than-exponential-tailed distributions, then the
original set size distribution is unrecoverable (see \cite[Theorem 2]{jsac13}).

\section{Results on degree distribution estimation}\label{sec:results}

%In Section~\ref{sec:dresults} we evaluate DUFS and DURW performances on estimating graph characteristics
%when incoming edges are not directly observable. Before comparing these methods, we investigate the impact of the parameters
%on each method.
Here we focus on the estimation of degree distributions on directed networks. This section is divided into four parts. In Section~\ref{sec:parameters}, we investigate the impact of DUFS parameters on estimation accuracy. We then compare DUFS against other random walk-based methods when
both outgoing and incoming edges are visible in Section~\ref{sec:visible}. In Section~\ref{sec:dresults}, we perform a similar comparison when only out-edges are visible. Last, in Section~\ref{dufs:behavior} we provide some analysis to explain the relationship observed between the NRMSE and the out-degree (in-degree) in the results.
We will refer to the edge-based estimator defined in~\eqref{eq:hattheta} as E-DUFS.
%We will refer to the edge-based estimator defined in~\eqref{eq:hattheta} and the hybrid estimator defined in~\eqref{eq:estimator_alpha} as E-DUFS and H-DUFS respectively.

\begin{figure}[htb]
\begin{center}
\includegraphics[width=0.99\textwidth]{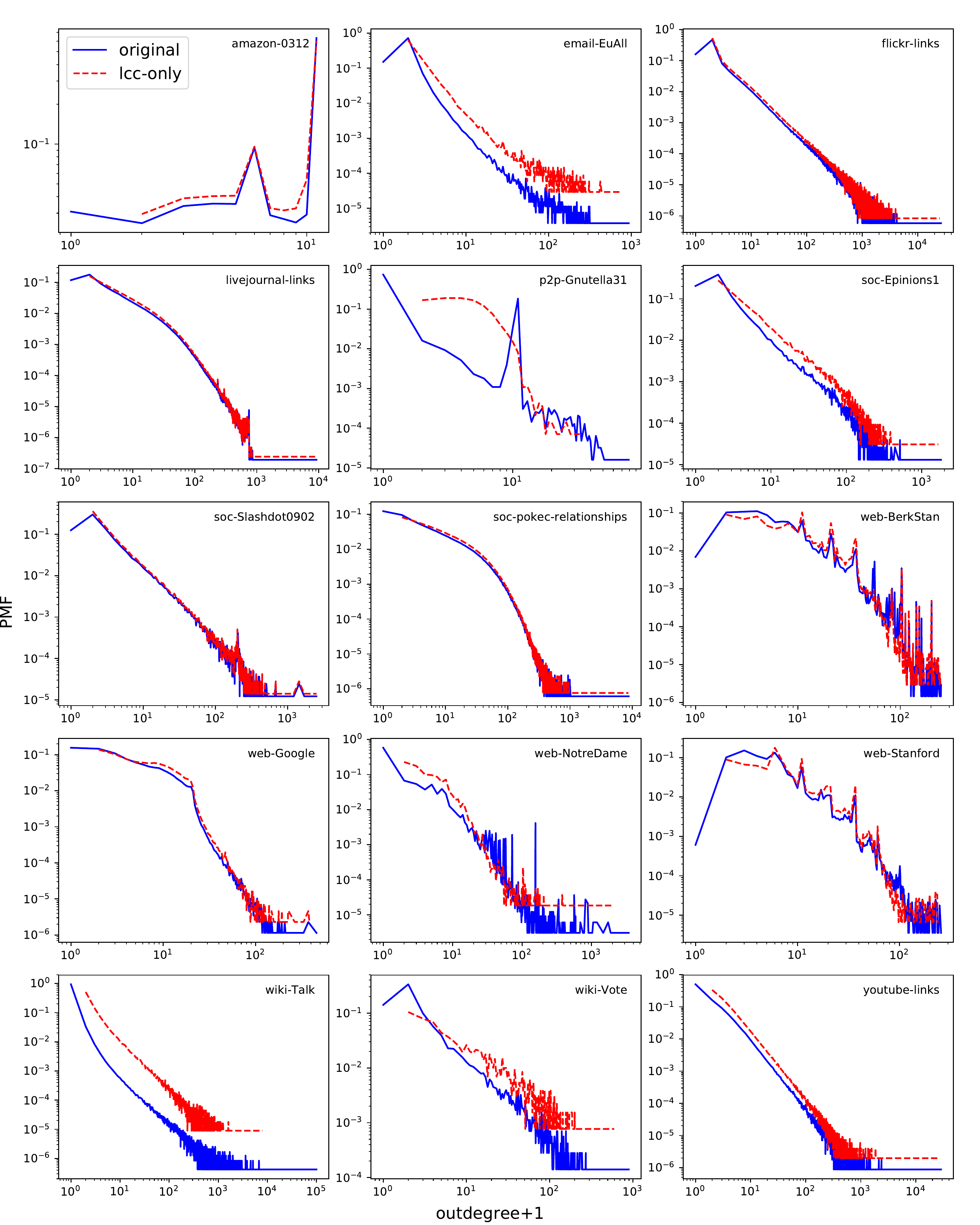}
\caption{Out-degree probability mass function (p.m.f.) for each network and its largest strongly connected component (LCC).
A large difference between these p.m.f.s suggests it is beneficial to use multiple walkers
and/or random jumps.
\label{fig:outdegree}}
\end{center}
\end{figure}

%DATASETS AND METRIC
The 15 directed network datasets in our evaluation were obtained from Stanford's SNAP \cite{SNAP}.
These datasets describe the topology of a variety of social networks, communication networks, web graphs, one Internet peer-to-peer networks
and one product co-purchasing networks. We found it informative to extract the
largest strongly connected component of each directed network and
to apply our methods to the resulting datasets -- hereby referred to as LCC datasets -- as well as to the original datasets.
Figure~\ref{fig:outdegree} shows the out-degree
probability mass function (p.m.f.) for each network, along with the out-degree p.m.f.\ for the corresponding
LCC dataset. We opt to show the p.m.f.\ instead of the complementary cumulative distribution function (CCDF) because
the estimation task in this work is defined in terms of the p.m.f.'s. Defining the estimation task in terms
of the CCDF would give DUFS an unfair advantage, as we will explain in Section~\ref{sec:visible}.

Simulations consist of sampling the network until a budget $B=0.1|V|$ (i.e., 10\% of the
number of vertices) is depleted.
Note that budget is decremented when walkers are initially placed and each time one of them moves to a node
and when they perform random jumps. We construct an undirected graph in the background throughout
each simulation. As a result, we assume that the cost to revisit a node is zero, even if this visit occurs due to
a random jump\footnote{Note that the alternative, i.e.\ always taking $c$ units off the budget per random jump, is unlikely
to impact results significantly when $B=0.1|V|$, since the vast majority of random jumps will find a non-visited node.}.

When both outgoing and incoming edges are observable, random walks disregard edge direction, and move as if
the network is undirected. In this scenario, we focus either on the estimation of the marginal out- and in-degree
distributions or the joint distribution. The methods we investigate here
can be used to estimate other node label distributions. For instance, if the underlying network is undirected,
we can estimate the (undirected) degree distribution or even non-topological properties, such
as the distribution of user nationalities in a social network. In the light of the impossibility results described
in the end of Section~\ref{sec:estimators}, we focus on out-degree distribution estimation
when incoming edges are not directly observable.

Let $\btheta = \{\theta_i\}_{\forall i \in \cL}$ denote the node label distribution, where $\theta_\ell$ is the fraction of vertices with label $\ell$.
Denote by  $\hat{\theta}_\ell$ the estimate for $\theta_\ell$.
We use normalized root mean square error (\NRMSE) of $\hat{\theta}_\ell$ as the error metric, which is a normalized measure of the dispersion of the estimates, defined as
\begin{equation} \label{eq:NRMSE}
% \MSE(\ell) = E[( \hat{\theta_\ell} - \theta_\ell )^2] \quad \textrm{and} \quad
 \NRMSE(\ell) = \frac{\sqrt{E[( \hat{\theta_\ell} - \theta_\ell )^2]}}{\theta_\ell}.
\end{equation}
In the case of marginal in-degree (out-degree) distribution, we refer to in-degrees (out-degrees) smaller than the average as the {\em head} of the distribution.
We refer to the largest 1\% in- (out-degree) values as the {\em tail} of the distribution.

%In that case, we can use the node label distribution estimator defined in eq.~\eqref{eq:estimator_alpha} to 
%characterize a network in terms of its out-degree distribution $\btheta$. We use two metrics
%commonly used to evaluate estimation error, namely mean square error (MSE) and normalized root mean square error (NRMSE).
%These are defined respectively as
%\begin{equation} \label{eq:NRMSE}
% \MSE(\ell) = E[( \hat{\theta_\ell} - \theta_\ell )^2] \quad \textrm{and} \quad \NRMSE(\ell) = \frac{\sqrt{E[( \hat{\theta_\ell} - \theta_\ell )^2]}}{\theta_\ell}.
%\end{equation}

\subsection{Impact of DUFS parameters and practical guidelines}\label{sec:parameters}

%intuition
To provide intuition on how random jump weight $w$ and budget per walker $b$ affect the accuracy of DUFS estimates,
assume for now that we replace samples collected via random walks by uniform edge samples
from the weighted undirected graph $G_u$. In this hypothetical scenario, the budget $B$ is used to collect $N \geq 1$ uniform node samples and
$B-Nc$ uniform edge samples.
Clearly, when the edge-based estimator defined in~\eqref{eq:hattheta}
is used, the most accurate node label distribution estimates are obtained by setting $N=1$,
(i.e.\ $b=B-c$). Hence, we focus on the case where the hybrid-estimator defined in~\eqref{eq:estimator_alpha} is used.
In particular, consider estimation of the out-degree distribution.

For a given value of $b$, the number of uniform node
samples will be $B/(c+b)$. For each of the remaining $B-B/(c+b)$ samples, a vertex $v$ is sampled in proportion to $\deg(v)+w$,
where $\deg(v)$ is the undirected degree of $v$ in $G_u$. The choice of $w$ and $b$ impose, individually,
a trade-off between estimation accuracy of the head and of tail of the distribution.
For a fixed value of $w$, smaller values of $b$ translate into
better estimates of the head (and worse estimates of the tail) because we collect more (less) information about that region of the distribution
from uniform node samples. For a fixed value of $b$, larger values of $w$ also translate into more (less) accurate estimates of the head (tail),
because random jumps are more likely to move a node to low in- and out-degree nodes (as they tend to occur more frequently).

%results; compare agains intuition
In what follows, we observe through simulations that {\em despite the uniform edge sampling approximation, the previous intuition holds for DUFS
head estimates,
  but not always for tail estimates}. In many cases, as we increase the number of walkers (i.e., decrease $b$) or increase $w$,
 we still obtain good
 estimates of the tail. This occurs because varying $w$ or $b$ changes the transition probability matrix that governs
the sampling process, and thus, the sample distribution.
%Since \{E,H\}-DUFS approximate RW samples by uniform samples over $E$,
%processes that generate samples whose distribution is closer to the uniform distribution results in less biased estimates.

%The impact of the random jump weight $w$ is easy to understand in the context of
%a single random walk with degree proportional jumps. This process is equivalent
%to a single random walk with no jumps over the original graph modified to include a
%new edge from each node to every other node with weight $w/(N-1)$. 

%what we do; conclusions
We simulate DUFS on each original network dataset for combinations of random jump weight $w \in \{0.1,1,10\}$ and budget per walker $b \in \{1,10,10^2,10^3\}$ (1000 runs each). \fm{For small values of $w$, DUFS behaves as FS, except for using the improved estimator. For large values of $w$, DUFS behaves as uniform node sampling. Last, for large values of $b$, DUFS behaves as DURW.}
%Values of $w$ much smaller and much larger than these would be approximately equivalent to DUFS without jumps and uniform node sampling, respectively.
%Larger values of $b$ would approximately correspond to DURW.
We consider four scenarios that correspond to whether the incoming edges are directly observable or not and
to two different costs of uniform node sampling $c=1$ or $c=10$. Evaluating these parameter combinations is useful to establish
practical guidelines for choosing DUFS parameters, which we summarize in Table~\ref{tab:summary}.
We observe that estimation accuracy tends to be lower for extreme values of these parameters, suggesting that combinations other than ones
investigated here would not provide large accuracy gains (if any).

\begin{table}[]
\centering
\caption{Practical guidelines on setting DUFS parameters to obtain accurate head or tail estimates
depending on in-edge visibility and node sampling cost $c$.}
\label{tab:summary}
\begin{tabular}{l|c|c|c|c}
        & \multicolumn{4}{c}{uniform node sampling cost} \\
        \cline{2-5}
        & \multicolumn{2}{c|}{$c=1$}                 & \multicolumn{2}{c}{$c=10$}                                                                \\
\hline
in-edges            & visible   & not visible                   & visible                       & not visible                                               \\
\hline
most accurate for & $w=10$    & $w=10$   & $w=1$   & $w=10$                                               \\
 small out-degrees &  $b=1$ &  $b=1$     & $b=10^2$ & $b=1$ \\
\hline
most accurate for  & $w=1$ & $w=1$ & $w=0.1$ & $w=0.1$ \\
large out-degrees & $b=10$ & $b=10,10^2,10^3$ & $b=10^3$ & $b=10, 10^2, 10^3$
\end{tabular}
\end{table}

\subsubsection*{Visible in-edges, $c=1$}
Figure~\ref{fig:grid} (all except bottom right) show typical results when varying $w$ and $b$. To avoid clutter,
we show only estimates for powers of two (or the closest
out-degree values) and omit results for
%$w=0$ as they are similar to those for $w=0.1$, and for
$b = 10^3$ %$b \in \{10^3,B-c\}$
as they are similar to those for $b=10^2$.
Figure~\ref{fig:grid} (bottom right)
shows similar results for amazon-0312, the dataset with the smallest maximum out-degree (max.\ is 10).
Similar to our intuition for uniform edge sampling, the NRMSE associated with the head increases with $b$ and decreases with $w$,
on virtually all datasets\footnote{For simplicity, the observations regarding the distribution head (tail) are based on the single smallest (largest) out-degree
on each dataset. Similar conclusions are obtained when combining NRMSEs associated with several of the smallest (largest) out-degrees.}. Also as expected, for a fixed values of $w$, $b=1$ yields larger errors in the tail
than $b \in \{10,100\}$ (except for amazon-0312). However, contrary to the intuition for uniform edge sampling,
$w=1$ matches or outperforms $w=0.1$ for (except for $b=1$). This is best visualized in Figure~\ref{fig:grid} (bottom right).
This happens because setting $w=1$ allows DUFS to sample regions with large probability mass (in this case, the head) and,
at the same time, allows the sampler to move walkers from low volume to high volume components more often than $w=0.1$.
We also observe that
$b=10$ outperforms $b \in \{10^2,10^3\}$ for $w \in \{0.1,1\}$.
Dataset amazon-0312 is the only dataset where $(w=10,b=1)$ obtained the best results over the entire out-degree distribution.
As a side note, we observe that for most datasets used here, in log-log scale, the NRMSE grows approximately linearly as a function of the out-degree up to a certain
point and then starts to decrease, roughly linearly too. In Section~\ref{dufs:behavior} we explain why this is the case.

\begin{figure}[]
 %   \subfloat[flickr-links\label{fig:flickr-links-grid}]{%
      \includegraphics[width=0.96\textwidth,height=0.66\textwidth]{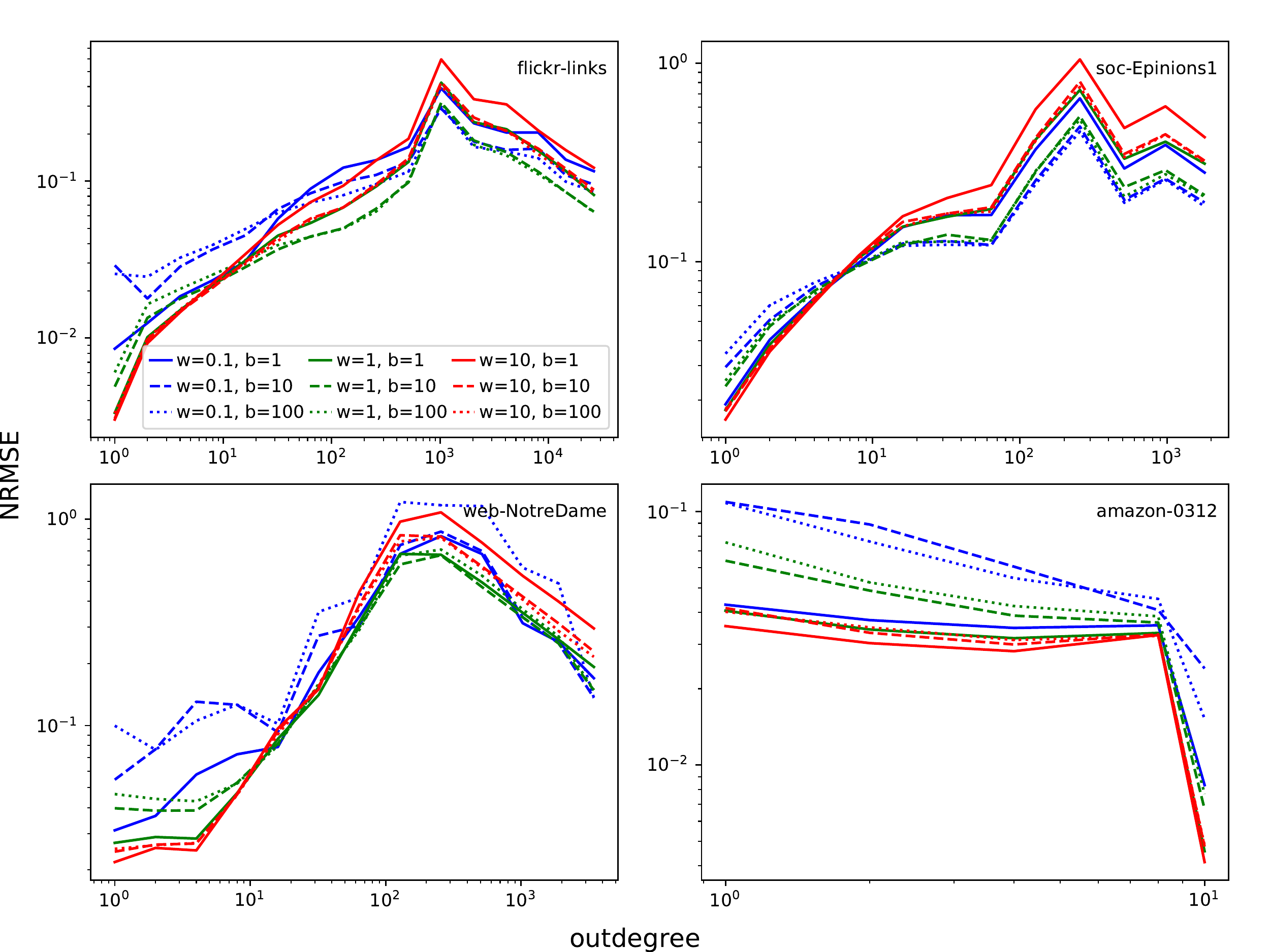}
%    }
%    \subfloat[soc-Epinions1\label{fig:soc-Epinions1-grid}]{%
%      \includegraphics[width=0.49\textwidth]{soc-Epinions1-grid.pdf}
%    }\\
%    \subfloat[web-NotreDame\label{fig:web-NotreDame-grid}]{%
%      \includegraphics[width=0.49\textwidth]{web-NotreDame-grid.pdf}
%    }
%     \subfloat[amazon-0312\label{fig:amazon-0312-grid}]{%
%      \includegraphics[width=0.49\textwidth]{amazon-0312-grid.pdf}
%    }
    \caption{(visible in-edges, $c=1$) Effect of DUFS parameters on datasets with many connected components, when $B=0.1|V|$ and $c=1$. Legend shows the average budget per walker ($b$)
    and jump weight ($w$). Trade-off shows that configurations that result in many uniform node samples, such as $(w = 10, b=1)$,
    yield accurate head estimates, whereas configurations such as $(w =1, b=10)$ yield accurate tail estimates.}
    \label{fig:grid}
 \end{figure}
\begin{figure}[]
      \includegraphics[width=0.96\textwidth]{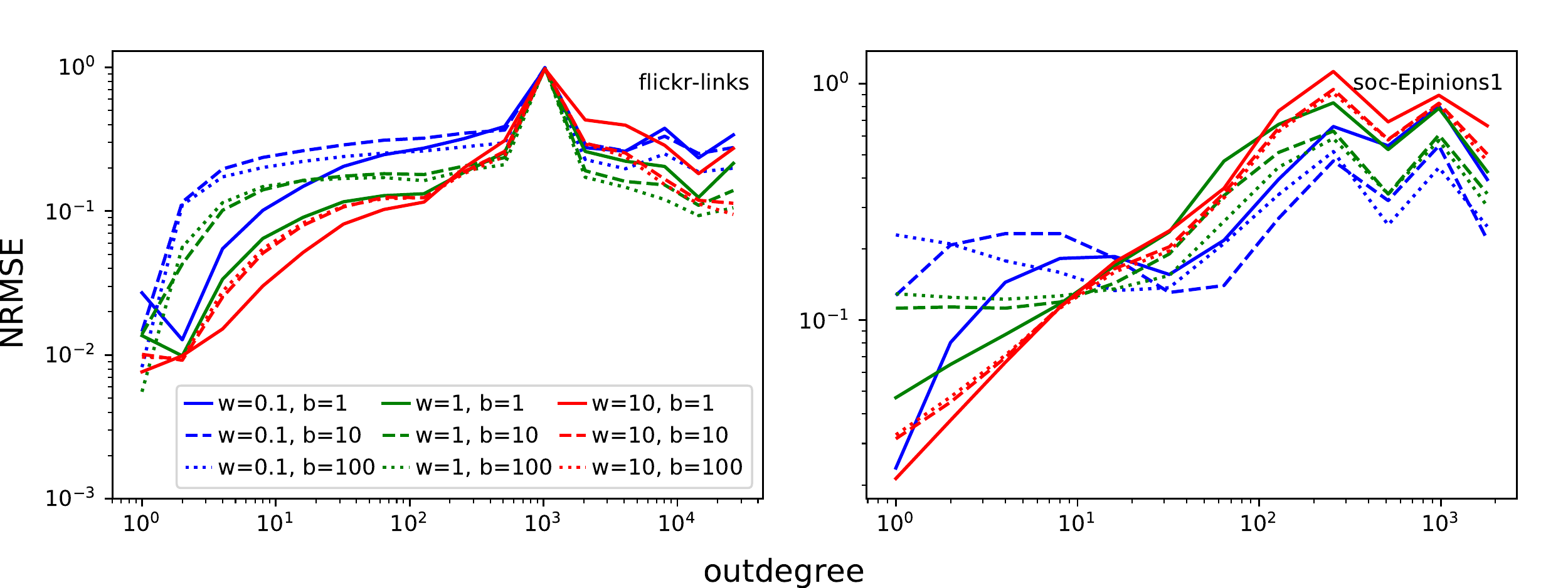}
%    \subfloat[flickr-links\label{fig:flickr-links-grid-v0}]{%
%      \includegraphics[width=0.49\textwidth]{flickr-links-grid-v0.pdf}
%    }
%    \subfloat[soc-Epinions1\label{fig:soc-Epinions1-grid-v0}]{%
%      \includegraphics[width=0.49\textwidth]{soc-Epinions1-grid-v0.pdf}
%    }
    \caption{(invisible in-edges, $c=1$) Effect of DUFS parameters on datasets with many connected components, when $B=0.1|V|$ and $c=1$. Legend shows the average budget per walker ($b$)
    and jump weight ($w$). Configurations that result in many walkers which jump too often, such as $(w\geq10, b = 1)$
    yield accurate head estimates, whereas configurations such as $(w=1, b=10^3)$, yield accurate tail estimates.}
    \label{fig:grid-v0}
 \end{figure}

\subsubsection*{Invisible in-edges, $c=1$}

The results we obtained are similar to those obtained for the visible in-edge scenario, but NRMSEs tend to be larger.
Figure~\ref{fig:grid-v0}
shows typical results for different DUFS parameters, represented by two datasets
(also shown in the previous figure).
Once again, the intuition for uniform edge sampling holds for the distribution head: decreasing $b$ and
increasing $w$ yield more accurate estimates for the smallest out-degrees.
While $b=1$ results in poor estimates for the largest out-degrees, our intuition regarding $w$ does not hold true
for the tail. More precisely, in most cases $w=1$ outperforms $w=0.1$ (one exception being dataset soc-Epinions1).
As opposed to the visible in-edge scenario, increasing $b$ tends to provide more accurate tail estimates for $w = 1$.
We investigate this effect in Section~\ref{sec:dresults}.
We find that, for a fixed $w$, larger values of $b$ make the random walks jump more often, moving them
from small volume components to large volume components, yielding better tail estimates.

\techreport{}{
\begin{figure}[]
 %   \subfloat[flickr-links\label{fig:flickr-links-grid}]{%
      \includegraphics[width=0.96\textwidth]{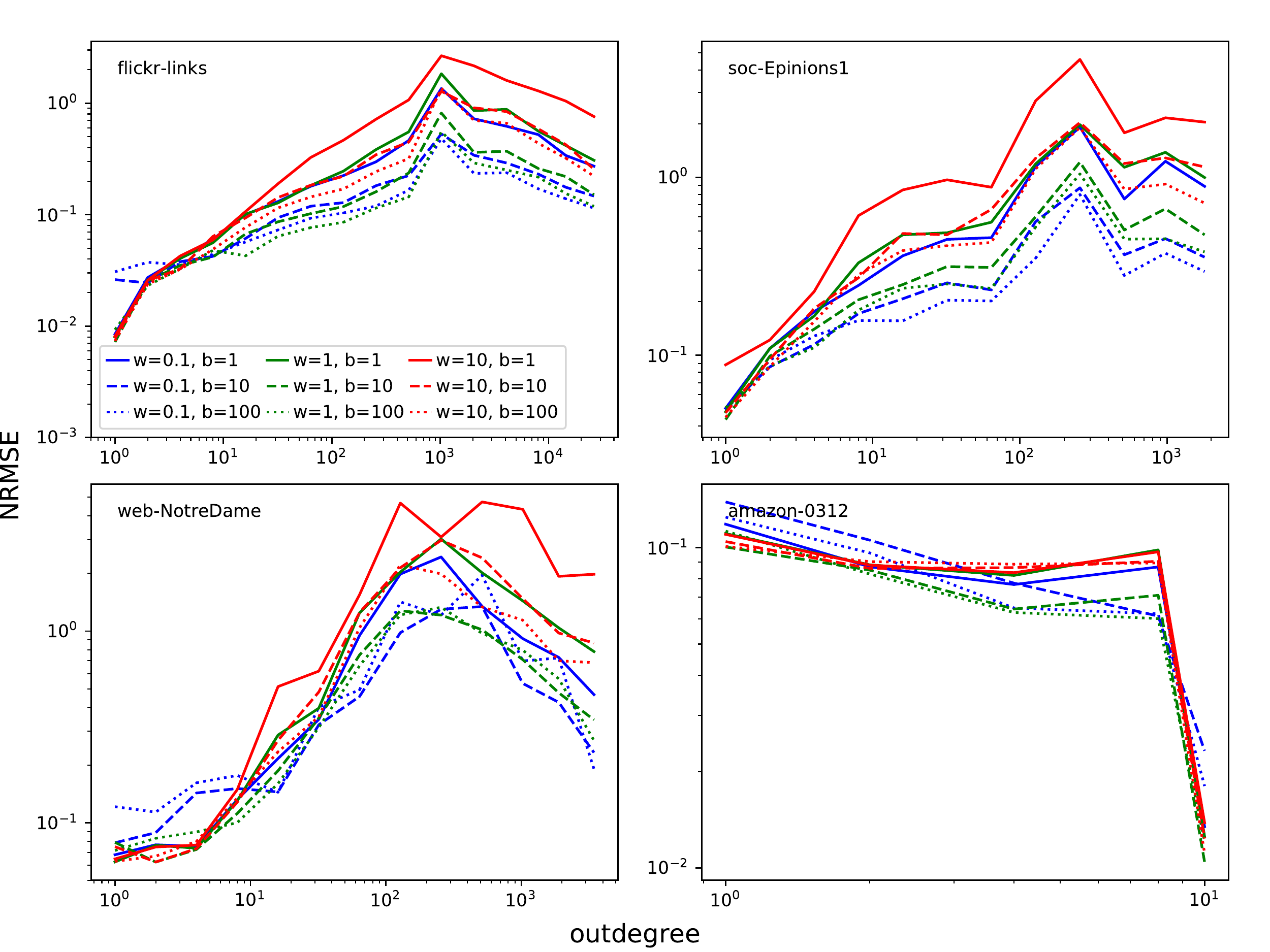}
%    }
%    \subfloat[soc-Epinions1\label{fig:soc-Epinions1-grid}]{%
%      \includegraphics[width=0.49\textwidth]{soc-Epinions1-grid.pdf}
%    }\\
%    \subfloat[web-NotreDame\label{fig:web-NotreDame-grid}]{%
%      \includegraphics[width=0.49\textwidth]{web-NotreDame-grid.pdf}
%    }
%     \subfloat[amazon-0312\label{fig:amazon-0312-grid}]{%
%      \includegraphics[width=0.49\textwidth]{amazon-0312-grid.pdf}
%    }
    \caption{(visible in-edges, $c=10$) Effect of DUFS parameters on datasets with many connected components, when $B=0.1|V|$ and $c=1$. Legend shows the average budget per walker ($b$)
    and jump weight ($w$). Trade-off shows that configurations that result in many uniform node samples, such as $(w = 10, b=1)$,
    yield accurate head estimates, whereas configurations such as $(w =1, b=10)$ yield accurate tail estimates.}
    \label{fig:grid-c10-v1}
 \end{figure}
\begin{figure}[]
      \includegraphics[width=0.96\textwidth]{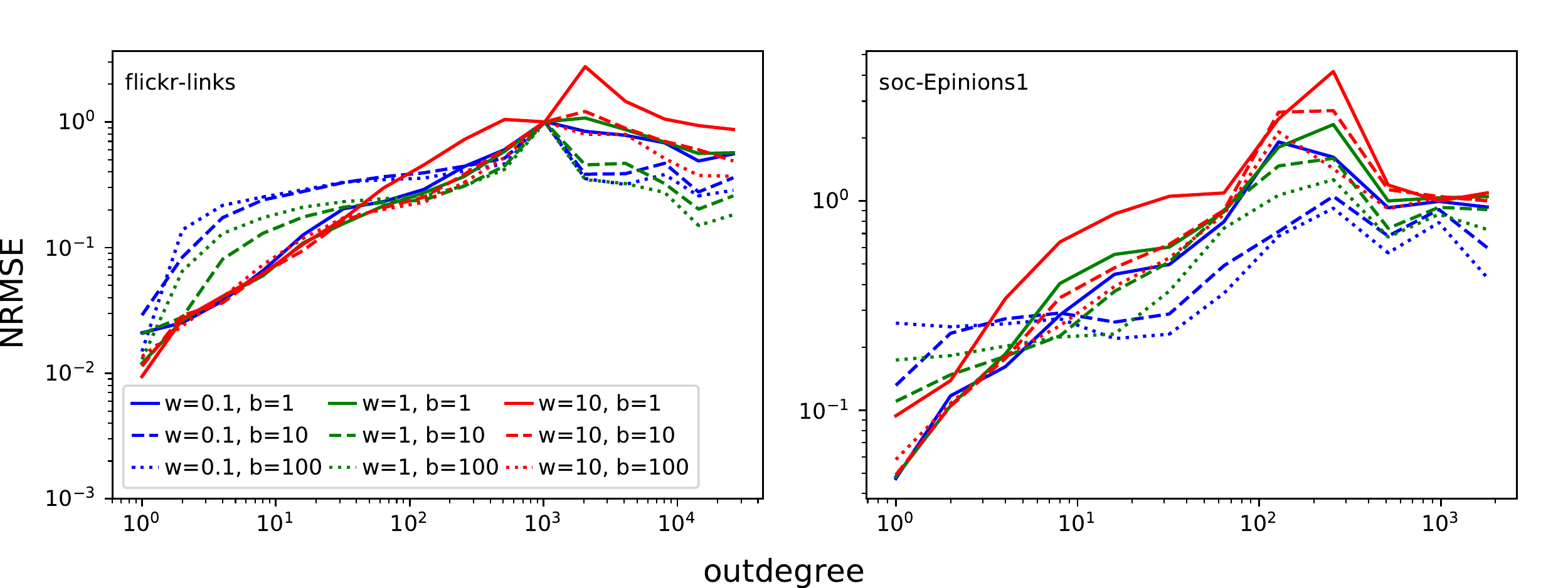}
%    \subfloat[flickr-links\label{fig:flickr-links-grid-v0}]{%
%      \includegraphics[width=0.49\textwidth]{flickr-links-grid-v0.pdf}
%    }
%    \subfloat[soc-Epinions1\label{fig:soc-Epinions1-grid-v0}]{%
%      \includegraphics[width=0.49\textwidth]{soc-Epinions1-grid-v0.pdf}
%    }
    \caption{(invisible in-edges, $c=10$) Effect of DUFS parameters on datasets with many connected components, when $B=0.1|V|$ and $c=1$. Legend shows the average budget per walker ($b$)
    and jump weight ($w$). Configurations that result in many walkers which jump too often, such as $(w\geq10, b = 1)$
    yield accurate head estimates, whereas configurations such as $(w=1, b=10^3)$, yield accurate tail estimates.}
    \label{fig:grid-c10-v0}
 \end{figure}
 }

\subsubsection*{Visible in-edges, $c=10$}
 
 Consider the case where the cost of obtaining uniform node samples is large,
 more precisely, 10 times larger than the cost of moving a walker. \techreport{Plots for this setting can be found in our technical report \cite{Murai18}.}{Figure~\ref{fig:grid-c10-v1} shows typical results for this setting.}
% Using large values of $w$ or small values of $b$ is prohibitive,
% as frequent random jumps and the cost to initialize many walkers would greatly reduce the total number of samples.
It is no longer clear
 that using many walkers and frequent random jumps achieves the most accurate head
 estimates, as this could rapidly deplete the budget.
 In fact, we observe that setting $w=10$ or $b=1$ yields
 poor estimates for both the smallest and largest out-degrees.
 While increasing the jump weight $w$ or decreasing $b$  sometimes improves estimates in the head,
 it rarely does so in the tail.
 The best results  for the smallest out-degrees are often observed when setting $w=1$ and $b=10$ or $10^2$.
 On the other hand, setting $(w=0.1,b =10^3)$ or $(w=1,b =10^2)$ usually achieves relatively small NRMSEs
 for the largest out-degree estimates.

\subsubsection*{Invisible in-edges, $c=10$}

 \techreport{Plots for this setting can be found in our technical report \cite{Murai18}.}{Figure~\ref{fig:grid-c10-v0} shows typical results for this setting.}
Unlike the scenario with visible in-edges, setting $w=10$ and $b=1$ often produces the most
accurate estimates for the smallest out-degrees. This is because many of the datasets
have nodes with no out-edges; these nodes can only be reached through a
neighbor or through random node sampling.
Conversely, the general trends for tail estimates are similar to those observed for the
visible in-edges case: large values of $w$ and small values of $b$ yield less accurate estimates for
the largest out-degree values. For $w=1$, however,
$b=10^2$ often outperforms $b=10^3$.
%In this case, the gains from obtaining samples whose distribution is closer to uniform outweigh the smaller number of samples.
On the other hand, for $w=0.1$ there is little difference in the estimates for different values of $b$.

\subsection{Evaluation of DUFS in the visible in-edges scenario}\label{sec:visible}

In this section we compare two variants of Directed Unbiased Frontier Sampling: E-DUFS, which uses the edge-based estimator
and DUFS, which uses the hybrid estimator, to each other and to a single random walk (SingleRW) and
multiple independent random walks (MultiRW). \fm{We do not include Frontier Sampling in the comparison as it is a special case of DUFS where $w = 0$ and we know from Section~\ref{sec:parameters} that allowing random jumps effectively reduce estimation errors.}
%We also contrast DUFS with random node and random edge sampling.
%Simulations consist of executing these sampling methods on the same network datasets mentioned
%in the beginning of Section~\ref{sec:results}.

% %The datasets used in the simulations are summarized in Table~\ref{tab:graphs}:
% Users are represented as vertices of a graph.
% In these websites a user can subscribe to other user updates; an edge $(u,v)$ exists between users $u$ and $v$ if user $u$ subscribes to user $v$. 
% At ``Livejournal'' and ``YouTube'' it is possible to query the incoming and outgoing edges of a given user.
% ``Internet RLT'' is a router-level Internet graph collected from traceroute measurements of 23 monitors distributed over the world~\cite{CAIDA}.
% Note that some of these graphs contain disconnected components (subgraphs).

\subsubsection{Out-degree and in-degree distribution estimates} \label{sec:outdegdist}

Here we focus on estimating the marginal in- and out-degree distributions.
Each simulation consists of 1000 runs from which we compute the empirical NRMSE.
%The \NRMSE is used to compare the accuracy of the estimates obtained from SingleRW, MultiRW, FS and MLE.
For MultiRW, E-DUFS and DUFS we set the average budget per walker to be $b=10$.
For conciseness, we only show a few representative results.
\begin{figure}[]
      \includegraphics[width=0.96\textwidth]{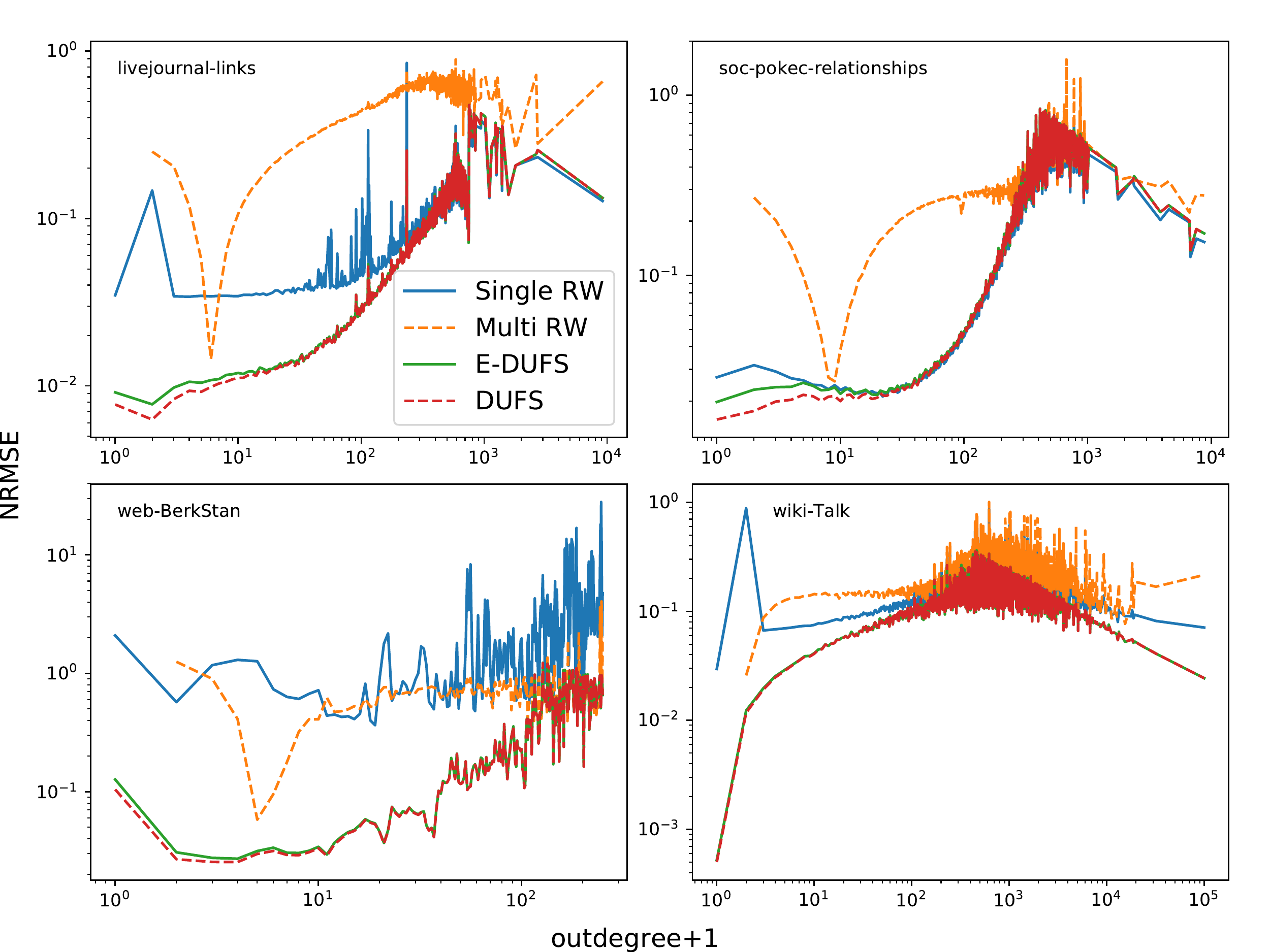}
%    \subfloat[livejournal-links\label{fig:livejournal-links-rw}]{%
%      \includegraphics[width=0.49\textwidth]{livejournal-links-rw.pdf}
%    }
%    \subfloat[soc-pokec-relationships\label{fig:soc-pokec-relationships-rw}]{%
%      \includegraphics[width=0.49\textwidth]{soc-pokec-relationships-rw.pdf}
%    }\\
%    \subfloat[web-BerkStan\label{fig:web-BerkStan}]{%
%      \includegraphics[width=0.49\textwidth]{web-BerkStan-rw.pdf}
%    }
%    \subfloat[wiki-Talk\label{fig:wiki-Talk-rw}]{%
%      \includegraphics[width=0.49\textwidth]{wiki-Talk-rw.pdf}
%    }
    \caption{Comparison of single random walk (SingleRW), multiple independent random walks (MultiRW),
    DUFS with edge-based estimator (E-DUFS) and with hybrid estimator (DUFS).
    MultiRW yields the worst results, as the edge sampling probability is not the same across different connected
    components. Both DUFS variants outperform SingleRW, but DUFS is slightly more accurate
    in the head.}
    \label{fig:und-rw}
 \end{figure}

Figure~\ref{fig:und-rw} shows typical results obtained when using SingleRW, MultiRW, E-DUFS and DUFS to estimate out-degree distributions on the datasets. In 8 out of 15 datasets, MultiRW yields much larger NRMSEs than does the SingleRW. As pointed out in \cite[Section~4.5]{TechReport}, this is due to the fact that the estimator in~\eqref{eq:hattheta} assumes that all edges are sampled with the same probability. This assumption is violated by MultiRW because the
stationary sampling probability depends on the size of the connected component within which each walker is located.
E-DUFS estimates are consistently more accurate than those of MultiRW and SingleRW, except on datasets where the
original graph and its LCC have similar out-degree distributions.
In some of these cases SingleRW slightly outperforms E-DUFS in the tail (see top-right fig.).
DUFS, in turn, outperforms E-DUFS in the head of the out-degree distribution and has similar performance when
estimating other out-degree values. For this reason, defining the estimation task in terms of the CCDF would give
DUFS an unfair advantage.

When restricted to the largest connected component, the performance differences between SingleRW and E-DUFS
and those between SingleRW and DUFS become smaller, for $B=0.1|V|$.
 %
%We now compare MLE against (uniform) node sampling and (uniform) edge sampling. We set the sampling cost of random node sampling to one and random edge sampling to two (as each edge samples two vertices).
%Figure~\ref{fig:und-indep} depicts representative results of those obtained when using node sampling, edge sampling and MLE. 
%In all datasets but amazon-0312, node sampling is outperformed by MLE for out-degrees at least slightly larger than the average. The amazon-0312 dataset is a special case because the domain size is very small: out-degree is capped at 10. As a result, 63\% of the nodes have out-degree 10 and the average out-degree is 8.
%In 8 out of 10 datasets MLE performance dominates that of edge sampling. In the remaining datasets,
%edge sampling is slightly better than MLE at the body or exhibits similar performance at the tail of the
%out-degree distribution. \fm{include the case with smaller hit ratio?}
%\begin{figure}[]
%    \subfloat[soc-Slashdot0902\label{fig:soc-Slashdot0902-indep}]{%
%      \includegraphics[width=0.49\textwidth]{soc-Slashdot0902-indep.pdf}
%    }
%    \subfloat[web-Google\label{fig:web-Google-indep}]{%
%      \includegraphics[width=0.49\textwidth]{web-Google-indep.pdf}
%    }
%    \caption{Comparison of random node sampling, random edge sampling and
%    Frontier Sampling with hybrid estimator (MLE).}
%    \label{fig:und-indep}
% \end{figure}
Results for in-degree distribution estimation are qualitatively similar and are omitted.

 \subsubsection{Joint in- and out-degree distributions}
 
% Figure~\ref{fig:combined_gains} shows typical reductions in NRMSE obtained with the combined estimator.
%\fm{replace this figure by out-degree distribution estimation.}
%The colormap corresponds to the ratio between the NRMSE for the combined
%and NRMSE for the original estimator. We observed across several datasets that
%this ratio often reaches $10^{-0.2} = 0.63$ and is almost always below $1.0$ for all points on the heatmap,
%indicating that the combined estimator consistently outperforms the original estimator.
%%
%\begin{figure}
%\centerline{\includegraphics[width=0.8\textwidth]{soc-pokec-logratio.pdf}}
%\caption{NRMSE ratio between combined estimator and original estimator when
%estimating joint distribution on the Pokec Social Network graph \protect\cite{SNAP} (for $B=0.1|V|$).}
%\label{fig:combined_gains}
%\end{figure}
 
%We evaluate DUFS accuracy
%when estimating joint in- and out-degree distributions on a large variety of real networks that
%are directed in nature.
%
   \begin{figure}[]
   \centering
    \subfloat[flickr-links\label{fig:flickr-links-joint}]{%
      \includegraphics[width=\textwidth]{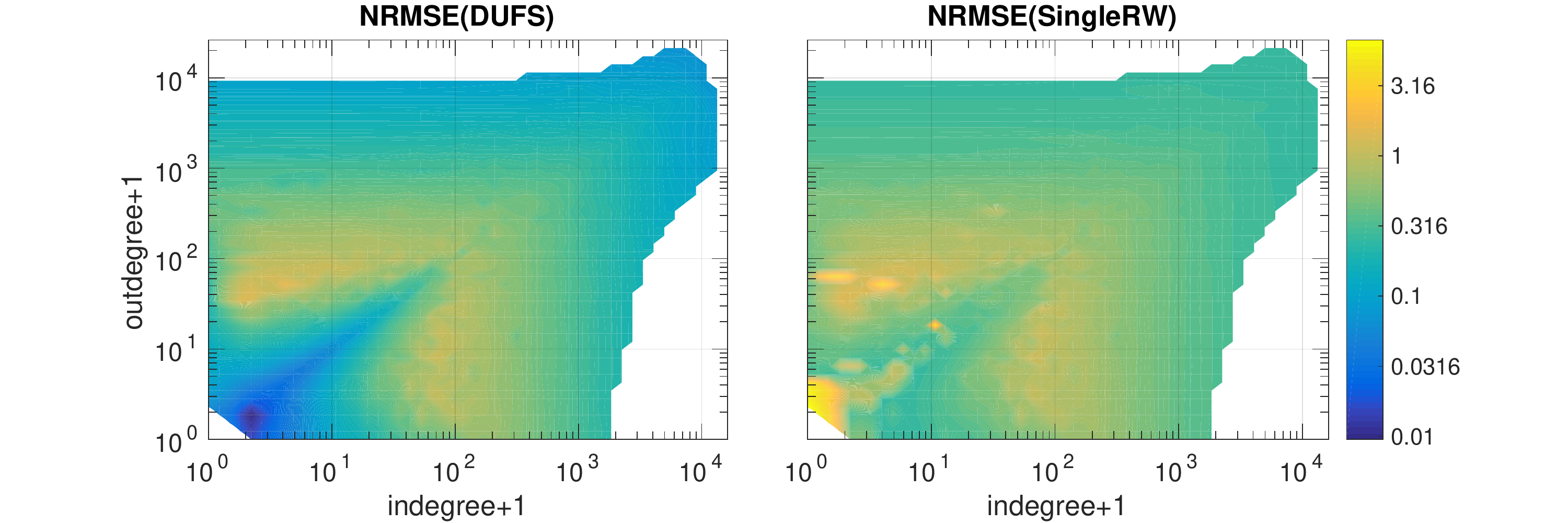}
    }\\
    \subfloat[youtube-links\label{fig:youtube-links-joint}]{%
      \includegraphics[width=\textwidth]{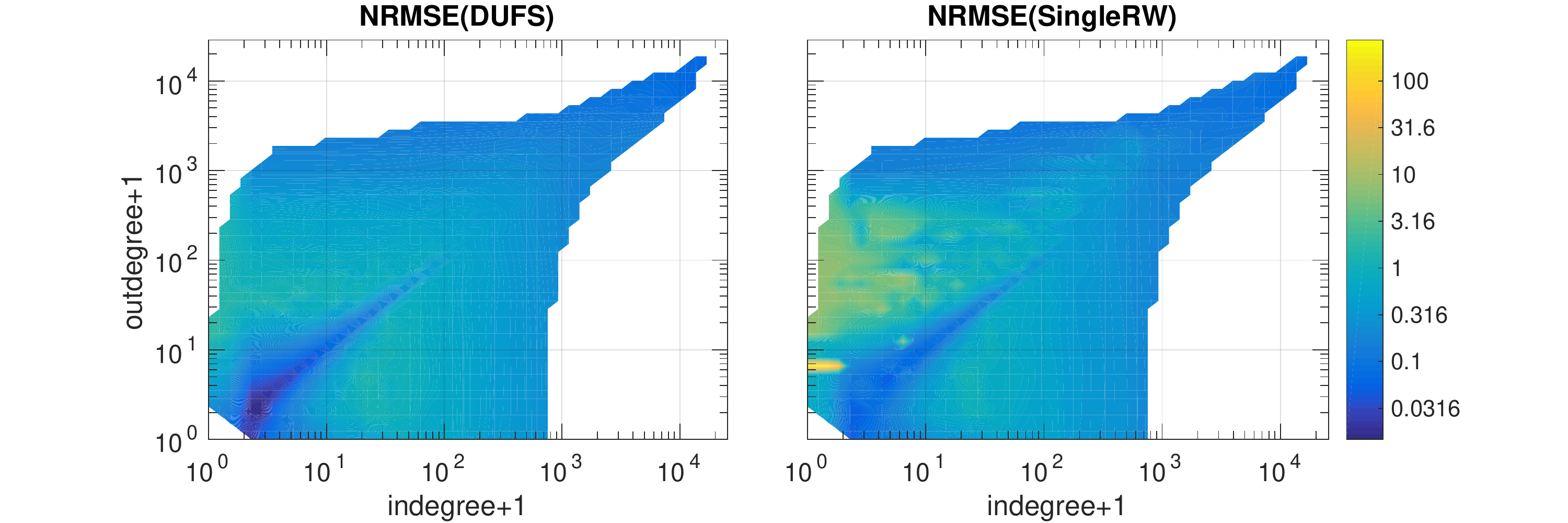}
    }\\
%    \subfloat[amazon-0312\label{fig:amazon-0312-joint}]{%
%      \includegraphics[width=0.87\textwidth]{amazon-0312-mlerwsw-jnt.pdf}
%    }\\
    \subfloat[web-Google\label{fig:web-Google-joint}]{%
      \includegraphics[width=\textwidth]{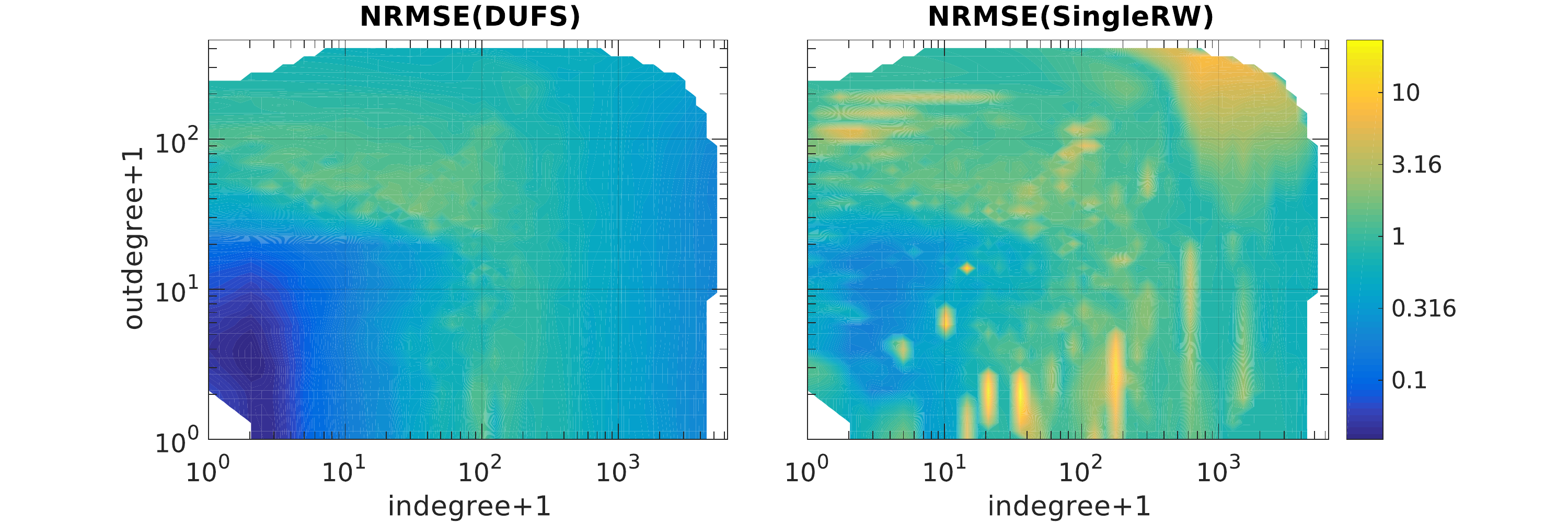}
    }
    \caption{Comparison between DUFS and SingleRW w.r.t. NRMSE when estimating the joint in- and out-degree distribution.
    In most cases SingleRW will exhibit ``hot spots'' (regions with large NRMSE), which are mitigated by DUFS.}
    \label{fig:jnt-mle-vs-rwsw}
 \end{figure}
 
We compare the NRMSEs associated with DUFS and SingleRW
for the estimates of the joint in- and out-degree distribution. We observe that DUFS consistently outperforms SingleRW on all datasets.
On 10 out of 15 datasets, the estimates corresponding to low in-degree and low out-degree exhibit much smaller errors when using DUFS than when using
SingleRW. Furthermore, DUFS also achieves smaller estimation errors for most of the remaining points of the joint distribution in 11 out of 15 datasets.
Figures~\ref{fig:jnt-mle-vs-rwsw}(a-b) show heatmaps corresponding to typical NRMSE results for DUFS and SingleRW.
Interestingly, we note that on the web graph datasets and on the email-EuAll dataset, DUFS outperforms SingleRW by one or two orders of magnitude,
as illustrated by Figure~\ref{fig:jnt-mle-vs-rwsw}(c), which shows the heatmap comparison for dataset web-Google. Although the NRMSE exhibited by SingleRW applied to the LCC datasets is much smaller, the comparison between DUFS and SingleRW is qualitatively similar
and is, therefore, omitted.
%  \begin{figure}[]
%    \subfloat[flickr-links\label{fig:flickr-links-joint}]{%
%      \includegraphics[trim={90 0 120 0},clip,width=0.5\textwidth]{flickr-links-mlerwsw-jnt.pdf}
%    }
%    \subfloat[youtube-links\label{fig:youtube-links-joint}]{%
%      \includegraphics[trim={90 0 120 0},clip,width=0.5\textwidth]{youtube-links-mlerwsw-jnt.pdf}
%    }\\
%    \subfloat[amazon-0312\label{fig:amazon-0312-joint}]{%
%      \includegraphics[trim={90 0 120 0},clip,width=0.5\textwidth]{amazon-0312-mlerwsw-jnt.pdf}
%    }
%    \subfloat[web-NotreDame\label{fig:web-NotreDame-joint}]{%
%      \includegraphics[trim={90 0 120 0},clip,width=0.5\textwidth]{web-NotreDame-mlerwsw-jnt.pdf}
%    }
%    \caption{NRMSE ratio between FS+ and FS of the estimated joint in- and out-degree distribution for various datasets.}
%    \label{fig:joint_undirected}
% \end{figure}

We then investigated the performance gains obtained by using the hybrid estimator instead of the original estimator.
Figures~\ref{fig:jnt-mle-vs-fs}(a-b) show the ratios between the NRMSEs obtained with DUFS (hybrid) to those obtained with
the E-DUFS (original) for two networks. We chose to use the NRMSE ratio (or equivalently, the root MSE ratio)
to make it easier to visualize the differences.
%because the performance
%differences between H-DUFS and E-DUFS are relatively smaller than those between H-DUFS and SingleRW, being difficult to see when plotting heatmaps of the absolute values side-by-side,
%especially when the \NRMSE spans several orders of magnitude, whereas \NRMSE ratios will
%span a smaller interval, making it suitable to use the same colorcode across different plots.
We observe that DUFS consistently outperforms
E-DUFS on all datasets. More precisely, the error ratio is rarely above one and, for points corresponding to small in- and out-degrees, it often lies below 0.9.
Results on most datasets are similar to that depicted in Figure~\ref{fig:jnt-mle-vs-fs}(a), but results on social networks datasets are closer to that shown in 
Figure~\ref{fig:jnt-mle-vs-fs}(b), where large in- and out-degrees also seem to benefit from the information contained in the walkers' initial locations. Results
for the LCC datasets are qualitatively similar, with accuracy gains from the hybrid estimator slightly larger on these datasets than on the original datasets.
  \begin{figure}[]
  \centering
    \subfloat[p2p-Gnutella31\label{fig:p2p-Gnutella31-joint}]{%
      \includegraphics[width=0.45\textwidth]{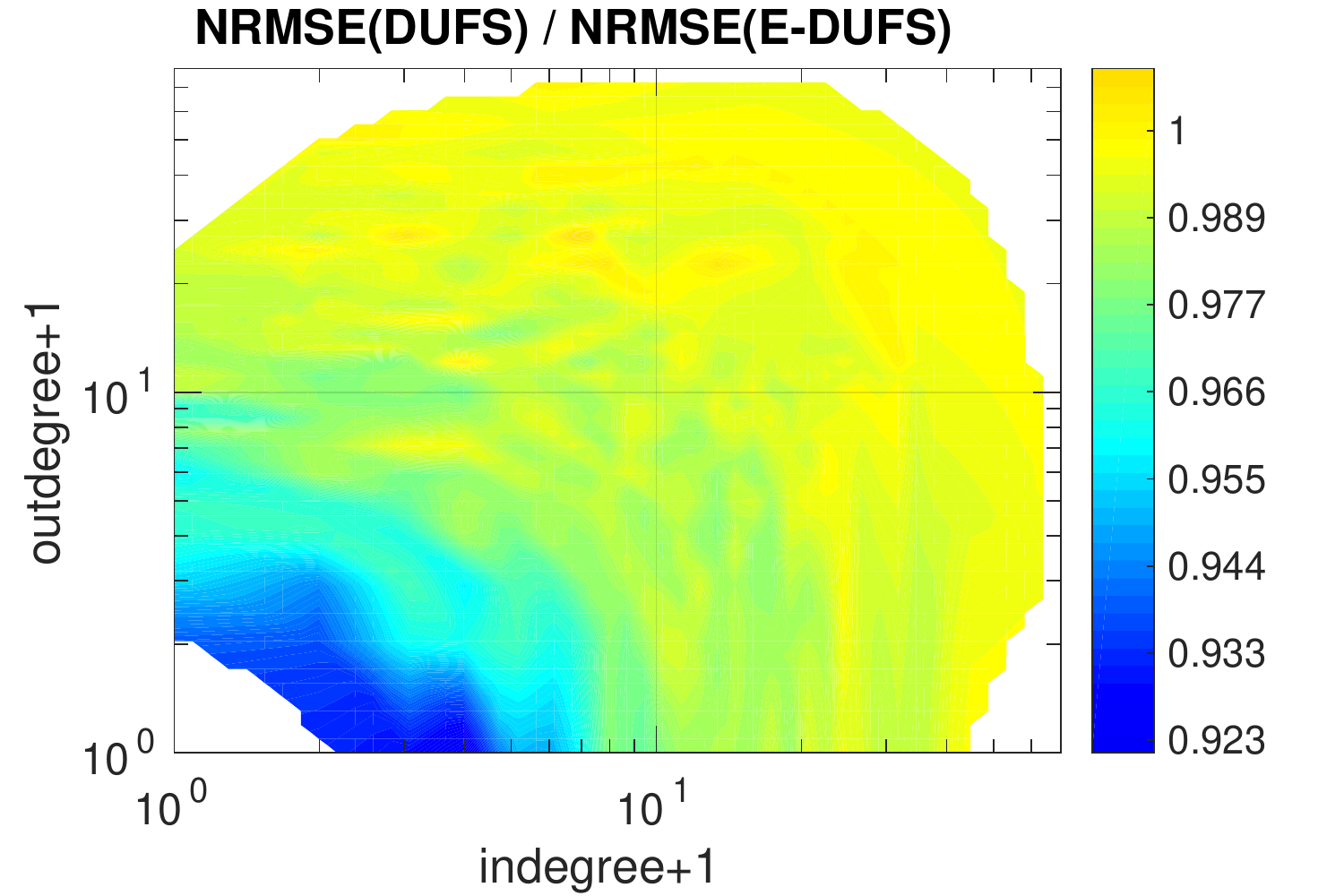}
      } \subfloat[soc-Slashdot0902\label{fig:soc-Slashdot0902-joint}]{%
      \includegraphics[width=0.45\textwidth]{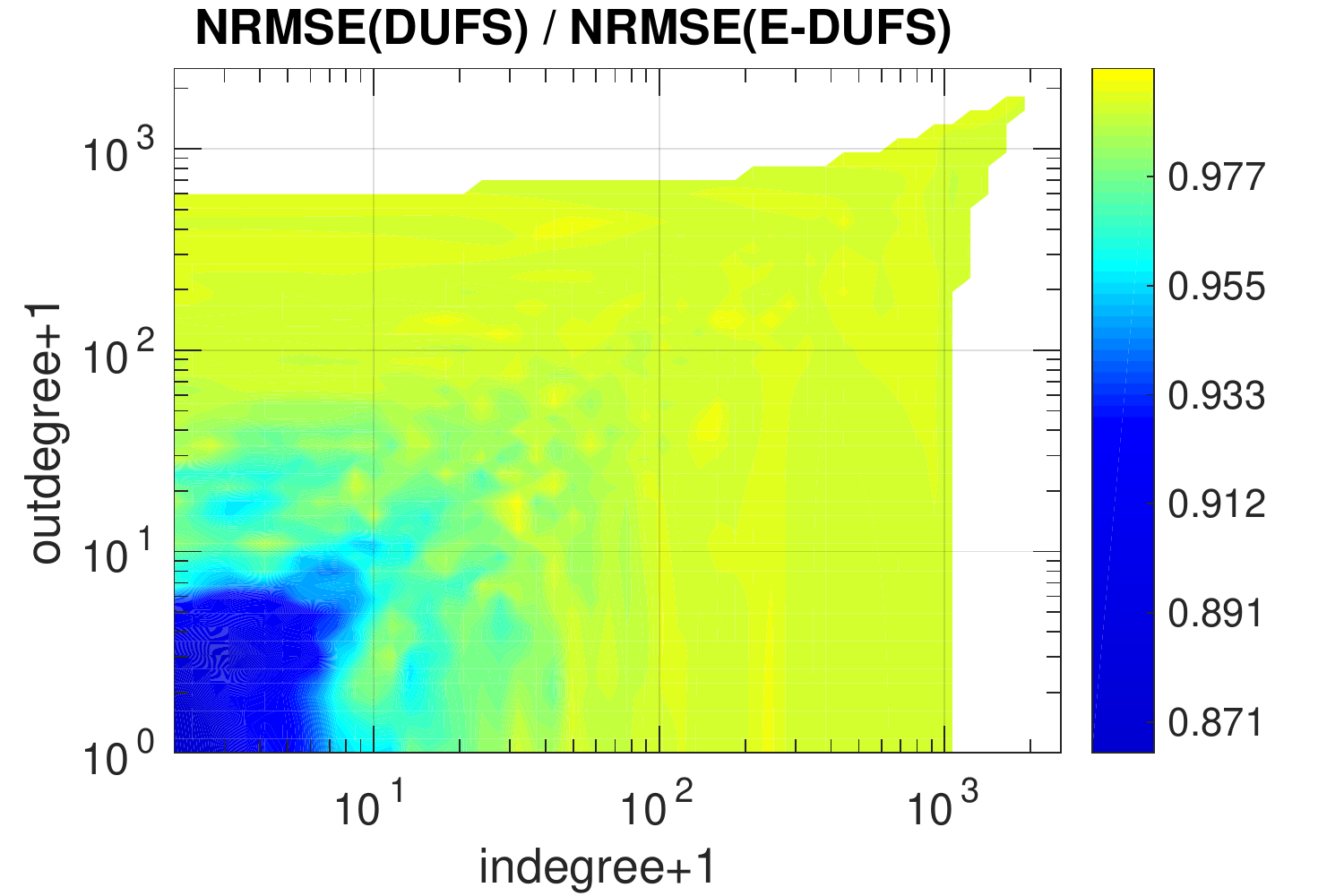}
    }
     \caption{NRMSE ratios between DUFS and E-DUFS of the estimated joint in- and out-degree distribution for two datasets.
     DUFS is typically better than E-DUFS at low in and out-degree regions ({\bf left}), but in social network graphs presented
     improvements over most of the joint distribution ({\bf right}).}
    \label{fig:jnt-mle-vs-fs}
 \end{figure}% Activate the following line by filling in the right side. If for example the name of the root file is Main.tex, write
% "...root = Main.tex" if the chapter file is in the same directory, and "...root = ../Main.tex" if the chapter is in a subdirectory.
 
%!TEX root =  ../frontier-tkdd.tex

\subsection{Evaluation of DUFS in the invisible in-edges scenario}\label{sec:dresults}

In this section, we compare the NRMSEs associated with DUFS and Directed Unbiased Random Walk (DURW) method when
estimating out-degree distributions in the case where in-edges are not directly observable.
We note that DURW is known to outperform a reference method for this scenario proposed in~\cite{Yossef2008}.
For a comparison between DURW and this reference method, please refer to~\cite{RibeiroINFOCOM2012}.

As we mentioned in Section~\ref{sec:parameters}, DURW results are similar to those obtained with DUFS when
the budget per walker $b$ is large, since DURW is a special case of DUFS where $b=B-c$. Therefore, we focus
on comparing DUFS for small values of $b$ and DURW, when the total number of uniform
node samples collected by each method is roughly the same. More precisely, we simulate DUFS for $b=10$
and $w=1$ and set the DURW parameter $w$ so that the number of node samples
differs by at most 1\% (averaged over 1000 runs). This aims to provide a fair comparison between these methods.

%Note that both DUFS and DURW require random node sampling.
%While the number of walkers determines how many node samples are needed to initialize DUFS,
%subsequent random jumps performed by DUFS and DURW deem the total number of
%random node samples a random variable, which depends on the graph topology, random jump
%weight, number of walkers and even the order in which nodes are visited. For a fair comparison,
%we take a given configuration $(w,b)$ for DUFS and vary the jump weight parameter $w^\prime$
%in DURW so that the number of random vertices sampled by both methods differs by 1\% at most (averaged over 100 runs).

We find that neither of the two methods consistently outperforms the other over all datasets.
The extra random jumps performed
by DURW will prevent the walker from spending much of the budget in small
volume components. As a result, DURW tends to exhibit larger errors in the
head but smaller errors in the tail of the out-degree distribution than DUFS.
Figure~\ref{fig:dufs-vs-durw} show typical results for $w=1$ and $b=10$.
DUFS exhibited lower estimation errors in the head of the distribution on 11 datasets,
being outperformed by DURW on one dataset and displaying comparable performance on the others. 
In 6 out of 15 datasets, DURW had better performance in the tail,
while DUFS yielded better results on other five datasets. 
Results for $w=1$ and $b \in \{10^2, 10^3\}$ are similar and are,
therefore, omitted. As $b$ increases, differences between DUFS and
DURW start to vanish.
  \begin{figure}[]
  \centering
     \includegraphics[width=0.98\textwidth]{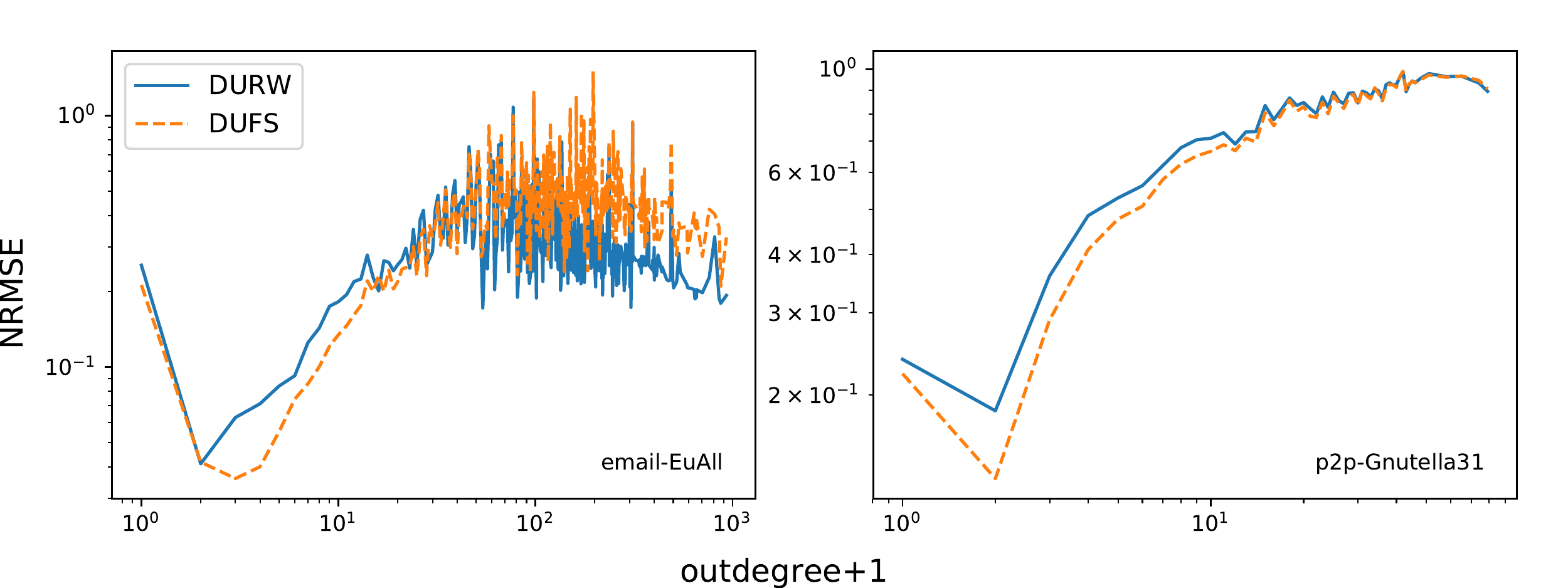}
%    \subfloat[email-EuAll\label{fig:email-EuAll-DUFS}]{%
%      \includegraphics[width=0.5\textwidth]{email-EuAll-DUFS.pdf}
%      } \subfloat[p2p-Gnutella31\label{fig:p2p-Gnutella31-DUFS}]{%
%      \includegraphics[width=0.5\textwidth]{p2p-Gnutella31-DUFS.pdf}
%    }
     \caption{NRMSEs associated with DUFS ($w=1$, $b=10$) and DURW ($w^\prime$ chosen to match average number of node samples)
     when estimating out-degree distribution. DURW performs more random jumps, thus better avoiding
     small volume components. This improves DURW results in the tail, but often results in lower accuracy in the head ({\bf left}).
     In one third of the datasets, DUFS yielded similar or better results than DURW
     over most out-degree points ({\bf right}).}
    \label{fig:dufs-vs-durw}
 \end{figure}

%We observe that setting the jump weight $w=$ and the budget per walker to $b=$ in DUFS works well for most datasets.
%Interestingly, the same values of jump weight and budget per walker were found to work well respectively in DURW and
%FS (in the case where in-edges are visible).

  \begin{figure}[ht!]
  \centering
        \includegraphics[width=0.98\textwidth]{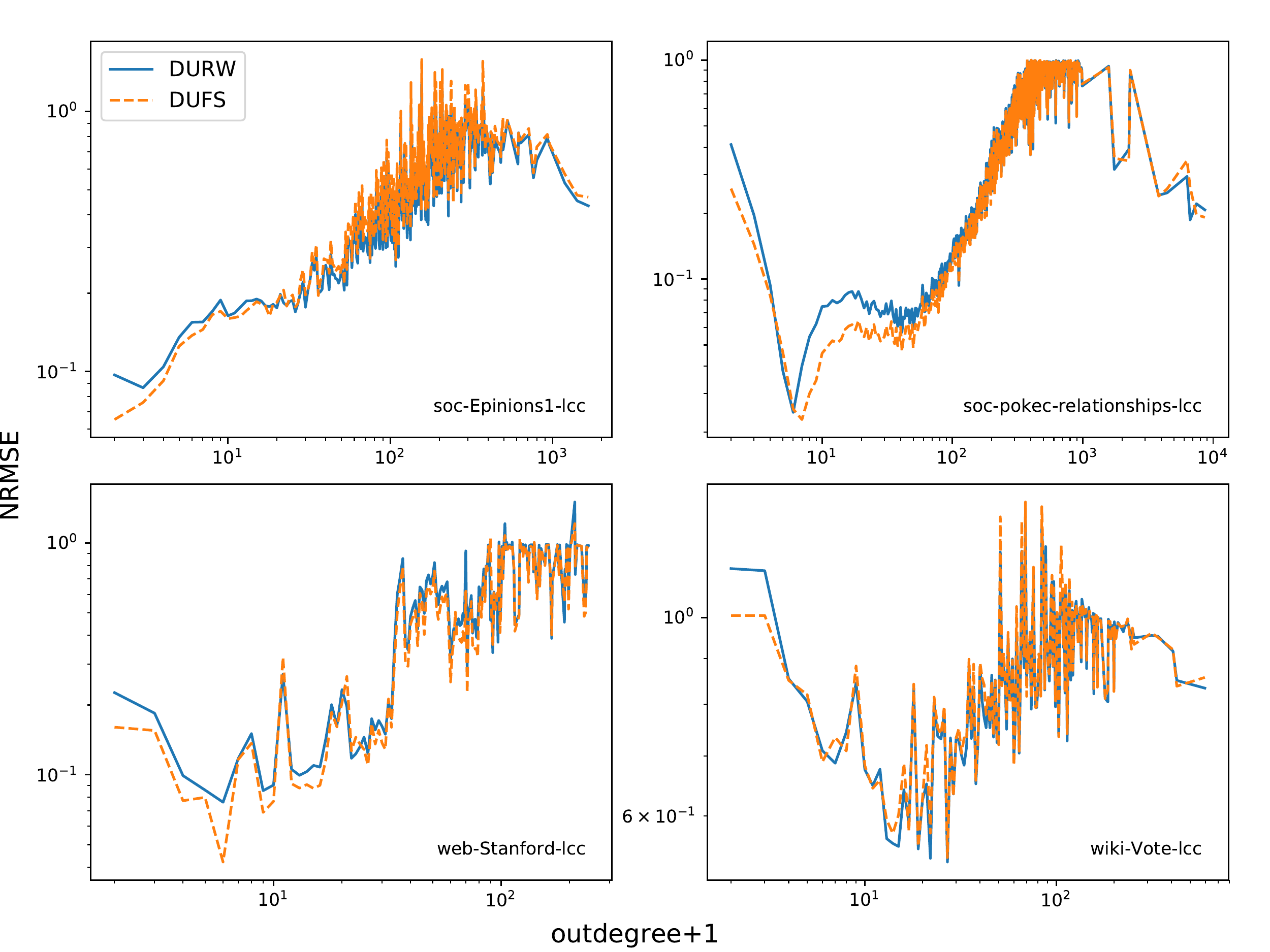}
%    \subfloat[soc-Epinions1-lcc\label{fig:soc-Epinions1-lcc-DUFS}]{%
%      \includegraphics[width=0.5\textwidth]{soc-Epinions1-lcc-DUFS.pdf}
%      } \subfloat[soc-pokec-relationships-lcc\label{fig:soc-pokec-relationships-lcc-DUFS}]{%
%      \includegraphics[width=0.5\textwidth]{soc-pokec-relationships-lcc-DUFS.pdf}
%    }\\
%        \subfloat[web-Stanford-lcc\label{fig:web-Stanford-lcc-DUFS}]{%
%      \includegraphics[width=0.5\textwidth]{web-Stanford-lcc-DUFS.pdf}
%      } \subfloat[wiki-Vote-lcc\label{fig:wiki-Vote-lcc-DUFS}]{%
%      \includegraphics[width=0.5\textwidth]{wiki-Vote-lcc-DUFS.pdf}
%    }
     \caption{NRMSEs associated with DUFS ($w=1$, $b=10$) and DURW ($w^\prime$ chosen to match average number of node samples)
     when estimating out-degree distribution.}
    \label{fig:dufs-vs-durw-lcc}
 \end{figure}

To better understand the impact of multiple connected components in DUFS and DURW performances, we simulate each method
 on the largest strongly connected component of each dataset
(i.e., on the LCC datasets). Figure~\ref{fig:dufs-vs-durw-lcc} shows typical results among the LCC datasets.
In most networks, DUFS yields smaller NRMSE than DURW in the head and yield similar results in the tail.
Once again, for larger $b$ the performances of DUFS and DURW become equivalent.

 \subsection{Relationship between NRMSE and out-degree distribution}\label{dufs:behavior}

Throughout Section~\ref{sec:results} we observed that the NRMSE associated with
RW-based methods tends to increase with out-degree up to a certain out-degree and to
decrease after that. Moreover, for some out-degree ranges the $\log $ NRMSE
seems to vary linearly with the $\log$ out-degree. %(see, for instance,
Figure~\ref{fig:grid}).  For simplicity, we discuss the undirected graph case,
but the extension to directed graphs is straightforward.  The RW methods
discussed here combine uniform node sampling with RW sampling approximated as
uniform edge sampling. For simplicity, we analyze below the accuracy of uniform
node and uniform edge sampling. We assume that each sampled edge produces one observation, obtained by retrieving the set of labels associated
with one of the adjacent vertices chosen equiprobably. Therefore both node
sampling and edge sampling collect node labels. 

Let
$\bbS = \{s_1,\ldots,s_B\}$ be the sequence of sampled vertices.
For uniform node sampling, the probability of observing a given label $\ell$ in $\cL(s_i)$ is $\theta_\ell$, for any $i=1,\ldots,B$.
The minimum variance unbiased estimator of $\theta_\ell$ is
\begin{equation}\label{eq:uvs_estimator}
 T_\textrm{vs}^\ell(\bbS) = \frac{1}{B} \sum_{i=1}^B {\mathds{1}}\{\ell \in \cL(s_i)\}.
\end{equation}
Note that the summation in~\eqref{eq:uvs_estimator} is binomially distributed with parameters $B$ and $\theta_\ell$.
It follows that the mean squared error (MSE) %-- a typical metric to assess the accuracy of estimates --
of $T^\ell_\textrm{vs}(\bbS)$ is
\begin{align}
  \MSE(T^\ell_\textrm{vs}(\bbS)) &= E[(T^\ell_\textrm{vs}(\bbS) - \theta_\ell)^2], \nonumber \\ 
   & = \frac{1}{ B } \theta_\ell(1-\theta_\ell) \label{eq:mse_vs}.
 \end{align}
 
 For uniform edge sampling, the probability of observing a given label $\ell \in \cL$ in the sample $\cL(s_i)$ for $i=1,\ldots,B$, equals
 $$\pi_\ell = \frac{\sum_{v \in V} {\mathds{1}}\{\ell \in \cL(v)\} \deg(v)  }{ \sum_{u \in V} \deg(u) }.$$
 In that case, the following
 estimator can be shown to be asymptotically unbiased
 \begin{equation}\label{eq:ues}
 T^\ell_\textrm{es}(\bbS) = \frac{1}{B} \frac{\sum_{k=1}^B {\mathds{1}}\{\ell \in \cL(s_k)\} \deg^{-1}(s_k)}{\sum_{j=1}^B  \deg^{-1}(s_j)}.
\end{equation}

In particular, when node labels are the undirected degrees of each node, the probability of observing a given degree $d$ becomes $\pi_d = d \theta_d / \bar{d}$, where $\bar{d}$ is the average undirected degree. The estimator for $B=1$ reduces to $T^d_\textrm{es}(\cS_1) = {\mathds{1}}\{s_1 = d\}$, which is a random variable distributed according to a Bernoulli with parameter $\pi_d$. As a result, the MSE for $B>1$ independent samples is
\begin{eqnarray} 
  \MSE(T^d_\textrm{es}(\bbS)) & = & \frac{1}{ B } \pi_d(1-\pi_d) = \frac{1}{ B } \frac{d\theta_d}{\bar d}\left(1-\frac{d\theta_d}{\bar d} \right)  \label{eq:mse_es} %\\
                                             %& = & \frac{d \theta_d(d-d \theta_d)}{ d^2 B }.
  \end{eqnarray}

  \begin{figure}
\centerline{\includegraphics[width=0.52\textwidth]{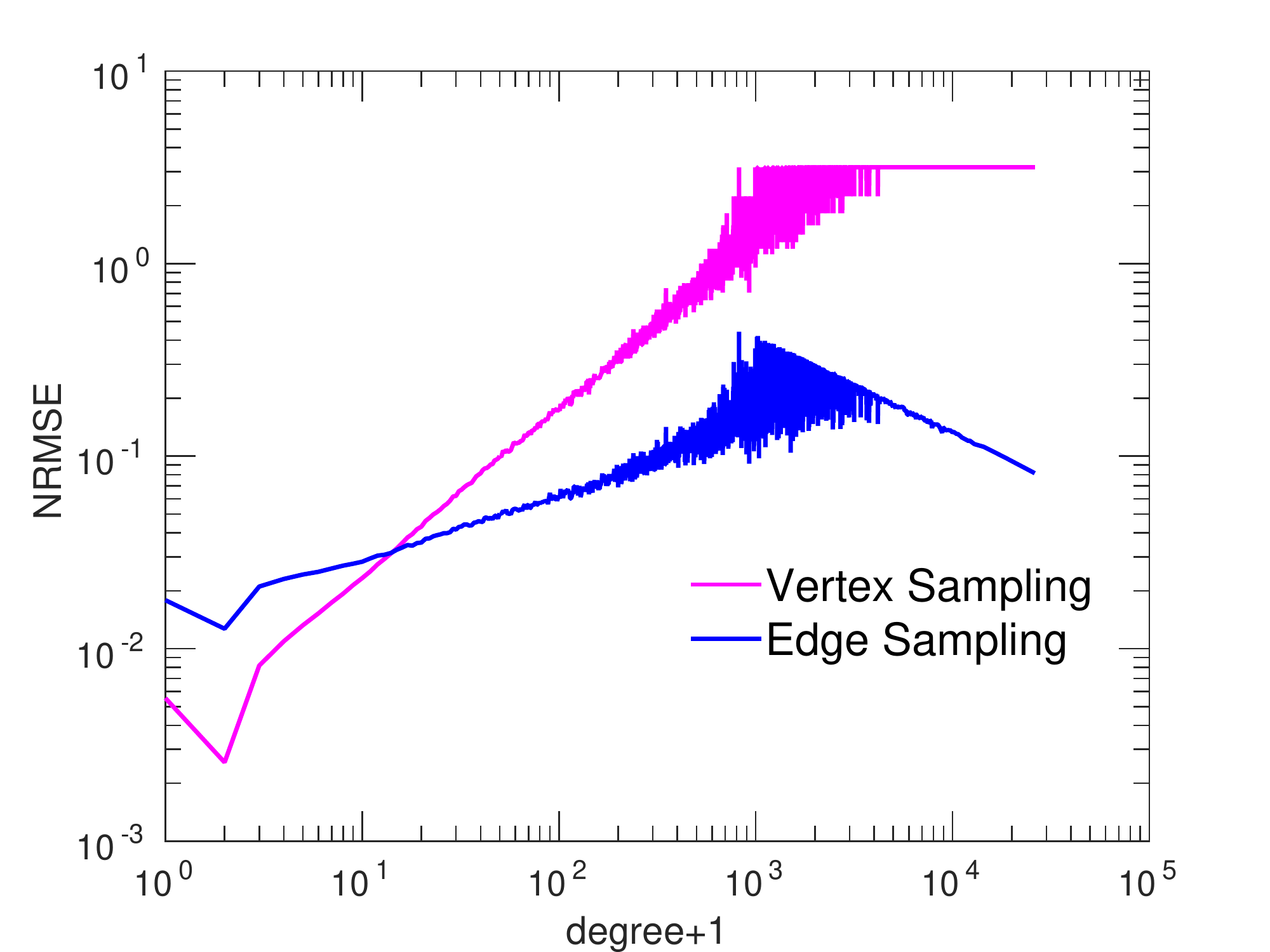}}
\caption{NRMSE associated with uniform node sampling and uniform edge sampling when
estimating degree distribution of the Flickr dataset (for $B=0.1|V|$).}
\label{fig:uniform}
\end{figure}

  Equations~\eqref{eq:mse_vs}~and~\eqref{eq:mse_es} characterize the conditions under which each sampling model is more accurate. More precisely, for all $i$ such that $\theta_d > \pi_d$ (or equivalently, $d < \bar{d}$), uniform node sampling yields better estimates than uniform edge sampling. This dichotomy is illustrated in Figure~\ref{fig:uniform}, which
  shows the NRMSE associated with degree distribution estimates resulting from each sampling model on the flickr-links dataset, for $B = 0.1|V|$.

Note that in log-log scale, both curves resemble a straight line for $d=2,\ldots,10^3$, which indicates a power law. For degrees larger than $5\times 10^3$,
the NRMSE associated with node sampling is constant, while the NRMSE associated with edge sampling decreases linearly with the degree.
We show that these observations are direct consequences of the fact that the degree distribution in this network (as well as many other real networks) approximately follows a power law distribution. However, the degree distribution of a finite network cannot be an exact power law distribution because the tail is
truncated.
As a result, most of the largest degree values are observed exactly once. This can be seen in Figure~\ref{fig:outdegree} by noticing that on the flickr-links (and many other datasets) the p.m.f.\ is constant for the largest out-degrees. Assume, for instance, that the degree distribution can be modeled as
 $$
 \theta_d = \begin{cases}
 d^{-\beta} / Z\,, & 1 \leq d \leq \tau \\
 1/|V|\,, & d > \tau,
 \end{cases}
 $$ for some $\beta \geq 1$ and some normalizing constant $Z$.

 From~\eqref{eq:mse_vs}, we have for uniform node sampling,
 \begin{equation}\label{eq:nrmse_vs}
 \mbox{NRMSE}(T^d_\textrm{vs}(\bbS)) = \sqrt{(1/\theta_d - 1)/B}.
 \end{equation}
 For $\theta_d \ll 1$, this implies
  \begin{equation*}
  \mbox{NRMSE}(T^d_\textrm{vs}(\bbS)) \approx \begin{cases}
 \sqrt{Zd^\beta/B}\,, & 1 \leq d \leq \tau \\
 \sqrt{|V|/B}\,, & d > \tau.
 \end{cases}
  \end{equation*}
  For $d > \tau$, the NRMSE is constant. Otherwise, taking the log on both sides yields
   \begin{equation}\label{eq:uvs}
   \log (\mbox{NRMSE}(T^d_\textrm{vs}(\bbS)))  \approx \frac{\beta}{2}\log d + \frac{1}{2}(\log Z - \log B),\; 1 \leq d \leq \tau,
 \end{equation}
 which explains the relationship observed for uniform node sampling in Fig.~\ref{fig:uniform}.

From~\eqref{eq:mse_es}, we have for uniform edge sampling,
  \begin{equation}\label{eq:nrmse_es}
 \mbox{NRMSE}(T^d_\textrm{es}(\bbS)) = \sqrt{(1/\pi_d - 1)/B}.
 \end{equation}
  For $\theta_d \ll 1$, this implies
  \begin{equation*}
  \mbox{NRMSE}(T^d_\textrm{es}(\bbS)) \approx \begin{cases}
 \sqrt{Z\bar{d}d^{\beta-1}/B}\,, & 1 \leq d \leq \tau \\
 \sqrt{\frac{|E|}{d}/B}\,, & d > \tau.
 \end{cases}
  \end{equation*}
Taking the log on both sides, it follows that
 \begin{equation}\label{eq:ues}
   \log (\mbox{NRMSE}(T^d_\textrm{es}(\bbS)))  \approx \begin{cases}
 \frac{\beta-1}{2}\log d + \frac{1}{2}(\log Z + \log \bar{d} - \log B) \,, & 1 \leq d \leq \tau \\
 -\frac{1}{2} \left( \log d - \log |E| - \log B  \right)\,, & d > \tau,
 \end{cases}
 \end{equation}
 which explains the linear increase followed by the linear decrease observed in Fig.~\ref{fig:uniform}.
 Although some RW-based methods \fm{can collect uniform node samples (e.g., via random jumps)},
 NRMSE trends for large degrees are better described by~\eqref{eq:ues} than by~\eqref{eq:uvs}, since
 most of the information about these degrees comes from RW samples.

 \section{Results on node label distributions estimation}\label{sec:attributes}

This section focuses on network datasets which possess (non-topological) node labels. Using these datasets, all of which represent undirected
networks, we investigate which combinations of DUFS parameters outperform uniform node sampling when estimating node label distributions of the top 10\% largest degree nodes. These nodes often represent the most important objects in a network.

Two of the four undirected attribute-rich datasets we use are social networks (DBLP and LiveJournal) obtained from Stanford SNAP, while two are information networks (DBpedia and Wikipedia) obtained from CMU's Auton Lab GitHub repository \verb|active-search-gp-sopt| ~\cite{Ma:2015ut}. In these datasets, node labels correspond to some type of group membership and a node is allowed to be part of multiple groups
simultaneously. Figure~\ref{fig:att} shows, on the left, the degree distribution for each network. On the right, it displays the relative frequency of each attribute in decreasing order (blue bars/dots) along with attribute frequency among the top 10\% largest degree nodes (red bars/dots).  

  \begin{figure}[ht!]
\centering
      \includegraphics[height=0.65\textheight]{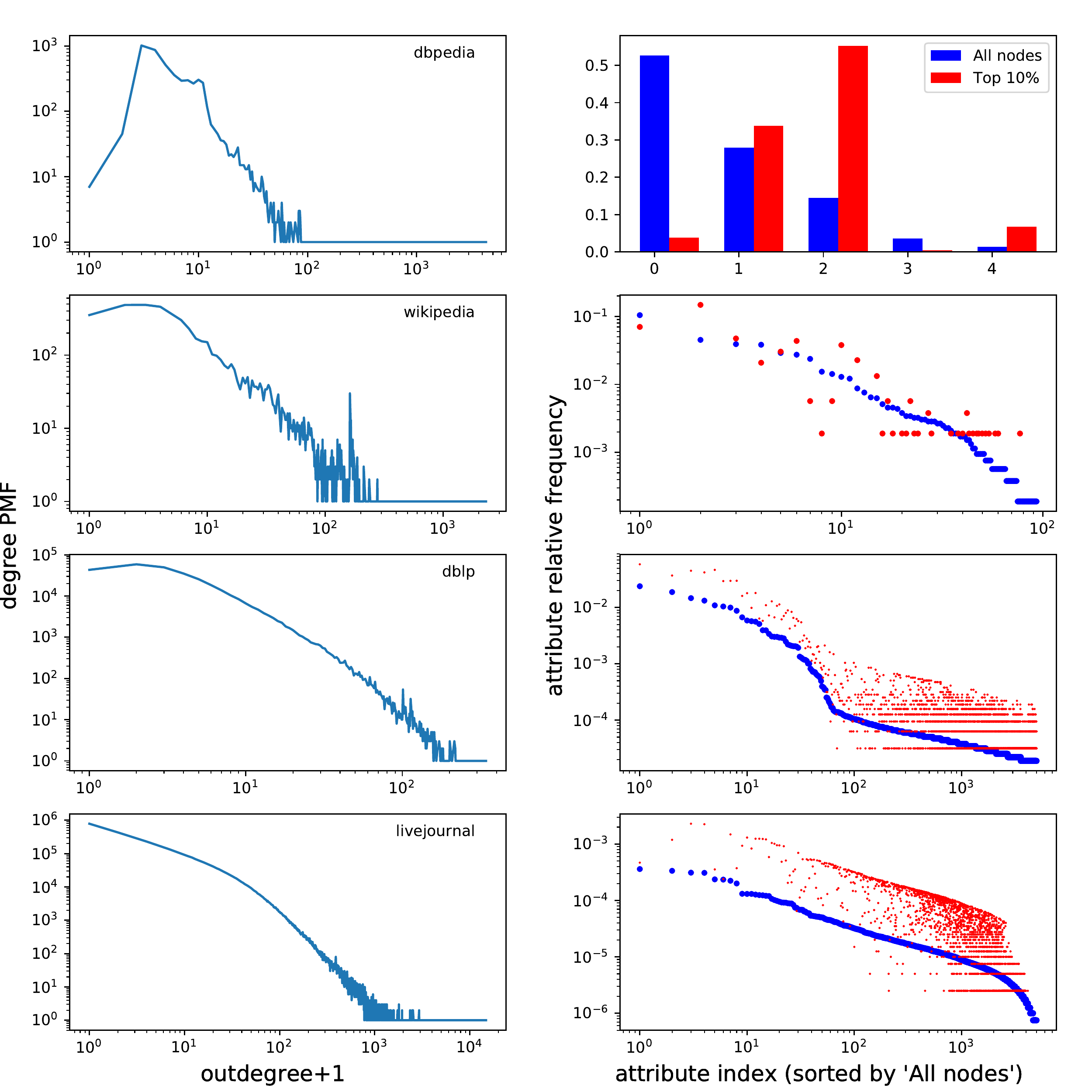}
%
%
%    \subfloat{%
%      \includegraphics[viewport=0  19 537 410,clip,width=0.4\textwidth]{dbpedia.pdf}
%    }
%    \subfloat{%
%      \includegraphics[viewport=0  19 537 410,clip,width=0.4\textwidth]{dbpedia_att.pdf}
%    }\\
%    \subfloat{%
%         \includegraphics[viewport=0  19 537 410,clip,width=0.4\textwidth]{wikipedia.pdf}
%    }
%    \subfloat{%
%      \includegraphics[viewport=0  19 537 410,clip,width=0.4\textwidth]{wikipedia_att.pdf}
%    }\\
%      \subfloat{%
%        \includegraphics[viewport=0  19 537 410,clip,width=0.4\textwidth]{dblp.pdf}
%    }
%     \subfloat{%
%      \includegraphics[viewport=0  19 537 410,clip,width=0.4\textwidth]{dblp_att.pdf}
%    }\\ 
%    \subfloat{%
%    \includegraphics[viewport=0  0 537 410,clip,width=0.4\textwidth]{lj.pdf}
%    }
%    \subfloat{%
%      \includegraphics[viewport=0  0 537 410,clip,width=0.4\textwidth]{lj_att.pdf}
%    }
    \caption{Degree and node attribute distribution for undirected attribute-rich networks.}
    \label{fig:att}
 \end{figure}

 %review what we did
We simulate 1000 DUFS runs on each undirected network for all combinations of random jump weight $w \in \{0.1, 1, 10\}$ and budget per walker $b \in \{1, 10, 10^2\}$. Figure~\ref{fig:att-results} compares the NRMSE associated with DUFS for different parameter combinations against uniform node sampling.
Uniform node sampling results are obtained analytically using eq.~\eqref{eq:nrmse_vs}. On DBpedia, Wikipedia and DBLP, almost all DUFS configurations outperform uniform node sampling. On LiveJournal, node sampling outperforms DUFS for attributes associated with large probability masses, but underperforms DUFS for attributes with smaller masses. In summary, we observe that DUFS with $w \in \{0.1,1.0\}$ and $b \in \{10,10^2\}$ yields superior accuracy than uniform node sampling when estimating node label distributions among the top 10\% largest degree nodes.
 
 %include figures
 \begin{figure}[]
     \subfloat[DBpedia\label{fig:dbpedia_att_results}]{%
      \includegraphics[width=0.46\textwidth]{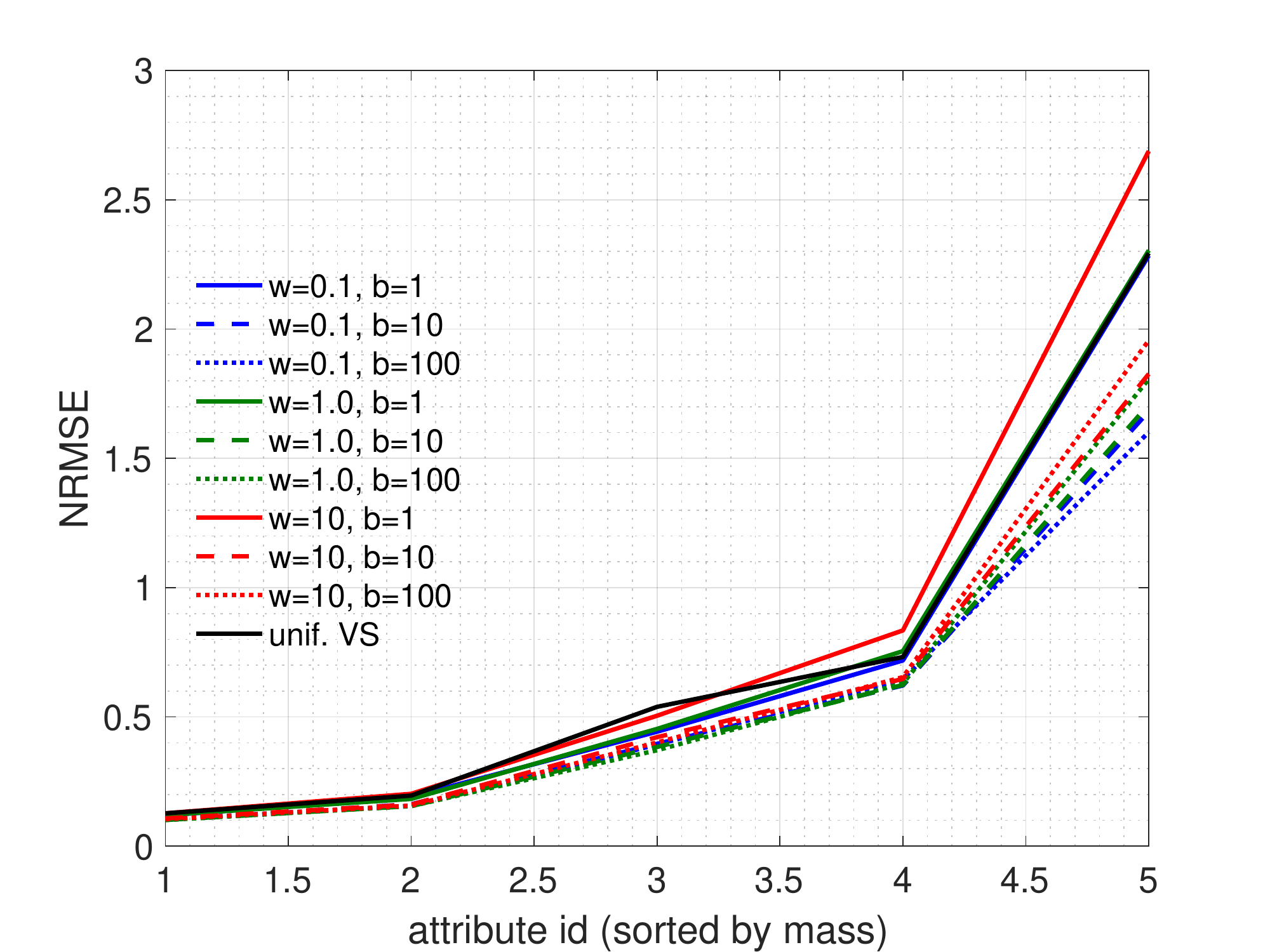}
    }
    \subfloat[Wikipedia\label{fig:wikipedia_att_results}]{%
      \includegraphics[width=0.46\textwidth]{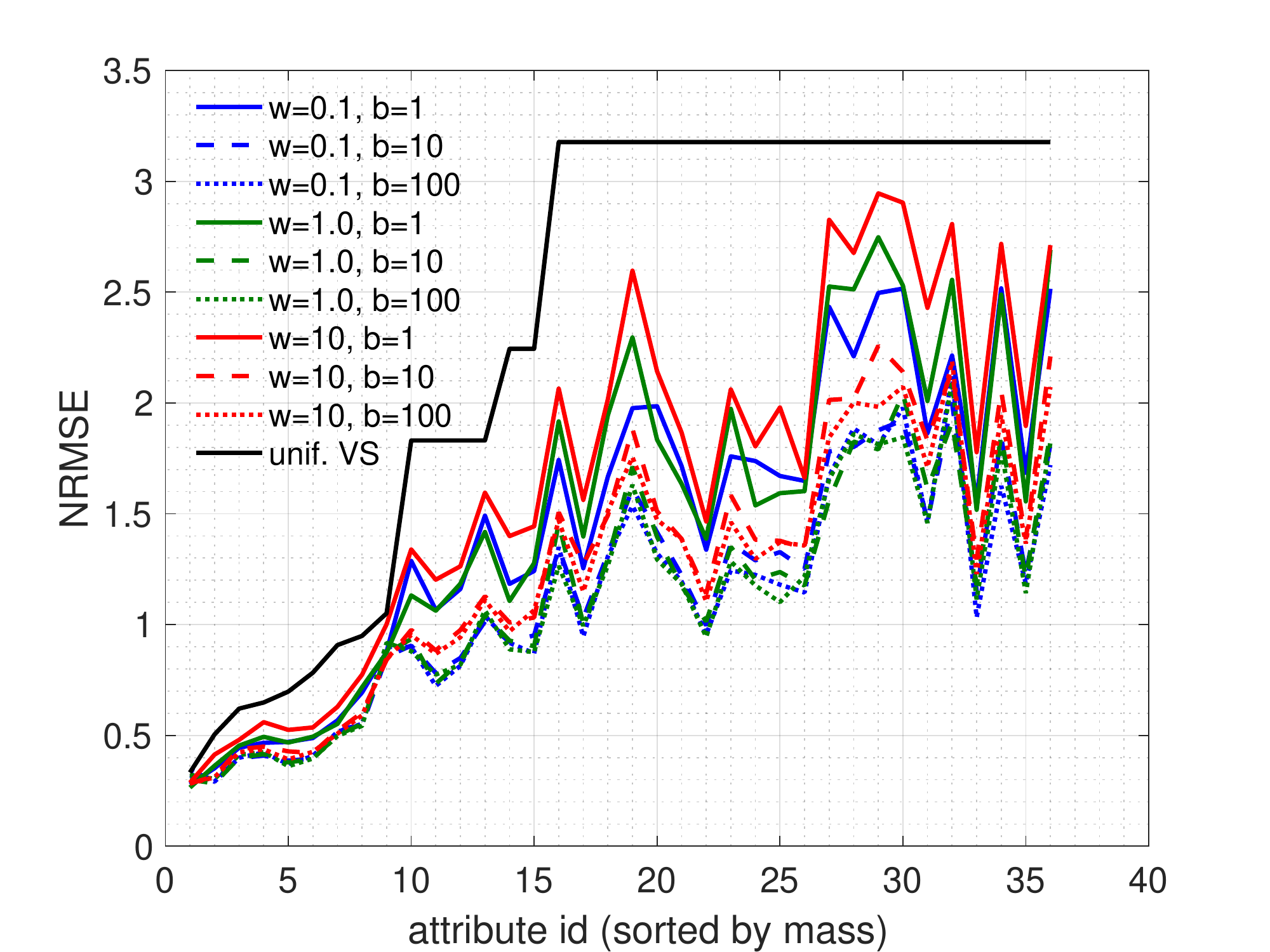}
    }\\
    \subfloat[DBLP\label{fig:dblp_att_results}]{%
      \includegraphics[width=0.46\textwidth]{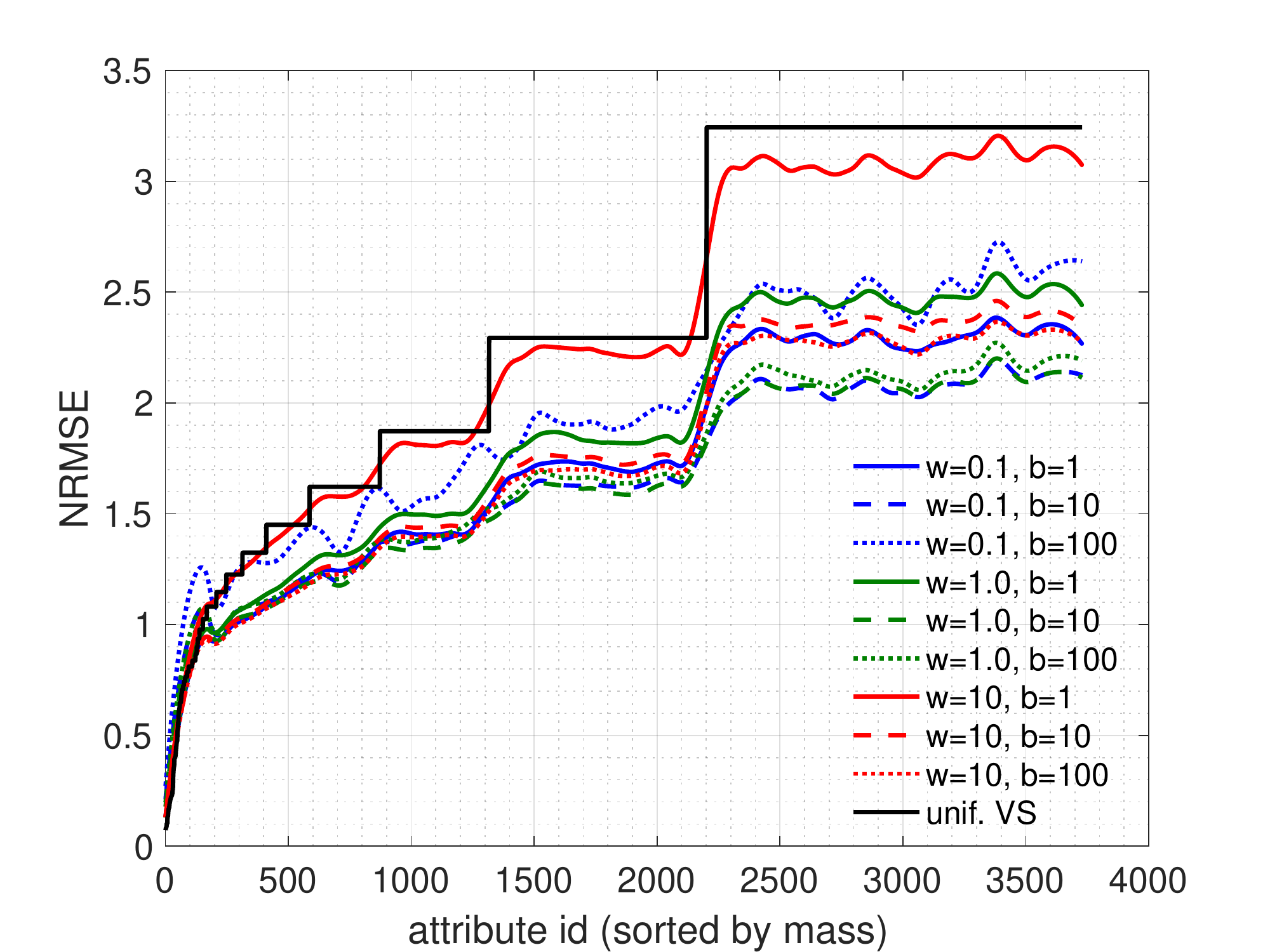}
    }
        \subfloat[LiveJournal\label{fig:lj_att_results}]{%
      \includegraphics[width=0.46\textwidth]{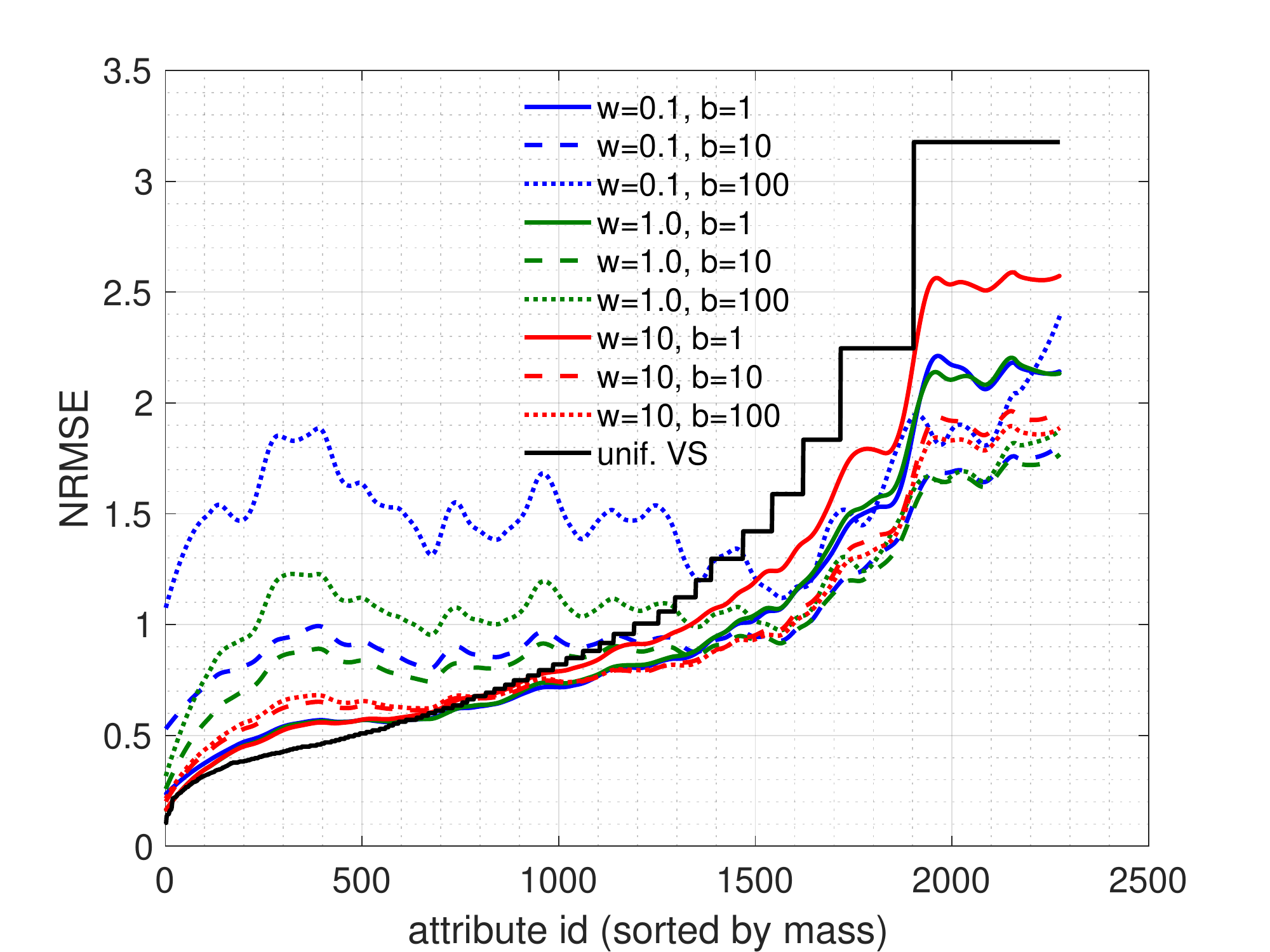}
    }
    \caption{Comparison of hybrid estimator (DUFS) with uniform node sampling. DUFS curves on DBLP plot are smoothed by a local regression using weighted linear least squares and a second degree polynomial model to avoid clutter. DUFS with $w \in \{0.1,1.0\}$ and $b \in \{10,10^2\}$ yields comparable or superior accuracy
    than uniform node sampling.}
    \label{fig:att-results}
 \end{figure}
 
 %\FloatBarrier
 
  \section{Discussion: DUFS performance in the absence of uniform node sampling}\label{sec:discussion}

In this section, we investigate the estimation accuracy of \{E,H\}-DUFS
when random walkers are {\em not} initialized uniformly over $V$.
We consider two simple non-uniform distributions over $V$ to determine the initial walker locations
walker positions:
\begin{itemize}
  \item Distribution \textsc{Prop}: proportional to the undirected degree, that is,
    \begin{equation}
      P(\textrm{initial walker location is }v) = \frac{\deg(v)}{\sum_{u \in V} \deg(u)};
    \end{equation}

  \item Distribution \textsc{Inv}: proportional to the reciprocal of the undirected degree, that is,
    \begin{equation}
      P(\textrm{initial walker location is }v) = \frac{\deg^{-1}(v)}{\sum_{u \in V} \deg^{1}(u)}.
    \end{equation}

\end{itemize}

We simulate E-DUFS and DUFS on each network dataset setting
 the budget per walker to $b \in \{1,10,10^2, B-1\}$ in a scenario where in-edges
 are visible, performing 100 runs. Note that $b=B-1$ corresponds to the case of a single random walker.
 Since we assume uniform node sampling (VS) is not available,
we must set the random jump weight to $w=0$. We include, however,
results obtained when the initial walker locations are determined via VS for comparison.
Figure~\ref{fig:nonuvs_e}
shows typical values of NRMSE associated with E-DUFS out-degree distribution
estimates. We observe that NRMSE decreases with the budget per walker until $b=10^2$, both for \textsc{Prop} and \textsc{Inv}.
Simulations for the case of a single walker ($b=B-1$) yielded poor results and are omitted.

\begin{figure}[]
      \includegraphics[width=0.98\textwidth]{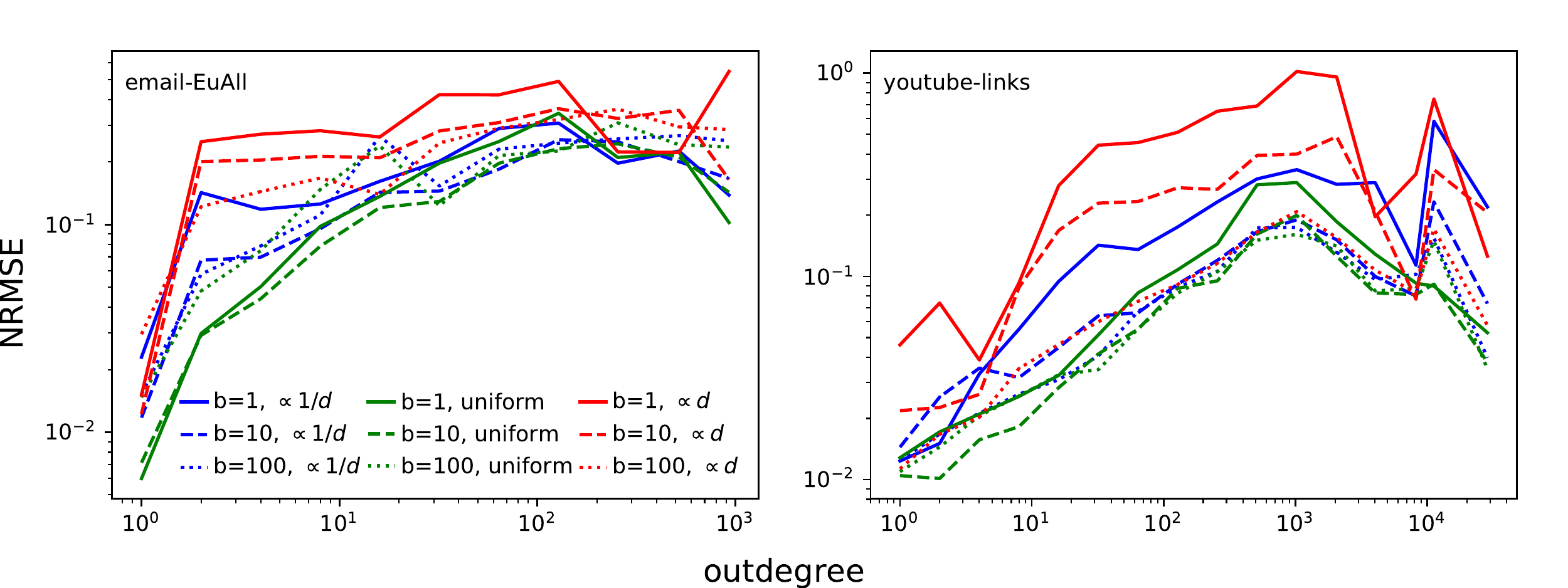}
%    \subfloat[email-EuAll]{%
%      \includegraphics[width=0.49\textwidth]{email-EuAll-nonuvs-e.pdf}
%    }
%    \subfloat[youtube-links]{%
%      \includegraphics[width=0.49\textwidth]{youtube-links-nonuvs-e.pdf}
%    }\\
%    \subfloat[web-Google\label{fig:web-Google-RVT}]{%
%      \includegraphics[width=0.49\textwidth]{web-Google-RVT.pdf}
%    }
%     \subfloat[wiki-Vote\label{fig:wiki-Vote}]{%
%      \includegraphics[width=0.49\textwidth]{wiki-Vote-RVT.pdf}
%    }
    \caption{Effect of initializing walkers non-uniformly over $V$ on E-DUFS
    accuracy. NRMSE decreases with budget per walker until $b=10^2$.}
    \label{fig:nonuvs_e}
 \end{figure}

Intuitively, using the hybrid estimator when the initial walker
locations come from some non-uniform distribution can incur unknown -- and
potentially large -- biases. We conducted a set of simulations with DUFS, which
corroborated this intuition. These results are omitted for conciseness.

In summary, our results indicate that when the initial walker locations are determined according to some unknown distribution,
a practitioner should use E-DUFS with moderately large $b$ (e.g., $10^2$). 
%H-DUFS yields very inaccurate estimates if their initial positions are not chosen via VS.

%\FloatBarrier

% Activate the following line by filling in the right side. If for example the name of the root file is Main.tex, write
% "...root = Main.tex" if the chapter file is in the same directory, and "...root = ../Main.tex" if the chapter is in a subdirectory.

%!TEX root =  ../frontier-tkdd.tex

\section{Related Work}\label{sec:related}

\textbf{Crawling methods for exploring undirected graphs}: A number of papers investigate crawling methods (e.g., breadth-first search, random walks, etc.) for generating subgraphs with similar topological properties as the underlying network~\cite{Leskovec2006,ICDM_Hubler08}. On the other hand, \cite{KDD_Maiya11} empirically investigates the performance of such methods w.r.t.\ specific measures of representativeness that can be useful in the context of specific applications (e.g., finding high-degree nodes for outbreak detection). However, these works
focus on techniques that yield biased samples of the network and do not possess any accuracy guarantees. \cite{Achlioptas2005,KurantJSAC2011} demonstrate that Breadth-First-Search (BFS) introduces a large bias towards high degree nodes, and that is difficult to remove these biases in general, although they can be reduced if the network in question is almost random~\cite{KurantJSAC2011}.
Random walk (RW) is  biased to sample high degree nodes, however its bias is known and can be easily corrected~\cite{TechReport}.
Random walks in the form of Respondent Driven Sampling (RDS)~\cite{Heckathorn2002,Salganik2004} have been used to estimate population densities using snowball samples of sociological studies.
The Metropolis-Hasting RW (MHRW)~\cite{Stutzbach2009} modifies the RW procedure \fm{to adjust for degree bias, in order to obtain uniform node samples.}
\cite{RibeiroCDC12,ChierichettiWWW2016} analytically prove that MHRW degree distribution estimates perform poorly in comparison to RWs.
Empirically, the accuracy of RW and MHRW has been compared in~\cite{WillingerRDS,Gjoka2010} and, as predicted by the theoretical results, RW is consistently more accurate than MHRW.

%To reduce the mixing time of a regular RW, which is critical for accelerating the performance of RW based crawling methods~\cite{ChierichettiWWW2016},
%\cite{TechReport} develops a method to coordinate multiple RWs
%for characterizing networks with loosely connected components.

Reducing the mixing time of a regular RW is one way of improving the performance of RW based crawling methods.
\cite{AvrachenkovRestarts} proves that random jumps increase the spectral gap of the random walk, which in turn,
 leads to faster convergence to the steady state distribution.
\cite{SIGM_Kurant11} assigns weights to nodes that are computed using their neighborhood information,
and develop a weighted RW-based method to perform stratified sampling on social networks.
They conduct experiments on Facebook and show that their stratified sampling technique achieves higher estimation accuracy than other methods. 
However, the neighborhood information in their method is limited to helping find random walk weights and is not used in the estimation of graph statistics of interest.
To solve this problem,
\cite{Dasgupta2012} randomly samples nodes (either uniformly or with a known bias) and then uses neighborhood information to improve its unbiased estimator.
\cite{GautamICDE2013} modifies the regular random walk by ``rewiring" the network of interest on-the-fly in order to reduce the mixing time of the walk.
 
\textbf{Crawling methods for exploring directed graphs:}
Estimating observable characteristics by sampling a directed graph (in this case, the Web graph) has been the subject of \cite{Yossef2008}
and~\cite{Henzinger2000},
which transform the directed graph of web-links into an undirected graph by adding reverse links, and then use a MHRW to sample webpages uniformly.
The DURW method proposed in~\cite{RibeiroINFOCOM2012} adapts the ``backward edge traversal" of~\cite{Yossef2008} to work with a pure random walk and random jumps. 
Both of these Metropolis-Hastings RWs (\cite{Yossef2008} and~\cite{Henzinger2000}) are designed to sample directed graphs and do not allow random jumps.
However, the ability to perform random jumps (even if jumps are rare) makes DURW and DUFS more efficient and accurate than the MetropolisHastings RW algorithm. Random walks with PageRank-style jumps are used in~\cite{Leskovec2006} to sample large graphs. 
In~\cite{Leskovec2006}, however, no technique is proposed to remove the large biases induced by the random walk and the random jumps, which makes this method unfit for estimation purposes. More recently, another method based on PageRank was proposed in~\cite{DNM}, but it assumes that obtaining uniform
node samples is not feasible. In the presence of multiple strongly connected components, this method offers no accuracy guarantees.

In the last decade, there has been a growing interest in graph sketching for processing massive networks.
A sketch is a compact representation of data. Unlike a sample, a sketch is computed over the entire graph,
that is observed as a data stream. For a survey on graph sketching techniques, please refer to~\cite{mcgregor2014graph}.

%In some cases, this stream will contain both edge insertions and deletions.
%There exist single-pass sketching techniques for solving problems such as
%determining if a graph is connected, bipartite, finding a sparse subgraph that preserves all node cut sizes,
%finding minimum spanning trees, etc \cite{mcgregor2014graph}. A graph sketch is space-efficient when it uses $O(n\textrm{ polylog }n)$ space,
%where $n$ is the number of nodes. In this paper, we consider the problem of characterizing networks through
%node label distributions, which are trivial to compute using $O(n)$ space in a single pass over the entire graph.
%Nevertheless, we empirically showed that it is possible to obtain good approximations for degree distributions
%sampling only 10\% of the vertices. In what follows, we contrast our work with existing works in sampling graphs.

% Activate the following line by filling in the right side. If for example the name of the root file is Main.tex, write
% "...root = Main.tex" if the chapter file is in the same directory, and "...root = ../Main.tex" if the chapter is in a subdirectory.
 
%!TEX root =  ../frontier-tkdd.tex

\section{Conclusion}\label{sec:conclusions}

In this paper, we proposed the Directed Unbiased Frontier Sampling (DUFS) method for
characterizing networks. DUFS generalizes the Frontier Sampling (FS) and
the Directed Unbiased Random Walk (DURW) methods. DUFS extends FS to
make it applicable to directed networks when incoming edges are not directly observable
by building on ideas from DURW. DUFS adapts DURW to use multiple coordinated walkers. Like DURW, DUFS
can also be applied to undirected networks without any modification.

We also proposed a novel estimator for node label distribution that can account
for FS and DUFS walkers initial locations -- or more generally, uniform node samples -- and a heuristic that
can reduce the variance incurred by node samples that happen to
sample nodes whose labels have extremely low probability masses.
When the proposed estimator is used in combination with the heuristic, we showed
that estimation errors can be significantly reduced in the distribution
head when compared with the estimator proposed in~\cite{TechReport},
regardless of whether we are estimating out-degree, in-degree or
joint in- and out-degree distributions.

We conducted an empirical study on the impact of DUFS parameters (namely, budget per walker
and random jump weight) on the estimation of out-degree and in-degree distributions using a large variety of datasets.
We considered four scenarios, corresponding to whether incoming edges are directly observable or not
and whether uniform node sampling has a similar or larger cost than
moving random walkers on the graph. This study allowed us to provide practical
guidelines on setting DUFS parameters to obtain accurate head estimates
or accurate tail estimates. When the goal is a balance between the two objectives,
intermediate configurations can be chosen.

Last, we compared DUFS with random walk-based methods designed
for undirected and directed networks. In our simulations for the scenario where in-edges are visible,
DUFS yielded much
lower estimation errors than a single random walk or multiple independent
random walks. We also observed that DUFS consistently outperforms
FS due to the random jumps and use of the improved estimator.
In the scenario where in-edges are unobservable, DUFS outperformed DURW when
estimating the probability mass associated with the smallest out-degree values (for equivalent parameter settings).
In addition, more often than not, DUFS slightly outperformed DURW when estimating
the mass associated to the largest out-degrees. In the presence of multiple
strongly connected components, DURW tends to move from small
to largest components more often than DUFS, sometimes exhibiting lower estimation
errors in the distribution tail. However, when restricting the estimation to the largest component,
DUFS outperforms DURW in virtually all datasets used in our simulations.

\appendix
\section*{Appendices}
\setcounter{section}{0}

\section{Hybrid estimator and its statistical properties} \label{app:properties}

Let us recall variables and constants used in the definition of the hybrid estimator:

\begin{tabular}{cl}
$n_{i}$ & number of node samples with label $i$\tabularnewline
$\theta_{i,j}$ & fraction of nodes in $G^{(t)}$ with label $i$ and undirected degree
$j$\tabularnewline
$m_{i,j}$ & number of edge samples with label $i$ and bias $j$\tabularnewline
$m_{i}=\sum_{j}m_{i,j}$ & total number of edge samples with label $i$\tabularnewline
$N=\sum_{i}n_{i}$ & total number of node samples\tabularnewline
$M=\sum_{i}m_{i}$ & total number of edge samples\tabularnewline
$B=N+M$ & total budget\tabularnewline
\end{tabular}

In this appendix, we derive the recursive variant of the hybrid estimator.
From that we derive its non-recursive variant. Next, we show that the
non-recursive variant is asymptotically unbiased.
In the case of undirected networks where the average degree is given, we show that the
resulting hybrid estimator of the undirected degree mass is the minimum variance
unbiased estimator (MVUE).

We approximate random walk samples in DUFS by uniform edge samples from $G_u$. Experience
from previous papers shows us that this approximation works very well in practice.
This yields the following likelihood function
\begin{equation}\label{eq:likelihood}
L(\btheta|\mathbf{n},\mathbf{m})  =  \frac{\prod_{i} \theta_i^{n_{i}} \prod_k (k\theta_{i,k})^{m_{i,k}}}{\left(\sum_{s,t} t\theta_{s,t}\right)^{M}}.
\end{equation}

The key idea in our derivation is that we can bypass the numerical estimation of the
$\theta_{i,j}$'s by noticing that $\theta_{i,j} \propto \theta_i$,
$\theta_{i,j} \propto m_{i,j}$ and $\theta_{i,j} \propto 1/j$.
Hence, the maximum
likelihood estimator of $\theta_{i,j}$ for $j=1,\ldots,Z$ is the
Horvitz-Thompson estimator 
\begin{equation}\label{eq:trick}
\hat \theta_{i,j} =  \frac{\theta_i m_{i,j}}{j\mu_i},
\end{equation}
where $\mu_i =\sum_k m_{i,k}/k$.

Substituting \eqref{eq:trick} in \eqref{eq:likelihood} yields
\begin{equation}
L(\btheta|\mathbf{n},\mathbf{m})  =   \frac{\prod_{i} \theta_i^{n_{i}} \prod_k (\theta_i m_{i,k}/\mu_i)^{m_{i,k}}}{\left(\sum_s \theta_s \sum_z m_{s,z}/\mu_s\right)^{M}}.
\end{equation}

The log-likelihood approximation is then given by
\begin{equation}
\mathcal{L}( \btheta|\mathbf{n},\mathbf{m}) =  -M\log\left(\sum_{s}  \theta_s \sum_z \frac{m_{s,z}}{\mu_s}\right)+\sum_{i}n_{i}\log \theta_i +\sum_k m_{i,k}(\log \theta_i + \log m_{i,k} - \log \mu_i). \label{eq:loglike}
\end{equation}

We rewrite $\theta_i$ as $e^{\beta_i} / \sum_j e^{\beta_j}$ to account for the distribution constraints $\sum_i \theta_i = 1$ and $\theta_i \in [0,1]$. Hence, we have
\begin{equation}
\mathcal{L}( \bbeta|\mathbf{n},\mathbf{m}) =  -M \log\left( \sum_{s} \frac{e^{\beta_s} m_{s}}{\mu_s}\right)+\sum_{i}(n_{i} + m_{i}) \beta_i - N \log\left(\sum_j e^{\beta_j}\right) + C, \label{eq:loglike}
\end{equation}
 where $m_i = \sum_k m_{i,k}$ and $C$ is a constant that does not depend on $\beta$.
 
 The partial derivative w.r.t. $\beta_i$ is given by
 \begin{equation}\label{eq:derivative}
\frac{\partial \mathcal{L}(\bbeta|\mathbf{n},\mathbf{m})}{\partial\beta_{i}} =  -\frac{M e^{\beta_i} m_{i}/ \mu_i}{\sum_{s} e^{\beta_s} m_{s}/ \mu_s } + n_i + m_i  -\frac{Ne^{\beta_i}}{\sum_j e^{\beta_j}}.
 \end{equation}
%Although we cannot solve and obtain a closed form expression for $\beta_i$,

Setting $\partial \mathcal{L}(\bbeta|\mathbf{n},\mathbf{m})/\partial\beta_{i} = 0$ and substituting back $\theta_i$ yields
\begin{eqnarray}
 %0 & = & -\frac{M (e^{\beta_i}/z) (m_{i}/ \mu_i)}{\sum_{s} (e^{\beta_s}/z)(m_{s}/ \mu_s) } + n_i + m_i  -\frac{Ne^{\beta_i}}{z}\\
 %0 & = & - \frac{e^{\beta_i}}{z} \left(N + M \frac{m_{i}/ \mu_i}{\sum_{s} (e^{\beta_s}/z)( m_{s}/ \mu_s) } \right) + n_i+m_i\\
 \theta_i^\star & = & \frac{n_i + m_i}{N + M \frac{m_{i}/ \mu_i}{\sum_{s} \theta_s^\star m_{s}/ \mu_s }}.
\end{eqnarray}

\begin{thm}
Let $N=cB$
and $M=(1-c)B$, for some $0<c<1$. The estimator
\begin{equation}
\hat{\theta}_{i}=\frac{n_{i}+m_{i}}{N+M\frac{m_{i}}{\mu_{i} \hat{d}}}, \label{eq:new}
\end{equation}
where $\mu_{i}=\sum_{k} m_{i,k}/k$ and $\hat{d}=M/\sum_{i} \mu_i$, is an asymptotically unbiased estimator of $\theta_{i}$.
\end{thm}

\begin{proof}
In the limit as $B\rightarrow\infty$,
we have
\[
E[n_{i}]=N\theta_{i},\qquad E[m_{i,k}]=M\frac{k\theta_{i,k}}{\sum_{s,l}l\theta_{sl}},\qquad E[m_{i}]=M\frac{\sum_{k}k\theta_{i,k}}{\sum_{s,l}l\theta_{s,l}},
\]
and thus,
\[
E[\mu_{i}]=M\frac{\sum_{k}k\theta_{i,k}/k}{\sum_{s,l}l\theta_{sl}}=M\frac{\theta_{i}}{\sum_{s,l}l\theta_{sl}}\quad\textrm{and}\quad E\left[\frac{m_{i}}{\mu_{i}}\right]=\frac{\sum_{k}k \theta_{i,k}}{\theta_{i}}.
\]
It follows that
\begin{eqnarray*}
\lim_{B\rightarrow\infty}E[\hat{d}] & = & \frac{M}{M\frac{\sum_{i}\theta_{i}}{\sum_{s,l}l\theta_{sl}}} = \sum_{s,l}l\theta_{sl}.
\end{eqnarray*}

Substituting the above in eq.~(\ref{eq:new}), we have
\[
\lim_{B\rightarrow\infty}E[\theta_{i}^{\star}] = \frac{N\theta_{i}+M\frac{\sum_{k}k\theta_{i,k}}{\sum_{s,l}l\theta_{s,l}}}{N+M\frac{\sum_{k}k\theta_{i,k}/\theta_{i}}{\sum_{s,l}l\theta_{s,l}}} = \theta_{i}.
\]
This concludes the proof.\end{proof}

In Section~\ref{sec:hybrid} we mentioned a special case of the previous estimator, where the node label is the undirected degree itself.
We prove that this estimator, denoted by $\hat \theta_i$ %, is described by eq.~\eqref{eq:und_theta}. We now prove that $\hat \theta_i$
is the minimum variance unbiased
estimator (MVUE) of $\theta_i$.

\begin{thm} \label{thm:unbiasedness}
The estimator
\begin{equation*}
\bar{\theta}_{i}=\frac{n_{i}+m_{i}}{N+Mi/\bar{\mu}}, 
\end{equation*}
where $\bar\mu = \sum_{j} j\theta_{j}$, is an unbiased estimator of $\theta_{i}$.
\end{thm}
\begin{proof}
We know that $n_i\sim\mathrm{Binomial}(N,\theta_{i})$ and
$m_{i}\sim\mathrm{Binomial}(M,i\theta_{i}/\bar{\mu})$. Hence,
\begin{eqnarray*}
E[\hat{\theta}_{i}] & = & \sum_{n_{i},m_{i}}\frac{n_{i}+m_{i}}{N+Mi/\bar{\mu}}\overbrace{\binom{N}{n_{i}}\theta_{i}^{n_{i}}(1-\theta_{i})^{N-n_{i}}}^{A(n_{i})}\overbrace{\binom{M}{m_{i}}\left(\frac{i\theta_{i}}{\bar{\mu}}\right)^{m_{i}}\left(1-\frac{i\theta_{i}}{\bar{\mu}}\right)^{M-m_{i}}}^{B(m_{i})}\\
 & = & \frac{1}{N+Mi/\bar{\mu}}\left(\sum_{n_{i}}n_{i}A(n_{i})\sum_{m_{i}}B(m_{i})+\sum_{m_{i}}m_{i}B(m_{i})\sum_{n_{i}}A(n_{i})\right)\\
 & = & \frac{1}{N+Mi/\bar{\mu}}\left(\sum_{n_{i}}n_{i}A(n_{i})+\sum_{m_{i}}m_{i}B(m_{i})\right)\\
 & = & \frac{1}{N+Mi/\bar{\mu}}\left(N\theta_{i}+Mi\theta_{i}/\bar{\mu}\right)\\
 & = & \theta_{i}.
\end{eqnarray*}
\end{proof}

Having proved that $\hat \theta_i$ is unbiased, we are now
ready to show that it is also the minimum variance unbiased estimator (MVUE).
In order to do so, we prove Lemmas~\ref{sufficiency} and~\ref{completeness} that show
respectively 
that $n_i+m_i$ is a sufficient and complete statistic of $\theta_i$.

\begin{lemma}\label{sufficiency}
The statistic $n_{i}+m_{i}$ is a sufficient statistic with respect to the likelihood of $\theta_{i}$.
\end{lemma}
\begin{proof}
The log-likelihood equation for estimator~\eqref{eq:und_theta}
is given by
\begin{eqnarray}
L(\btheta| \mathbf{n},\mathbf{m}) & = & \frac{\prod_i \theta_i^{n_i} \prod_j (j \theta_j)^{m_j}}{\hat \mu^M} \nonumber \\
& = & \frac{\prod_j j ^{m_j}}{\hat \mu^M} \prod_i \theta_i^{n_i+m_i} \label{eq:factorization}. 
\end{eqnarray}

We can see from eq.~\eqref{eq:factorization} that the likelihood function $L(\btheta| \mathbf{n},\mathbf{m})$
 can be factored into a product such that one factor, $\prod_j j^{m_j}/\hat{\mu}^M$, does not depend on $\theta_i$
 and the other factor, which does depend on $\theta_i$, depends on $\mathbf{n}$ and $\mathbf{m}$ only through $n_i+m_i$.
From the Fisher-Neyman factorization Theorem~\cite{lehmann1991theory}, we conclude that 
$n_i+m_i$ is a sufficient statistic for the distribution of the sample.
 \end{proof}
 
We now prove that $n_i+m_i$ is also a complete statistic for the distribution of the sample.

\begin{definition} Let $X$ be a random variable whose probability distribution
belongs to a parametric family of probability distributions $P_{\theta}$
parametrized by $\theta$. The statistic $s$ is said to be complete
for the distribution of $X$ if for every measurable function $g$
(which must be independent of $\theta$) the following implication
holds:
\[
E(g(s(X)))=0\ \textrm{for all }\mbox{\ensuremath{\theta}}\Rightarrow P_{\theta}(g(s(X))=0)=1\ \textrm{for all }\theta.
\]
\end{definition}
 
\begin{lemma} \label{completeness}
The statistic $n_{i}+m_{i}$ is a complete statistic w.r.t.\ the likelihood of $\theta_{i}$.
\end{lemma}
\begin{proof}
\begin{eqnarray}
E[g(n_{i}+m_{i})] & = & 0\nonumber \\
\sum_{n_{i},m_{i}}g(n_{i}+m_{i})P_{\theta}(n_{i},m_{i}) & = & 0\nonumber \\
\sum_{n_{i},m_{i}}g(n_{i}+m_{i})A(n_{i})B(m_{i}) & = & 0\label{eq:completeness}
\end{eqnarray}

The LHS of (\ref{eq:completeness}) is a polynomial of degree $M+N$
on $\theta_{i}$. Hence, it can be written as
\begin{equation}
C_{0}+C_{1}\theta_{i}+C_{2}\theta_{i}^{2}+\ldots+C_{N+M}\theta_{i}^{N+M}=0.\label{eq:condition}
\end{equation}

We prove that $P_{\theta}(g(s(X))=0)=1\ \textrm{for all }\theta$ by contradiction.
Suppose that there is a $\theta$ such that $P_{\theta}(g(s(X)) \neq 0)>0$.
In order to have $E(g(s(X)))=0$, there must be terms for which $g(.)$ is
strictly positive and  terms for which $g(.)$ is
strictly negative. Let $g(h_1)$ be the smallest $h_1$ such that $g(h_1) > 0$.
Let $g(h_2)$ be the smallest $h_2$ such that $g(h_2) < 0$.
Let $h = \min(h_1,h_2)$.

Expanding $A(n_{i})B(m_{i})$ in eq.~\eqref{eq:completeness} we note that the degree of the resulting polynomial is at least $n_{i}+m_{i}$ on $\theta_{i}$.
Therefore, the coefficient $C_h$ in eq.~\eqref{eq:condition} associated with $\theta_{i}^h$ cannot have terms of $g(.)$ larger than $h$.
Then $C_h$ can only be zero if $h_1 = h_2$ which is a contradiction.
\end{proof}

\begin{thm}\label{thm:mvue2}
 The unbiased estimator $\bar \theta_i$ is the minimum variance unbiased estimator (MVUE) of $\theta_i$.
\end{thm}
\begin{proof}
According to the Lehmann-Scheffe Theorem~\cite{lehmann1991theory}, if $T(\mathbb{S})$ is a complete sufficient statistic, there is at most
one unbiased estimator that is a function of $T(\mathbb{S})$. From Lemmas~\ref{sufficiency} and~\ref{completeness}, we have that
$n_i+m_i$ is a complete sufficient statistic of $\theta_i$. Clearly, the unbiased estimator $\hat\btheta$ in eq.~\eqref{eq:new} is a function
 $n_i+m_i$. Therefore, $\hat\theta_i$ must be the MVUE.
\end{proof}

Alternatively, we can prove Theorem~\ref{thm:mvue2}
from Lemmas~\ref{sufficiency} and~\ref{completeness}
by showing that applying the Rao-Blackwell Theorem to the unbiased estimator $\hat{\theta}_{i}$
using the complete sufficient statistic $n_i+m_i$
 yields exactly the same estimator:
\begin{eqnarray*}
\theta_{i}^{\prime} & = & E\left[\hat{\theta}_{i}|n_i+m_i \right]\\
 & = & \sum_{t_{j}} t_{j}P(\hat{\theta}_{i}=t_{j}|n_{i}+m_{i})\\
 & = & \sum_{t_{j}} t_{j}1\left\{ \frac{n_{i}+m_{i}}{N+Mi/\bar{\mu}}=t_{j}\right\} \\
 & = & \frac{n_{i}+m_{i}}{N+Mi/\bar{\mu}}.
\end{eqnarray*}

\appendixhead{MURAI}

%% Acknowledgments
%\begin{acks}
%The authors would like to thank Dr. Maura Turolla of Telecom
%Italia for providing specifications about the application scenario.
%\end{acks}

% Bibliography
\bibliographystyle{ACM-Reference-Format-Journals}
\bibliography{frontier-tkdd}


\begin{thebibliography}{00}

%%% the Web address.

\ifx \showCODEN    \undefined \def \showCODEN     #1{\unskip}     \fi
\ifx \showDOI      \undefined \def \showDOI       #1{{\tt DOI:}\penalty0{#1}\ }
  \fi
\ifx \showISBNx    \undefined \def \showISBNx     #1{\unskip}     \fi
\ifx \showISBNxiii \undefined \def \showISBNxiii  #1{\unskip}     \fi
\ifx \showISSN     \undefined \def \showISSN      #1{\unskip}     \fi
\ifx \showLCCN     \undefined \def \showLCCN      #1{\unskip}     \fi
\ifx \shownote     \undefined \def \shownote      #1{#1}          \fi
\ifx \showarticletitle \undefined \def \showarticletitle #1{#1}   \fi
\ifx \showURL      \undefined \def \showURL       #1{#1}          \fi

\bibitem[\protect\citeauthoryear{Achlioptas, Clauset, Kempe, and
  Moore}{Achlioptas et~al\mbox{.}}{2009}]%
        {Achlioptas2005}
{Dimitris Achlioptas}, {Aaron Clauset}, {David Kempe}, {and} {Cristopher
  Moore}. 2009.
\newblock \showarticletitle{On the Bias of Traceroute Sampling: Or, Power-law
  Degree Distributions in Regular Graphs}.
\newblock {\em J. ACM\/} {56}, 4, Article 21 (July 2009), 28 pages.
\newblock
\showISSN{0004-5411}
\showDOI{%
\url{http://dx.doi.org/10.1145/1538902.1538905}}


\bibitem[\protect\citeauthoryear{Avrachenkov, Ribeiro, and Towsley}{Avrachenkov
  et~al\mbox{.}}{2010}]%
        {AvrachenkovRestarts}
{Konstantin Avrachenkov}, {Bruno Ribeiro}, {and} {Don Towsley}. 2010.
\newblock {\em Improving Random Walk Estimation Accuracy with Uniform
  Restarts}.
\newblock Springer Berlin Heidelberg, Berlin, Heidelberg, 98--109.
\newblock
\showISBNx{978-3-642-18009-5}
\showDOI{%
\url{http://dx.doi.org/10.1007/978-3-642-18009-5_10}}


\bibitem[\protect\citeauthoryear{Bar-Yossef and Gurevich}{Bar-Yossef and
  Gurevich}{2008}]%
        {Yossef2008}
{Ziv Bar-Yossef} {and} {Maxim Gurevich}. 2008.
\newblock \showarticletitle{Random sampling from a search engine's index}.
\newblock {\it J. ACM} {55}, 5 (2008), 1--74.
\newblock


\bibitem[\protect\citeauthoryear{Boccaletti, Latora, Moreno, Chavez, and
  Hwang}{Boccaletti et~al\mbox{.}}{2006}]%
        {SurveyMeasuresGraphs}
{S. Boccaletti}, {V. Latora}, {Y. Moreno}, {M. Chavez}, {and} {D.-U. Hwang}.
  2006.
\newblock \showarticletitle{Complex networks: Structure and dynamics}.
\newblock {\em Physics Reports\/} {424}, 4-5 (2006), 175--308.
\newblock
\showISSN{0370-1573}
\showDOI{%
\url{http://dx.doi.org/10.1016/j.physrep.2005.10.009}}


\bibitem[\protect\citeauthoryear{Chiericetti, Dasgupta, Kumar, Lattanzi, and
  Sarl\'{o}s}{Chiericetti et~al\mbox{.}}{2016}]%
        {ChierichettiWWW2016}
{Flavio Chiericetti}, {Anirban Dasgupta}, {Ravi Kumar}, {Silvio Lattanzi},
  {and} {Tam\'{a}s Sarl\'{o}s}. 2016.
\newblock \showarticletitle{On Sampling Nodes in a Network}. In {\em
  Proceedings of the 25th International Conference on World Wide Web} {\em (WWW
  '16)}. International World Wide Web Conferences Steering Committee, Republic
  and Canton of Geneva, Switzerland, 471--481.
\newblock
\showISBNx{978-1-4503-4143-1}
\showDOI{%
\url{http://dx.doi.org/10.1145/2872427.2883045}}


\bibitem[\protect\citeauthoryear{Dasgupta, Kumar, and Sivakumar}{Dasgupta
  et~al\mbox{.}}{2012}]%
        {Dasgupta2012}
{Anirban Dasgupta}, {Ravi Kumar}, {and} {D. Sivakumar}. 2012.
\newblock \showarticletitle{Social Sampling}. In {\em Proceedings of the 18th
  ACM SIGKDD International Conference on Knowledge Discovery and Data Mining}
  {\em (KDD '12)}. ACM, New York, NY, USA, 235--243.
\newblock
\showISBNx{978-1-4503-1462-6}
\showDOI{%
\url{http://dx.doi.org/10.1145/2339530.2339572}}


\bibitem[\protect\citeauthoryear{Eagle, Pentland, and Lazer}{Eagle
  et~al\mbox{.}}{2009}]%
        {MobileFriends}
{Nathan Eagle}, {Alex~(Sandy) Pentland}, {and} {David Lazer}. 2009.
\newblock \showarticletitle{Inferring friendship network structure by using
  mobile phone data}.
\newblock {\em Proceedings of the National Academy of Sciences\/} {106}, 36
  (2009), 15274--15278.
\newblock
\showDOI{%
\url{http://dx.doi.org/10.1073/pnas.0900282106}}


\bibitem[\protect\citeauthoryear{Gjoka, Butts, Kurant, and Markopoulou}{Gjoka
  et~al\mbox{.}}{2010}]%
        {Gjoka2010}
{Minas Gjoka}, {Carter~T. Butts}, {Maciej Kurant}, {and} {Athina Markopoulou}.
  2010.
\newblock \showarticletitle{Walking in Facebook: A Case Study of Unbiased
  Sampling of OSNs}. In {\em Proceedings of IEEE INFOCOM 2010}. 1--9.
\newblock
\showISSN{0743-166X}
\showDOI{%
\url{http://dx.doi.org/10.1109/INFCOM.2010.5462078}}


\bibitem[\protect\citeauthoryear{Heckathorn}{Heckathorn}{1997}]%
        {Heckathorn97}
{Douglas~D. Heckathorn}. 1997.
\newblock \showarticletitle{Respondent-Driven Sampling: A New Approach to the
  Study of Hidden Populations}.
\newblock {\em Social Problems\/} {44}, 2 (1997), 174--199.
\newblock
\showISSN{0037-7791}
\showDOI{%
\url{http://dx.doi.org/10.2307/3096941}}


\bibitem[\protect\citeauthoryear{Heckathorn}{Heckathorn}{2002}]%
        {Heckathorn2002}
{Douglas~D. Heckathorn}. 2002.
\newblock \showarticletitle{Respondent-Driven Sampling II: Deriving Valid
  Population Estimates from Chain-Referral Samples of Hidden Populations}.
\newblock {\em Social Problems\/} {49}, 1 (2002), 11--34.
\newblock
\showISSN{0037-7791}
\showDOI{%
\url{http://dx.doi.org/10.1525/sp.2002.49.1.11}}


\bibitem[\protect\citeauthoryear{Henzinger, Heydon, Mitzenmacher, and
  Najork}{Henzinger et~al\mbox{.}}{2000}]%
        {Henzinger2000}
{Monika~R. Henzinger}, {Allan Heydon}, {Michael Mitzenmacher}, {and} {Marc
  Najork}. 2000.
\newblock \showarticletitle{On near-uniform URL sampling}.
\newblock {\em Computer Networks\/} {33}, 1-6 (2000), 295 -- 308.
\newblock
\showISSN{1389-1286}
\showDOI{%
\url{http://dx.doi.org/10.1016/S1389-1286(00)00055-4}}


\bibitem[\protect\citeauthoryear{Hubler, Kriegel, Borgwardt, and
  Ghahramani}{Hubler et~al\mbox{.}}{2008}]%
        {ICDM_Hubler08}
{Christian Hubler}, {H-P Kriegel}, {Karsten Borgwardt}, {and} {Zoubin
  Ghahramani}. 2008.
\newblock \showarticletitle{Metropolis Algorithms for Representative Subgraph
  Sampling}. In {\em 2008 Eighth IEEE International Conference on Data Mining}.
  283--292.
\newblock
\showISSN{1550-4786}
\showDOI{%
\url{http://dx.doi.org/10.1109/ICDM.2008.124}}


\bibitem[\protect\citeauthoryear{Kurant, Gjoka, Butts, and Markopoulou}{Kurant
  et~al\mbox{.}}{2011a}]%
        {SIGM_Kurant11}
{Maciej Kurant}, {Minas Gjoka}, {Carter~T. Butts}, {and} {Athina Markopoulou}.
  2011a.
\newblock \showarticletitle{Walking on a Graph with a Magnifying Glass:
  Stratified Sampling via Weighted Random Walks}. In {\em ACM SIGMETRICS 2011}.
  ACM, New York, NY, USA, 281--292.
\newblock
\showISBNx{978-1-4503-0814-4}
\showDOI{%
\url{http://dx.doi.org/10.1145/1993744.1993773}}


\bibitem[\protect\citeauthoryear{Kurant, Markopoulou, and Thiran}{Kurant
  et~al\mbox{.}}{2011b}]%
        {KurantJSAC2011}
{Maciej Kurant}, {Athina Markopoulou}, {and} {Patrick Thiran}. 2011b.
\newblock \showarticletitle{Towards Unbiased BFS Sampling}.
\newblock {\em IEEE Journal on Selected Areas in Communications\/} {29}, 9
  (September 2011), 1799--1809.
\newblock


\bibitem[\protect\citeauthoryear{Lehmann, Casella, and Casella}{Lehmann
  et~al\mbox{.}}{1991}]%
        {lehmann1991theory}
{Erich~Leo Lehmann}, {George Casella}, {and} {George Casella}. 1991.
\newblock {\em Theory of point estimation}.
\newblock Wadsworth \& Brooks/Cole Advanced Books \& Software.
\newblock


\bibitem[\protect\citeauthoryear{Leskovec and Faloutsos}{Leskovec and
  Faloutsos}{2006}]%
        {Leskovec2006}
{Jure Leskovec} {and} {Christos Faloutsos}. 2006.
\newblock \showarticletitle{Sampling from Large Graphs}. In {\em Proceedings of
  the 12th ACM SIGKDD International Conference on Knowledge Discovery and Data
  Mining} {\em (KDD '06)}. ACM, New York, NY, USA, 631--636.
\newblock
\showISBNx{1-59593-339-5}
\showDOI{%
\url{http://dx.doi.org/10.1145/1150402.1150479}}


\bibitem[\protect\citeauthoryear{Leskovec and Krevl}{Leskovec and
  Krevl}{2014}]%
        {SNAP}
{Jure Leskovec} {and} {Andrej Krevl}. 2014.
\newblock {SNAP Datasets}: {Stanford} Large Network Dataset Collection.
\newblock \url{http://snap.stanford.edu/data}.   (June 2014).
\newblock


\bibitem[\protect\citeauthoryear{Leskovec, Lang, Dasgupta, and
  Mahoney}{Leskovec et~al\mbox{.}}{2008}]%
        {LeskovecCommunity}
{Jure Leskovec}, {Kevin~J. Lang}, {Anirban Dasgupta}, {and} {Michael~W.
  Mahoney}. 2008.
\newblock \showarticletitle{Statistical Properties of Community Structure in
  Large Social and Information Networks}. In {\em Proceedings of the 17th
  International Conference on World Wide Web} {\em (WWW '08)}. ACM, New York,
  NY, USA, 695--704.
\newblock
\showISBNx{978-1-60558-085-2}
\showDOI{%
\url{http://dx.doi.org/10.1145/1367497.1367591}}


\bibitem[\protect\citeauthoryear{Ma, Huang, and Schneider}{Ma
  et~al\mbox{.}}{2015}]%
        {Ma:2015ut}
{Yifei Ma}, {Tzu-Kuo Huang}, {and} {Jeff~G Schneider}. 2015.
\newblock \showarticletitle{{Active Search and Bandits on Graphs using
  Sigma-Optimality.}}. In {\em Conference on Uncertainty in Artificial
  Intelligence}. 542--551.
\newblock


\bibitem[\protect\citeauthoryear{Maiya and Berger-Wolf}{Maiya and
  Berger-Wolf}{2011}]%
        {KDD_Maiya11}
{Arun~S. Maiya} {and} {Tanya~Y. Berger-Wolf}. 2011.
\newblock \showarticletitle{Benefits of Bias: Towards Better Characterization
  of Network Sampling}. In {\em Proceedings of the 17th ACM SIGKDD
  International Conference on Knowledge Discovery and Data Mining} {\em (KDD
  '11)}. ACM, New York, NY, USA, 105--113.
\newblock
\showISBNx{978-1-4503-0813-7}
\showDOI{%
\url{http://dx.doi.org/10.1145/2020408.2020431}}


\bibitem[\protect\citeauthoryear{McGregor}{McGregor}{2014}]%
        {mcgregor2014graph}
{Andrew McGregor}. 2014.
\newblock \showarticletitle{Graph stream algorithms: a survey}.
\newblock {\em ACM SIGMOD Record\/} {43}, 1 (2014), 9--20.
\newblock


\bibitem[\protect\citeauthoryear{Mislove, Marcon, Gummadi, Druschel, and
  Bhattacharjee}{Mislove et~al\mbox{.}}{2007}]%
        {Mislove}
{Alan Mislove}, {Massimiliano Marcon}, {Krishna~P. Gummadi}, {Peter Druschel},
  {and} {Bobby Bhattacharjee}. 2007.
\newblock \showarticletitle{Measurement and Analysis of Online Social
  Networks}. In {\em Proceedings of the 7th ACM SIGCOMM Conference on Internet
  Measurement} {\em (IMC '07)}. ACM, New York, NY, USA, 29--42.
\newblock
\showISBNx{978-1-59593-908-1}
\showDOI{%
\url{http://dx.doi.org/10.1145/1298306.1298311}}


\bibitem[\protect\citeauthoryear{Murai, Ribeiro, Towsley, and Wang}{Murai
  et~al\mbox{.}}{2013}]%
        {jsac13}
{Fabricio Murai}, {Bruno Ribeiro}, {Don Towsley}, {and} {Pinghui Wang}. 2013.
\newblock \showarticletitle{On Set Size Distribution Estimation and the
  Characterization of Large Networks via Sampling}.
\newblock {\em IEEE Journal on Selected Areas in Communications\/} {31}, 6
  (June 2013), 1017--1025.
\newblock
\showISSN{0733-8716}
\showDOI{%
\url{http://dx.doi.org/10.1109/JSAC.2013.130604}}


\bibitem[\protect\citeauthoryear{Murai, Ribeiro, Towsley, and Wang}{Murai
  et~al\mbox{.}}{2018}]%
        {Murai18}
{Fabricio Murai}, {Bruno Ribeiro}, {Don Towsley}, {and} {Pinghui Wang}. 2018.
\newblock {\em Characterizing Directed and Undirected Networks via
  Multidimensional Walks with Jumps}.
\newblock {T}echnical {R}eport arXiv:1703.08252.
\newblock


\bibitem[\protect\citeauthoryear{Rasti, Torkjazi, Rejaie, Duffield, Willinger,
  and Stutzbach}{Rasti et~al\mbox{.}}{2009}]%
        {WillingerRDS}
{Amir~H. Rasti}, {Mojtaba Torkjazi}, {Reza Rejaie}, {Nick Duffield}, {Walter
  Willinger}, {and} {Daniel Stutzbach}. 2009.
\newblock \showarticletitle{Respondent-Driven Sampling for Characterizing
  Unstructured Overlays}. In {\em Proceedings of the IEEE INFOCOM 2009}.
  2701--2705.
\newblock
\showISSN{0743-166X}
\showDOI{%
\url{http://dx.doi.org/10.1109/INFCOM.2009.5062215}}


\bibitem[\protect\citeauthoryear{Ribeiro, Gauvin, Liu, and Towsley}{Ribeiro
  et~al\mbox{.}}{2010}]%
        {MySpace}
{Bruno Ribeiro}, {William Gauvin}, {Benyuan Liu}, {and} {Don Towsley}. 2010.
\newblock \showarticletitle{On MySpace Account Spans and Double Pareto-Like
  Distribution of Friends}. In {\em INFOCOM IEEE Conference on Computer
  Communications Workshops , 2010}. 1--6.
\newblock
\showDOI{%
\url{http://dx.doi.org/10.1109/INFCOMW.2010.5466698}}


\bibitem[\protect\citeauthoryear{Ribeiro and Towsley}{Ribeiro and
  Towsley}{2010}]%
        {TechReport}
{Bruno Ribeiro} {and} {Don Towsley}. 2010.
\newblock \showarticletitle{Estimating and Sampling Graphs with
  Multidimensional Random Walks}. In {\em Proceedings of the 10th ACM SIGCOMM
  Conference on Internet Measurement} {\em (IMC '10)}. ACM, New York, NY, USA,
  390--403.
\newblock
\showISBNx{978-1-4503-0483-2}
\showDOI{%
\url{http://dx.doi.org/10.1145/1879141.1879192}}


\bibitem[\protect\citeauthoryear{Ribeiro and Towsley}{Ribeiro and
  Towsley}{2012}]%
        {RibeiroCDC12}
{Bruno Ribeiro} {and} {Don Towsley}. 2012.
\newblock \showarticletitle{On the estimation accuracy of degree distributions
  from graph sampling}. In {\em 51st IEEE Conference on Decision and Control
  (CDC 2012)}. 5240--5247.
\newblock
\showISSN{0191-2216}
\showDOI{%
\url{http://dx.doi.org/10.1109/CDC.2012.6425857}}


\bibitem[\protect\citeauthoryear{Ribeiro, Wang, Murai, and Towsley}{Ribeiro
  et~al\mbox{.}}{2012}]%
        {RibeiroINFOCOM2012}
{Bruno Ribeiro}, {Pinghui Wang}, {Fabricio Murai}, {and} {Don Towsley}. 2012.
\newblock \showarticletitle{Sampling directed graphs with random walks}. In
  {\em Proceedings of IEEE INFOCOM 2012}. 1692--1700.
\newblock
\showISSN{0743-166X}
\showDOI{%
\url{http://dx.doi.org/10.1109/INFCOM.2012.6195540}}


\bibitem[\protect\citeauthoryear{Salehi and Rabiee}{Salehi and Rabiee}{2013}]%
        {DNM}
{M. Salehi} {and} {H.~R. Rabiee}. 2013.
\newblock \showarticletitle{A Measurement Framework for Directed Networks}.
\newblock {\em IEEE Journal on Selected Areas in Communications\/} {31}, 6
  (June 2013), 1007--1016.
\newblock
\showISSN{0733-8716}
\showDOI{%
\url{http://dx.doi.org/10.1109/JSAC.2013.130603}}


\bibitem[\protect\citeauthoryear{Salganik and Heckathorn}{Salganik and
  Heckathorn}{2004}]%
        {Salganik2004}
{Matthew~J. Salganik} {and} {Douglas~D. Heckathorn}. 2004.
\newblock \showarticletitle{Sampling and estimation in hidden populations using
  respondent-driven sampling}.
\newblock {\em Sociological Methodology\/}  {34} (2004), 193--239.
\newblock


\bibitem[\protect\citeauthoryear{Stutzbach, Rejaie, Duffield, Sen, and
  Willinger}{Stutzbach et~al\mbox{.}}{2009}]%
        {Stutzbach2009}
{Daniel Stutzbach}, {Rea Rejaie}, {Nick Duffield}, {Subhabrata Sen}, {and}
  {Walter Willinger}. 2009.
\newblock \showarticletitle{On unbiased sampling for unstructured peer-to-peer
  networks}.
\newblock {\em IEEE/ACM Transactions on Networking\/} {17}, 2 (April 2009),
  377--390.
\newblock


\bibitem[\protect\citeauthoryear{Volz and Heckathorn}{Volz and
  Heckathorn}{2008}]%
        {RDSprob}
{Erik Volz} {and} {Douglas~D. Heckathorn}. 2008.
\newblock \showarticletitle{Probability Based Estimation Theory for Respondent
  Driven Sampling}.
\newblock {\em Journal of Official Statistics\/} {24}, 1 (03 2008), 79.
\newblock
\showISBNx{0282423X}


\bibitem[\protect\citeauthoryear{Zhou, Zhang, Gong, and Das}{Zhou
  et~al\mbox{.}}{2016}]%
        {GautamICDE2013}
{Zhuojie Zhou}, {Nan Zhang}, {Zhiguo Gong}, {and} {Gautam Das}. 2016.
\newblock \showarticletitle{Faster Random Walks by Rewiring Online Social
  Networks On-the-Fly}.
\newblock {\em ACM Trans. Database Syst.\/} {40}, 4, Article 26 (Jan. 2016), 36
  pages.
\newblock
\showISSN{0362-5915}
\showDOI{%
\url{http://dx.doi.org/10.1145/2847526}}


\end{thebibliography}
                             % Sample .bib file with references that match those in
                             % the 'Specifications Document (V1.5)' as well containing
                             % 'legacy' bibs and bibs with 'alternate codings'.
                             % Gerry Murray - March 2012

% History dates
% \received{February 2007}{March 2009}{June 2009}

% Electronic Appendix
% \elecappendix

\medskip

%\section{This is an example of Appendix section head}
%
%Channel-switching time is measured as the time length it takes for
%motes to successfully switch from one channel to another. This
%parameter impacts the maximum network throughput, because motes
%cannot receive or send any packet during this period of time, and it
%also affects the efficiency of toggle snooping in MMSN, where motes
%need to sense through channels rapidly.
%
%By repeating experiments 100 times, we get the average
%channel-switching time of Micaz motes: 24.3 $\mu$s. We then conduct
%the same experiments with different Micaz motes, as well as
%experiments with the transmitter switching from Channel 11 to other
%channels. In both scenarios, the channel-switching time does not have
%obvious changes. (In our experiments, all values are in the range of
%23.6 $\mu$s to 24.9 $\mu$s.)
%
%\section{Appendix section head}
%
%The primary consumer of energy in WSNs is idle listening. The key to
%reduce idle listening is executing low duty-cycle on nodes. Two
%primary approaches are considered in controlling duty-cycles in the
%MAC layer.

\end{document}